\documentclass[aps,pra,onecolumn,nobibnotes,superscriptaddress,nofootinbib]{revtex4-2}
\usepackage{graphicx} 
\usepackage{caption,graphics,physics}
\usepackage{subcaption}
\usepackage{amsmath,xcolor,bm}
\usepackage{qcircuit}
\usepackage{mathtools,amsmath,amsthm,amsfonts,amssymb,amscd}
\usepackage{hyperref}
\usepackage[ruled,vlined]{algorithm2e}
\usepackage{algpseudocode}
\usepackage{multirow}
\usepackage{array}
\usepackage[capitalize]{cleveref}

\hypersetup{
    colorlinks=true,
    linkcolor=red,
    filecolor=magenta,      
    urlcolor=cyan,
}
\usepackage{geometry}

\newtheorem{theorem}{Theorem}

\newtheorem{observation}{Observation}
\newtheorem{lemma}{Lemma}

\newtheorem{example}{Example}
\newtheorem{remark}{Remark}
\newtheorem{proposition}{Proposition}

\newtheorem{definition}{Definition}

\newtheorem{corollary}{Corollary}
\newtheorem{problem}{Problem}

\newlength{\ketketwidth}
\newlength{\ketwidth}

\Crefname{appsec}{Appendix}{Appendices}

\begin{document}

\title{Particle-Number-Preserving Quantum Signal Processing for Thermal Properties: Canonical-Ensemble Estimation and Grand-Canonical Reconstruction}

\author{Taehee Ko}
\affiliation{School of Computational Sciences, Korea Institute for Advanced Study, Seoul, Republic of Korea}


\author{Sungbin Lim}
\email{Corresponding author: sungbin@korea.ac.kr}
\affiliation{Department of Statistics, Korea University, Seoul, Republic of Korea}
\affiliation{National Research Council of Science \& Technology, MSIT, Republic of Korea}

\author{Sangkook Choi}
\email{Corresponding author: sangkookchoi@kias.re.kr}
\affiliation{School of Computational Sciences, Korea Institute for Advanced Study, Seoul, Republic of Korea}

\begin{abstract}
We develop an end-to-end framework for estimating canonical thermal properties of particle-number-conserving second-quantized Hamiltonians using quantum simulations.  This framework integrates initial state preparation with mixed Dicke states, particle-number-preserving Trotter simulation, and generalized quantum signal processing to estimate canonical partition functions and observables. Building on this canonical scheme, we introduce a grand-canonical reconstruction scheme in which the same quantum measurement data can be reused for arbitrary chemical potentials through classical reweighting, thereby avoiding repeated quantum circuit executions as the chemical potential is varied. For physically relevant Hamiltonians, we prove the existence of system-size-independent temperature regimes in which the canonical and reconstructed grand-canonical schemes achieve sublinear time complexity in the dimensions of  fixed-particle-number subspaces, up to polynomial factors. Consequently, both schemes achieve lower asymptotic time complexity than classical Lanczos-based methods in these regimes. We numerically validate our theoretical results for the two-dimensional Hubbard model in both weak- and intermediate-coupling regimes. 

\end{abstract}

\maketitle

\section*{Introduction}

Particle-number conservation is a fundamental symmetry of various quantum many-body Hamiltonians, including those describing molecular systems~\cite{petras2020using,harsha2020wave,liu2018ab,li2025dissipative,shen2020finite,sokolov2021microcanonical}, ultracold atoms~\cite{kasper2017implementing,cao2017unified}, and lattice electronic systems such as twisted bilayer graphene~\cite{xie2021twisted}. When the particle number is fixed, the canonical ensemble provides the standard framework for calculating thermal properties: the partition function and observables are evaluated within the selected fixed-particle-number subspace. Classical canonical methods have been developed~\cite{shen2023stable,nocera2016symmetry,prelovvsek2013ground,barthel2016matrix}, among which Lanczos-based approaches~\cite{prelovvsek2013ground} provide rigorous complexity guarantees under general parameter regimes. Specifically, their computational cost scales linearly with the dimension of the selected sector in general.

In contrast to classical algorithms, a general quantum framework with rigorous complexity guarantees for estimating canonical thermal properties of interacting fermionic systems within fixed-particle-number sectors has not yet been established, although several related works provide rigorous guarantees in other  settings. For example, several algorithms have been developed within the Lindbladian framework~\cite{rouze2026optimal,vsmid2025polynomial,chen2023quantum,tong2025fast,li2025dissipative}, but those designed for fermionic systems are not directly applicable to estimating generic canonical thermal properties. In particular, the Lindbladian constructions in \cite{vsmid2025polynomial,tong2025fast} do not preserve particle-number symmetry. Although a particle-number-preserving Lindbladian was proposed in \cite{li2025dissipative}, its theoretical guarantees are restricted to estimating ground-state properties of Hartree–Fock models. Beyond these Lindbladian approaches, algorithms based on quantum singular value transformation (QSVT) or linear combinations of unitaries (LCU) \cite{gilyen2019quantum,tosta2024randomized,coopmans2023predicting} enable the estimation of grand-canonical thermal properties but do not directly provide algorithms for estimating canonical thermal properties. 

These advances and limitations of classical and quantum algorithms raise two natural questions:
\begin{itemize}
\item Can we construct a quantum framework with rigorous complexity guarantees for estimating canonical-ensemble properties, including partition functions and expectation values of observables, within a selected particle-number sector?

\item Can such a framework achieve a computational advantage over classical Lanczos-based methods in certain parameter regimes?
\end{itemize}


In this work, we answer both questions in the affirmative by proposing an end-to-end quantum algorithm for estimating canonical thermal properties of the second-quantized model \eqref{ex: models}. Building on this canonical scheme, we further present a simple  grand-canonical  scheme that reconstructs grand-canonical estimates from canonical-sector estimates.

For the grand-canonical ensemble, we retain the classical Lanczos method as our main classical baseline for comparison, despite recent advances by Tong and collaborators~\cite{qprk-k7wn,chen2025convergence}. In particular, \cite{qprk-k7wn} establishes temperature thresholds above which Gibbs sampling can be performed in polynomial time. However, our numerical results indicate that for the Hubbard model we considered, these thresholds lie far above physically relevant temperatures, as summarized in \cref{tab:mus_combined}. Nevertheless, these theoretical guarantees may not necessarily reflect the practical performance of the algorithm. Similarly, \cite{chen2025convergence} provides another classical algorithm with polynomial-time guarantees, but the practical relevance of the parameter regimes covered by these guarantees remains unclear from a numerical perspective, particularly regarding how rapidly the interaction-strength threshold decreases with increasing inverse temperature and whether it remains sufficiently large to accommodate physically relevant interaction strengths (e.g., $U=8$ in the Hubbard model \eqref{eq: model hubbard}).

Our technical contributions are summarized as follows:
\begin{enumerate}
    \item The proposed canonical ensemble (\cref{alg: cano})  unifies several independent quantum computing techniques (\cref{def: Dicke,Initial state} the GQSP circuit in \cref{fig: GQSP}). The required gate complexity per circuit sample scales polynomial. Our analysis guarantees an exponentially improved scaling of the number of Trotter steps compared with that obtained from a conventional analysis \cite{nielsen2010quantum,dong2022ground}. Notably, the Trotterized Hamiltonian simulation in our algorithm is defined with a normalization factor corresponding to the fixed-particle subspace not the full Hilbert space. This is consistent with the case of block encodings \cite{liu2025block}. In addition, canonical-ensemble gate and circuit-sampling complexity bounds are independent of the chemical potential $\mu$ (\cref{thm: main 1}).

    \item Based on our canonical scheme, we introduce a grand-canonical reconstruction scheme (\cref{alg: grand cano new}). The advantage of our approach is that it does not require re-execution of the quantum circuit for each trial value of $\mu$. Therefore, it is particularly appealing for practitioners who need to determine an appropriate chemical potential through an extensive scanning procedure, as discussed in \cite{shen2020finite}. Moreover, quantum circuit and circuit-sampling complexities for the grand-canonical scheme are independent of $\mu$, rather being comparable to those of the canonical scheme corresponding to a fixed-particle subspace (Corollary \ref{cor: new grand scheme}).

    \item For physically relevant Hamiltonians whose spectral norm scales linearly with system size, we prove the existence of system-size-independent temperature thresholds above which our canonical and grand-canonical algorithms have sublinear   time complexity in the dimension of fixed-particle subspaces, up to polynomial factors (Observation \ref{obs 2}). We numerically validate this scaling comparison for a two-dimensional Hubbard model. 
\end{enumerate}

\section*{Results}

\begin{center}
    \textbf{Related work}
\end{center}

 \paragraph*{\textbf{Canonical scheme}}In classical algorithms, Lanczos methods combined with stochastic trace estimation and their variants \cite{chen2023krylov,epperly2024xtrace,prelovvsek2013ground} provide rigorous computational complexity guarantees under general parameter regimes. In particular, their computational costs scale linearly with the dimension of a fixed-particle-number subspace, without restrictions on coupling strength or temperature. In contrast, the complexity of our quantum algorithm (\cref{alg: cano}) depends on the inverse temperature, and its time complexity scales sublinearly with the dimension of the fixed-particle-number subspace at constant temperatures for physically relevant Hamiltonians, up to polynomial factors (Observation \ref{obs 2}). For quantum algorithms, to the best of our knowledge, rigorous complexity guarantees for estimating canonical thermal properties has not yet been established, although only a few recent works have made relevant advances~\cite{liu2025block,goh2026direct}. Explicit circuit constructions of block encodings restricted to fixed-particle-number subspaces \cite{liu2025block} could be combined with QSP techniques and our proposed initial state (\cref{def: Dicke}) to estimate canonical thermal properties, although this possibility was not explored in that work. With Hamiltonian simulation as input model, an LCU-based approach \cite{goh2026direct} also employs a Dicke state as its initial state. However, the Trotter circuits employed in that approach do not preserve particle-number symmetry, and thus the resulting estimates may contain contributions from outside the selected particle-number sector. In contrast, our approach constructs particle-number-preserving Trotter circuits (\cref{Trotter circuits}), ensuring that Trotterization errors remain confined to the selected particle-number sector and enabling sector-restricted error analysis (e.g., the relative-error analysis for partition function estimation in \eqref{eq: Z analysis}). Furthermore, the approach in \cite{goh2026direct} does not directly provide a procedure for estimating generic observables, such as correlation functions \cite{king2025triply,ko2025fermionic}, whereas our approach provides such a procedure with rigorous guarantees as shown in \cref{thm: main 1}.

\paragraph*{\textbf{Grand-canonical scheme}}  
For classical algorithms, to our knowledge, only the aforementioned Lanczos-based methods and the randomized sampling algorithms  \cite{chen2025convergence,qprk-k7wn} provide rigorous computational complexity guarantees. Quantum algorithms based on LCU, QSP, or Lindbladian frameworks typically assume access to a block encoding of the Hamiltonian \cite{gilyen2019quantum,tosta2024randomized,coopmans2023predicting,de2024partition,vsmid2025polynomial,tong2025fast} and are applicable to the grand-canonical ensemble. Other approaches employ different input models, including walk operators \cite{chowdhury2021computing}, Pauli sampling \cite{wang2024qubit}, and Trotter formulas \cite{chakraborty2025quantum}. In these approaches, both circuit complexity and circuit-sampling complexity can depend on the chemical potential. In contrast, our approach (\cref{alg: grand cano new}) achieves chemical-potential-independent circuit complexity and circuit-sampling complexity by reconstructing grand-canonical properties from canonical ones (\cref{cor: new grand scheme}).

\begin{center}
    \textbf{Main result}
\end{center}

We consider the second-quantized Hamiltonians for a system of $L$ spin orbitals (or spatial modes),
\begin{equation}\label{ex: models}
    \begin{split}
    H
    =
    \sum_{i,j}
    h_{i,j}
    c_{i}^\dagger c_{j}
    +
    \frac{1}{2}
    \sum_{i,j,u,v}
    g_{i,j,u,v}
    c_{i}^\dagger c_{j}^\dagger
    c_{u} c_{v},    
    \end{split}
\end{equation}
where $h_{ji}=h_{ij}^{*}$ and $g_{vuji}=g_{ijuv}^{*}$. For a fixed particle-number sector $s$, the canonical density operator is
\begin{equation}
    \rho_{L,s}(\beta)
    =
    \frac{e^{-\beta H(L,s)}}{\tr\!\left(e^{-\beta H(L,s)}\right)}.
    \label{eq:canonical-density-operator}
\end{equation}
The two target quantities are the thermal expectation value and the canonical partition function,
\begin{align}
    \langle O\rangle_{\beta,s}
    &:= \tr\!\left(\rho_{L,s}(\beta)O\right),
    \label{eq:canonical-thermal-expectation}\\
    Z_{L,s}(\beta)
    &:= \tr\!\left(e^{-\beta H(L,s)}\right).
    \label{eq:canonical-partition-function}
\end{align} 
Here, $H(L,s)$ denotes the restriction of the full many-body Hamiltonian to the $s$-particle sector. The sector label $s$ takes one of two forms depending on the symmetries of the target system: $s=k$ denotes the total particle-number sector, while $s=(k,k')$ denotes the spin-resolved sector with $k$ spin-up and $k'$ spin-down particles when the corresponding number operators independently commute with the Hamiltonian. As is often the case in practice, we focus on observables $O$ that can be expressed as linear combinations of a polynomial number of Pauli operators, that is,
\begin{align}\label{eq: observable}
    O=\sum_{a=1}^{\mathcal{O}(\mathrm{poly}(L))}c_aP_a, \quad c_a=\mathcal{O}(1),\quad P_a:\text{Pauli operator}.
\end{align}We summarize the main notation and parameters used throughout this work in \cref{tab:notation}. We define the problem as the following. 
\begin{problem}\label{def: problem}
Given a Hamiltonian $H$ in \eqref{ex: models} and an inverse temperature
$\beta \geq 0$, let $\epsilon\in(0,1)$ be a target precision, and let
$\eta\in(0,1)$ be a failure probability. For a fixed particle number sector $s$, let $H(L,s)$ denote the restriction of $H$ to the subspace corresponding to the $s$-particle sector. Let $O$ be 
an observable with $\norm{O}_2\le 1$. Define
    $Z_{L,s}(\beta)
    :=
    \tr\!\left(
        e^{-\beta H(L,s)}
    \right)$. The goal is to obtain $\widehat{Z}$ and $\widehat{o}$ such that
\begin{align}
    \abs{\widehat{Z} - Z_{L,s}(\beta)}
    &\leq
    \epsilon Z_{L,s}(\beta), \\
    \abs{\widehat{o} -\langle O\rangle_{\beta,s}}
    &\leq
    \epsilon,
\end{align}
with probability at least $1-\eta$.
\end{problem}This choice of error metrics follows the standard conventions for estimating thermal properties \cite{bravyi2022quantum,wang2024qubit}. Our main result is summarized as follows,
\begin{theorem}\label{thm: main 1}
Assume a normalization factor $\lambda_{L,s}>0$ such that $\norm{\frac{H(L,s)}{\lambda_{L,s}}}_2 \le 1-\delta$ for some $\delta \in (0,1)$, where $\norm{\cdot}_2$ denotes the spectral norm of a matrix. Using the first-order Trotter formula  with a GQSP circuit in \cref{fig: GQSP}, Problem \ref{def: problem} can be solved with the following:
\begin{itemize}
\item $2L$ (for the $s=k$ sector) or $4L$ (for the $s=(k,k')$ sector) qubits for the initial state preparation 
\item $L$ (for the $s=k$ sector) or $2L$ (for the $s=(k,k')$ sector) system qubits, and a single ancilla qubit for the GQSP
\item $r = \Omega\left(\Lambda_{H(L,s)}^2\max\left[\frac{1}{\lambda_{L,s}\delta},\; \frac{\beta }{\log(1+\epsilon)}\right]\right)$ Trotter steps per Hamiltonian simulation
\item $\mathcal{O}\left(r L^4\left(\ln\left(\frac{1}{\epsilon}\right)+\beta\lambda_{L,s} \right)\frac{1}{\delta}\right)$ gate complexity
\item $\mathcal{O}\left(C\frac{d_s e^{\beta\lambda_{L,s}}}{Z_{L,s}(\beta)\epsilon^2}\log\frac{1}{\eta}\right)$ circuit sampling complexity, where
\begin{align*}
    C = \begin{cases}
        1 & \text{for estimating \(Z_{L,s}(\beta)\)} \\
        \left( \sum_{a=1}^{\mathcal{O}(\mathrm{poly}(L))} \abs{c_a} \right)^2 & \text{for estimating \(\langle O\rangle_{\beta,s}\)}
    \end{cases}.
\end{align*}
\end{itemize}Here, the factor $\Lambda_{H(L,s)}$ is defined in \eqref{eq: lambda_H_s}
\end{theorem}
The proof of the theorem is given in the Methods section.   In \cref{thm: main 1}, we assume a normalization factor $\lambda_{L,s}$ and we can find a conservative factor by using the Gershigorin's theorem. For a tighter factor, one may utilize  semi-definite programming  relaxations \cite{baumgratz2012lower,mazziotti2011large}.

\begin{center}
    \textbf{Algorithm overview}
\end{center}

We provide an overview of our numerical schemes for canonical and grand-canonical thermal property estimations below. We fix a particle-number sector $s$ and let $\Pi_s$ be its projector. We define $d_s=\tr(\Pi_s)$ be the dimension of fixed-particle subspace.  For $s=k$, the known sector dimension is $d_s=\binom{L}{k}$; for a spin-resolved sector $s=(k,k')$, it is $d_s=\binom{L}{k}\binom{L}{k'}$.

\begin{center}
    \textit{Canonical scheme}
\end{center}

\paragraph{Prepare the mixed Dicke state as the initial state}
We initialize the system in the mixed Dicke state $\omega_{L,s}=\Pi_s/d_s$ defined in \cref{def: Dicke}. This state is the maximally mixed state corresponding to the sector $s$. To encode this state into a quantum circuit, we prepare a purification of $\omega_{L,s}$ as detailed in \cref{Initial state}.

\paragraph{Construct a block encoding of the exponential matrix function using a Trotterized GQSP}
With the initial state $\omega_{L,s}$, we construct a circuit using the GQSP \cite[Theorem 6]{motlagh2024generalized}, which approximately block-encode $e^{-\frac{\beta H}{2}}$ up to a subnormalization factor. The circuit involves a single ancilla qubit with Hamiltonian-simulation input model, as depicted in \cref{fig: GQSP}. We choose a normalization factor, $\lambda_{L,s}$, such that $\norm{\frac{H(L,s)}{\lambda_{L,s}}}_2<1$. This normalization factor is introduced when compiling the Hamiltonian simulation using a Trotter circuit, which is reasonable due to the choice of our initial state (\cref{lem: state GQSP}). To block-encode $e^{-\frac{\beta H}{2}}$, we find a subnormalization factor by finding a corresponding Fourier-series approximation (\cref{lem: fourier approx exp}) and evaluating the absolute sum of Fourier coefficients as in \cref{lem: state GQSP}.  Then, the canonical thermal properties in Problem \ref{def: problem} are formulated as the probability of observing state $\ket{0}$ in the ancilla qubit and the usual observable measurement conditioned on that event as follows,
\begin{align}
    P(0) &= \tr\big(f(e^{iH})^\dag \omega_{L,s} f(e^{iH})\big) \approx \frac{ e^{-\beta\lambda_{L,s}} \tr\big(e^{-\beta H(L,s)}\big)}{d_s}, \label{eq: true p0} \\
    \langle O \rangle &= \frac{\tr\big(f(e^{iH})^\dag \omega_{L,s} f(e^{iH}) O\big)}{P(0)} \approx \frac{\tr\big(e^{-\beta H(L,s)}O\big)}{\tr\big(e^{-\beta H(L,s)}\big)} \label{eq: true O}.
\end{align}Meanwhile, standard Trotter circuits generally do not preserve the number symmetry and the above estimates may involve \emph{biases} from outside the target sector $s$. To address this, we construct number-symmetry-preserving Trotter circuits by extending the unitary coupled cluster ansatz used in variational quantum algorithms \cite{romero2019strategies} to both symmetric and antisymmetric forms. We specify the corresponding Trotter-circuit blocks (\cref{Trotter circuits}). Notably, any finite product of these blocks also commutes with the number operator $N$. Equivalently, the finite Trotter product can be viewed as an exact evolution under an effective Hamiltonian $\widetilde{H}$. As proved in \cref{lem: effective Ham}, $[\widetilde{H},N]=0$, so the Trotter error is controlled within the target sector $s$.

\paragraph{Estimate the partition function and observables}To perform the task in Problem \ref{def: problem}, we follow standard approaches based on the Monte Carlo sampling as outlined in \cref{alg: cano}. Here we briefly describe the procedure. 

We execute the circuit $N_{\mathrm{shots}}$ times. Let $A_j=1$ if shot $j$ returns the GQSP ancilla in $\ket{0}$ and let $A_j=0$ otherwise, and let $N_0=\sum_{j=1}^{N_{\mathrm{shots}}}A_j$ be the number of successful shots. We consider the sample mean estimator, $\widehat{P}_s(0)=N_0/N_{\mathrm{shots}}$. This unbiased estimate approximates the partition function,
\begin{equation}
Z_{L,s}(\beta)\approx d_se^{\beta\lambda_{L,s}}\widehat{P}_s(0)=d_se^{\beta\lambda_{L,s}}\frac{N_0}{N_{\mathrm{shots}}}.
\label{eq:partition-estimator-overview}
\end{equation}
For simplicity, we consider the case that $O$ is a Pauli operator. To estimate $\langle O\rangle_{\beta,s}$, we measure both the ancilla qubit and the system register. Then, we perform post-processing on measurements such that for the measurements where $\ket{0}$ is observed in the ancilla qubit, we collect the estimates of observable   $O$. Let $M_0$ be the number of such occurrences and $o_j$ denote the estimated observable in shot $j$. Then, we obtain a good approximation of the observable from the sample mean,
\begin{equation}
    \langle O\rangle_{\beta,s}\approx \frac{1}{M_0}\sum_{j=1}^{M_0}o_j.
\end{equation}
If $O$ is a sum of Pauli operators, we repeat this procedure for the Pauli operators respectively.

\paragraph{Scaling of the number of Trotter steps}
The overall estimation error is governed by the approximation relation $f(e^{i\widetilde{H}}) \approx e^{-\frac{\beta}{2}\widetilde{H}} \approx e^{-\frac{\beta}{2}H}$. The first approximation is determined by the degree of the Fourier-series approximation. The second approximation error is controlled by the number of Trotter steps; for the first-order formula with $r$ steps and simulation time $\tau$, we have $\norm{\widetilde{H}-H}_2 \le \Lambda_H^2\frac{\tau}{r}$ (\cref{lem: p_order_generator}). By using Weyl's inequality, achieving precision $\epsilon$  requires a polynomial number of Trotter steps:
\begin{align}\label{eq: our trotter step}
    r = \Omega\left(
    \Lambda_H^2\max\left[\frac{1}{\lambda_{L,s}\delta}, \frac{\beta}{\log(1+\epsilon)}\right]\right).
\end{align}
However, if one applies a standard approach \cite[Appendix C]{dong2022ground} to the problem considered, an exponential scaling of the number of Trotter steps is required. As this approach is based on the repetitive use of triangle inequalities \cite[eq (4.63)]{nielsen2010quantum}, it requires a precision-dependent number of Trotter steps. Thus, $\mathcal{O}((\epsilon P(0))^{-\frac{1}{p}})$ Trotter steps are needed due to the relative-error convention in Problem \ref{def: problem}, and the probability $P(0)$ scales exponentially with $\beta$ and the subspace spectral norm estimate $\lambda_{L,s}$. In contrast, the number of Trotter steps obtained from our approach in \eqref{eq: our trotter step} is independent of $P(0)$, as detailed in \eqref{eq: Z analysis}.

\begin{center}
    \textit{Grand-canonical reconstruction}
\end{center}

Building on the canonical scheme (\cref{alg: cano}), we develop a grand-canonical scheme (\cref{alg: grand cano new}). Let $\mathcal{S}$ be the set of particle sectors and let $n_s$ be the particle number in sector $s$, namely $n_s=k$ for $s=k$ and $n_s=k+k'$ for $s=(k,k')$.  We first estimate the canonical thermal properties $Z_{L,s}(\beta)$ and $\tr(e^{-\beta H(L,s)}O)$ in every relevant sector, which are obtained via \cref{alg: cano} from measurement data. Then, we exploit the following classical-reweighting representations of the grand partition function and grand-canonical observable:
\begin{equation}
\mathcal{Z}_L(\beta,\mu)=\sum_{s\in\mathcal{S}}e^{\beta\mu n_s}Z_{L,s}(\beta),
\end{equation}
and
\begin{equation}
\begin{split}
\langle O\rangle_{\beta,\mu}^{\mathrm{GC}}
&=\frac{\sum_{s\in\mathcal{S}}e^{\beta\mu n_s}Z_{L,s}(\beta)\langle O\rangle_{\beta,s}}{\sum_{s\in\mathcal{S}}e^{\beta\mu n_s}Z_{L,s}(\beta)}.
\end{split}
\end{equation}
Therefore, for any trial value of $\mu$, we are able to estimate these grand-canonical properties by reusing the canonical thermal properties $Z_{L,s}(\beta)$ and $\tr(e^{-\beta H(L,s)}O)$ obtained from the canonical scheme in \cref{alg: cano}. This approach contrasts with previous quantum algorithms, including Lindbladian algorithms \cite{vsmid2025polynomial,tong2025fast}, QSVT, and LCU approaches \cite{gilyen2019quantum,tosta2024randomized,coopmans2023predicting}, which require changing the input $H-\mu N$ depending on $\mu$ and re-executing a corresponding quantum circuit to obtain new measurement data.

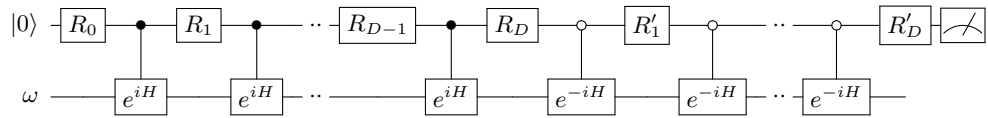
\begin{figure}[htbp]
    \centering
\makebox[\textwidth][c]{
\Qcircuit @C=0.4em @R=1.4em {
    & \lstick{\ket{0}} & \gate{R_0} & \ctrl{1} & \gate{R_1} & \ctrl{1} & \qw & \push{\;\cdot\cdot\;} \qw & \gate{R_{D-1}} & \ctrl{1} & \gate{R_D} & \ctrlo{1} & \gate{R_1'} & \ctrlo{1} & \qw & \push{\;\cdot\cdot\;} \qw & \ctrlo{1} & \gate{R_{D}'} & \meter \\
    & \lstick{\omega} & \qw & \gate{e^{iH}} & \qw & \gate{e^{iH}} & \qw & \push{\;\cdot\cdot\;} \qw & \qw & \gate{e^{iH}} & \qw & \gate{e^{-iH}} & \qw & \gate{e^{-iH}} & \qw & \push{\;\cdot\cdot\;} \qw & \gate{e^{-iH}} & \qw
    }
    }
    \caption{ 
    GQSP circuit for canonical thermal estimation \cite{motlagh2024generalized}, formed by single-qubit rotations $R_k$ and controlled Hamiltonian simulations, denoted as $e^{\pm iH}$. $D$ denotes the degree of the Laurent polynomial. The rotation angles correspond to the circuit that proportionally block encodes $e^{-\frac{\beta}{2}H}$. The system register is initialized in the mixed Dicke state $\omega_{L,s}$ (\cref{def: Dicke}). The controlled time-evolution blocks are implemented  using the number-symmetry-preserving Trotter circuits (\cref{Trotter circuits}). The probability of observing the ancilla qubit in the $\ket{0}$ state is proportional to the partition function. Conditioned on this event, the target-register measurement yields a sample for the thermal observable.}\label{fig: GQSP}
\end{figure}

\begin{center}
    \textbf{Circuit components for the canonical scheme}
\end{center}

In the present work, we consider the JW transformation of the model \eqref{ex: models}, where we let
\begin{align}
    c_j
    =
    \frac{1}{2}
    \left(
    \prod_{\ell=1}^{j-1}Z_\ell
    \right)
    (X_j+iY_j),\quad
    c_j^\dagger
    =
    \frac{1}{2}
    \left(
    \prod_{\ell=1}^{j-1}Z_\ell
    \right)
    (X_j-iY_j).
    \label{eq:JW-transformation}
\end{align}Under the JW transformation, the computational basis of
the qubit register is identified with the fermionic occupation-number
basis. In particular,
\begin{align}
    \ket{x_1x_2\cdots x_L}
    \longleftrightarrow
    \left(c_1^\dagger\right)^{x_1}
    \left(c_2^\dagger\right)^{x_2}
    \cdots
    \left(c_L^\dagger\right)^{x_L}
    \ket{\mathrm{vac}}, \quad \ket{\mathrm{vac}}=\ket{0}^{\otimes L},
\end{align}
up to a possible sign depending on the convention for the ordering of the creation operators. Here
\(x_i\in\{0,1\}\) denotes the occupation of the \(i\)-th fermionic
mode. Consequently, the Hamming weight
\begin{align}
    \mathrm{wt}(x)
    :=
    \sum_{i=1}^{L}x_i
\end{align}
is equal to the total particle number of the corresponding Fock-basis
state. We denote  the particle number operator by $n_j:=c_j^\dagger c_j$. For $i<j$, define the JW string by
    $Z_{i:j}:=\prod_{\ell=i+1}^{j-1}Z_\ell$,
where an empty product is understood to be the identity.

We define mixed Dicke states as initial states. These states are normalized projectors onto
fixed-particle subspaces.
\begin{definition}\label{def: Dicke}
For a system modeled with \(L\) total fermionic modes (e.g., spin-orbitals), we define the mixed Dicke state for a total particle number \(k\) by
\begin{align}
    \omega_{L,k}
    &:=
    \frac{1}{\binom{L}{k}}
    \sum_{\substack{x\in\{0,1\}^{L}\\
                    \mathrm{wt}(x)=k}}
    \ket{x}\bra{x}.
\end{align}
For systems of spin-$\frac{1}{2}$ fermions mapped over \(L\) spatial modes (e.g., spatial orbitals or lattice sites), where the numbers of spin-up and spin-down particles are independently conserved, we define the mixed Dicke state containing \(k\) spin-up particles and \(k'\) spin-down particles by
\begin{align}
    \omega_{L,(k,k')}
    &:=
    \frac{1}{\binom{L}{k}\binom{L}{k'}}
    \sum_{\substack{x,y\in\{0,1\}^{L}\\
                    \mathrm{wt}(x)=k,\,
                    \mathrm{wt}(y)=k'}}
    \ket{x,y}\bra{x,y}.
\end{align}
\end{definition}
In other words, \(\omega_{L,k}\) is the maximally mixed state in the
\(k\)-particle sector, whereas \(\omega_{L,(k,k')}\) is the maximally
mixed state in the sector containing \(k\) spin-up particles and
\(k'\) spin-down particles.  These states can be prepared efficiently by first
constructing a pure Dicke state with ancilla qubits and then tracing out the
ancilla register. Preparing the pure Dicke
state requires only a linear circuit
depth \cite{bartschi2019deterministic}.
\begin{proposition}[Initial state preparation]\label{Initial state}
    Preparing $\omega_{L,k}$ requires $2L$ qubits, $\mathcal{O}(kL)$ gates, and $\mathcal{O}(L)$ circuit depth. For $\omega_{L,(k,k')}$, $4L$ qubits and the same scaling of gate count and circuit depth are required.
\end{proposition}
\begin{proof}
    We prepare the purified state:
    \begin{align}
        \binom{L}{k}^{-\frac{1}{2}} \sum_{\substack{x \in \{0,1\}^L \\ \mathrm{wt}(x)=k}} \ket{x}\ket{x},
    \end{align}
    which can be achieved by preparing a pure Dicke state \cite{bartschi2019deterministic} in the first $L$-qubit register and applying CNOT gates controlled on the first register and targeted to a second $L$-qubit register. Tracing out the second register yields the desired mixed Dicke state. The overall circuit complexity is similar to the complexity established in \cite{bartschi2019deterministic}. A similar scaling of circuit complexity is required for preparing the state $\omega_{L,(k,k')}$. 
\end{proof}

We turn to Trotter circuit constructions to replace the Hamiltonian simulation block in the circuit in \cref{fig: GQSP}. As mentioned earlier, we consider circuits that preserve
the particle-number symmetry of the second-quantized model
in \eqref{ex: models}. We summarize the relevant cases below. 
\begin{example}[Number-symmetry-preserving Trotter circuit construction]\label{Trotter circuits}
    \item[\textnormal{(i)}] \textbf{$1$-body terms.}
        
        For $i\neq j$, the JW transformed $1$-body term is $h_{ij}c_i^\dagger c_j + h_{ij}^{*}c_j^\dagger c_i = \frac{1}{2}Z_{i:j} \big[ \Re(h_{ij})P_{ij}^{(\mathrm{R})} - \Im(h_{ij})P_{ij}^{(\mathrm{I})} \big]$, where $P_{ij}^{(\mathrm{R})} := X_iX_j+Y_iY_j$ and $P_{ij}^{(\mathrm{I})} := X_iY_j-Y_iX_j$. Because the Pauli strings within each $P_{ij}$ mutually commute and the grouped generator preserves particle number, they are exponentiated together to form the Trotter blocks,
        \begin{align}
            U_{\mathrm{R}}^{(ij)}(\tau) := \exp\left[ -i\tau\frac{\Re(h_{ij})}{2} Z_{i:j}P_{ij}^{(\mathrm{R})} \right], \quad 
            U_{\mathrm{I}}^{(ij)}(\tau) := \exp\left[ i\tau\frac{\Im(h_{ij})}{2} Z_{i:j}P_{ij}^{(\mathrm{I})} \right].
        \end{align}

        For the diagonal case ($i=j$), $h_{ii}\in\mathbb{R}$ and the term reduces to $h_{ii}n_i = \frac{h_{ii}}{2}(I-Z_i)$. Up to a global phase, its Trotter block is
        \begin{align}
            U^{(ii)}(\tau) = \exp\left( i\tau\frac{h_{ii}}{2}Z_i \right).
        \end{align}

    \item[\textnormal{(ii)}] \textbf{$2$-body terms.}

    The JW representation of $H_{pquv}^{(2)} := g_{pquv}c_p^\dagger c_q^\dagger c_u c_v + \text{h.c.}$ depends on the index equality pattern. We consider the three non-zero cases as follows.
    \begin{itemize}
        \item[\textnormal{(a)}] \textit{Four pairwise distinct indices ($v<u<q<p$).}

        Let $Z_{vuqp} := \big(\prod_{a=v+1}^{u-1}Z_a\big)\big(\prod_{b=q+1}^{p-1}Z_b\big)$. The term becomes $H_{pquv}^{(2)} = \frac{1}{8}Z_{vuqp} \big[ \Re(g_{pquv})P_{vuqp}^{(\mathrm{R})} + \Im(g_{pquv})P_{vuqp}^{(\mathrm{I})} \big]$, where the 8-term Pauli sums are grouped as:
        \begin{align}
            P_{vuqp}^{(\mathrm{R})}
            :={}&
            -X_vX_uX_qX_p
            +
            X_vX_uY_qY_p
            -
            X_vY_uX_qY_p
            -
            X_vY_uY_qX_p
            \nonumber\\
            &-
            Y_vX_uX_qY_p
            -
            Y_vX_uY_qX_p
            +
            Y_vY_uX_qX_p
            -
            Y_vY_uY_qY_p,\\ 
            P_{vuqp}^{(\mathrm{I})}
            :={}&
            -
            X_vX_uY_qX_p
            -
            Y_vX_uY_qY_p
            -
            X_vY_uY_qY_p
            -
            X_vX_uX_qY_p
            \nonumber\\
            &+
            Y_vX_uX_qX_p
            +
            X_vY_uX_qX_p
            +
            Y_vY_uY_qX_p
            +
            Y_vY_uX_qY_p.
        \end{align}
        As shown in \cref{lem: effective Ham}, the strings within each operator mutually commute. The corresponding grouped number-preserving blocks are
        \begin{align}
            U_{\mathrm{R}}^{(pquv)}(\tau) &:= \exp\left[ -i\tau\frac{\Re(g_{pquv})}{8} Z_{vuqp}P_{vuqp}^{(\mathrm{R})} \right], \;
            U_{\mathrm{I}}^{(pquv)}(\tau) := \exp\left[ -i\tau\frac{\Im(g_{pquv})}{8} Z_{vuqp}P_{vuqp}^{(\mathrm{I})} \right].
        \end{align}

        \item[\textnormal{(b)}] \textit{One coincidence ($p=u$, with $q<v$).} 
        
        The term simplifies via $c_p^\dagger c_q^\dagger c_p c_v = -n_p c_q^\dagger c_v$ to $H_{pqpv}^{(2)} = -\frac{1}{4}(I-Z_p)Z_{q:v} \big[ \Re(g_{pqpv})P_{qv}^{(\mathrm{R})} + \Im(g_{pqpv})P_{qv}^{(\mathrm{I})} \big]$. The Trotter blocks are
        \begin{align}
            U_{\mathrm{R}}^{(pqpv)}(\tau) &:= \exp\left[ i\tau\frac{\Re(g_{pqpv})}{4} (I-Z_p)Z_{q:v}P_{qv}^{(\mathrm{R})} \right], \;
            U_{\mathrm{I}}^{(pqpv)}(\tau) := \exp\left[ i\tau\frac{\Im(g_{pqpv})}{4} (I-Z_p)Z_{q:v}P_{qv}^{(\mathrm{I})} \right].
        \end{align}

        \item[\textnormal{(c)}] \textit{Two coincidences ($p=u, q=v, p\neq q$).}

        The term reduces to $c_p^\dagger c_q^\dagger c_p c_q = -n_pn_q$, giving $H_{pqpq}^{(2)} = -2\Re(g_{pqpq})n_p n_q$. Since $n_pn_q = \frac{1}{4}(I-Z_p)(I-Z_q)$, the Trotter block is
        \begin{align}
            U^{(pq)}(\tau) &:= \exp\left[ i\tau\frac{\Re(g_{pqpq})}{2} (I-Z_p)(I-Z_q) \right].
        \end{align}

    \end{itemize}
\end{example}
We remark that a Trotter circuit for $e^{iH\tau}$ constructed using any ordering of these blocks is applicable in \cref{fig: GQSP}, as each individual block preserves the total particle number.

\begin{center}
    \textbf{Numerical schemes for canonical and grand-canonical ensemble properties}
\end{center}

The circuit constructions described above, combined with the GQSP in \cref{fig: GQSP}, enable the estimation of canonical ensemble properties as outlined in \cref{alg: cano}. 
To estimate the canonical partition function, we use the extension of the relation \eqref{eq: true p0} with an associated effective Hamiltonian $\widetilde{H}$,
\begin{align}   P(0)&=\tr(f(e^{i\widetilde{H}})^\dag \omega_{L,s} f(e^{i\widetilde{H}}))\approx \frac{1}{d_s}\frac{\tr(e^{-\beta \widetilde{H}(L,s)})}{e^{\beta\lambda_{L,s}}}.
\end{align}We note that the Hamiltonian simulation is projected onto the target particle subspace, since $[\widetilde{H},N]=0$ due to the construction in \cref{Trotter circuits} and $\omega_{L,s}$ is the initial state chosen as a mixed Dicke state (\cref{def: Dicke}). We notice that the normalization factor chosen for the implementation of GQSP is an upper bound of the spectral norm of the sub-Hamiltonian corresponding to the particle subspace, $\lambda_{L,s}$, which is probably smaller than in the case of the full  Hilbert space. 

When $\ket{0}$ is measured in the ancilla qubit, we estimate an observable corresponding to the target sector by noticing the following relation similar to \eqref{eq: true O},
\begin{align}
\tr(\frac{f(e^{i\widetilde{H}})^\dag \omega_{L,(k,k')} f(e^{i\widetilde{H}})}{\tr(f(e^{i\widetilde{H}})^\dag \omega_{L,(k,k')} f(e^{i\widetilde{H}}))}O)\approx \tr(\frac{e^{-\beta \widetilde{H}(L,(k,k'))}}{\tr(e^{-\beta \widetilde{H}(L,(k,k'))})}O).
\end{align}

\begin{algorithm}[H]
\SetAlgoLined
    \KwData{System size $L$, sector $s \in \{k, (k,k')\}$, inverse temperature $\beta$, precision $\epsilon>0$, normalization factor $\lambda_{L,s}$ (where $\norm{H(L,s)/\lambda_{L,S}}\leq 1-\delta$), the number of shots $N_{\text{shots}}$ (circuit sampling complexity in \cref{thm: main 1}), and observable $O$}
    \KwResult{Approximations of the partition function $\widehat{Z}_{L,s}(\beta)$ and the thermal expectation value $\widehat{o}_{L,s}$ of the observable $O$}

     1. Construct the GQSP circuit (\cref{fig: GQSP,Initial state,Trotter circuits}).

     2. \textbf{Partition Function Estimate:} Execute the circuit  $N_{\textrm{shots}}$ times. Let $N_0$ be the number of ancilla $\ket{0}$ outcomes. Estimate $\widehat{Z}_{L,s}(\beta)= d_s e^{\beta\lambda_{L,s}} \frac{N_0}{N_{\textrm{shots}}}$, where $d_s$ is the sector dimension ($\binom{L}{k}$ or $\binom{L}{k}\binom{L}{k'}$).

     3. \textbf{Observable Estimate:}  Conditioned on the $M_0$ occurrences among $N_{\textrm{shots}}$, where the ancilla is $\ket{0}$, measure $O$ on the system register to obtain outcomes $\{o_j\}_{j=1}^{M_0}$. Estimate $\widehat{o}_{L,s} = \frac{1}{M_0} \sum_{j=1}^{M_0}o_j$.
     
    \caption{Quantum estimation of canonical ensemble properties}
    \label{alg: cano}
\end{algorithm}
We extend \cref{alg: cano} to the grand canonical ensemble as in \cref{alg: grand cano new}. The key idea is to apply \cref{alg: cano} as a subroutine across all sectors with a fixed number of particles, aggregate the estimates of canonical ensemble properties, and use  those estimates to evaluate grand canonical ensemble properties via the following identities,
\begin{align}\label{eq: identities}
&\tr(e^{-\beta(H-\mu N)}O)
=\sum_{k,k'}e^{\beta\mu(k+k')}\tr(e^{-\beta H(L,(k,k'))})\frac{\tr(e^{-\beta H(L,(k,k'))}O)}{\tr(e^{-\beta H(L,(k,k'))})},\nonumber\\
&\tr(e^{-\beta(H-\mu N)})
=\sum_{k,k'}e^{\beta\mu(k+k')}\tr(e^{-\beta H(L,(k,k'))}).
\end{align}
Notice that, except for the scalar weight $e^{\beta\mu(k+k')}$, all quantities in the identities are defined by canonical ensemble properties, which correspond to the outputs of \cref{alg: cano}. Thus, once the outputs are obtained from quantum computations, classical postprocessing based on \eqref{eq: identities} is performed to estimate grand canonical properties.

\cref{alg: grand cano new} is guaranteed to achieve a similar goal as in Problem \ref{def: problem} for the grand canonical ensemble, provided that its subroutine, \cref{alg: cano}, does so as well (e.g. relative error for the partition function estimation and additive error for the thermal average estimation). Formally, we state this observation as follows.  
\begin{corollary}\label{cor: new grand scheme}
Suppose that the canonical partition functions are estimated to relative error $\epsilon$ and the corresponding canonical thermal averages are estimated to additive error $\epsilon$, uniformly over all particle-number sectors in \cref{alg: grand cano new}. Then, for any $\epsilon\in(0,1)$, the resulting grand-canonical partition function is estimated to relative error $\epsilon$, and the grand-canonical thermal average is estimated to additive error
$\mathcal{O}\left(\epsilon\right)$. Subsequently, the gate complexity and circuit sampling complexity of \cref{alg: grand cano new} are determined by the  corresponding maximums of them in \cref{thm: main 1} over all particle sectors $s$.

\end{corollary}

\begin{proof}
For $s \in \mathcal{S}$, we let $n_s$ denote the particle number for sector $s$, and let
\begin{align}
Z_{L,s}(\beta):=
\tr\left(e^{-\beta H(L,s)}\right),\;
o_{L,s}:=
\frac{\tr\left(e^{-\beta H(L,s)}O\right)}{Z_{L,s}(\beta)},\;
w_{L,s}:=
e^{\beta\mu n_s}Z_{L,s}(\beta).
\end{align}
By \eqref{eq: identities}, the exact grand-canonical thermal average and partition function are represented by
\begin{align}
o_L^{\mathrm{GC}}
=
\frac{\sum_{s \in \mathcal{S}} w_{L,s} o_{L,s}}{\sum_{s \in \mathcal{S}} w_{L,s}},\quad \mathcal{Z}_L(\beta,\mu)=\sum_{s \in \mathcal{S}} w_{L,s}.
\end{align}
Let $\widehat w_{L,s}:= e^{\beta\mu n_s}\widehat Z_{L,s}(\beta)$ and
$\widehat{\mathcal{Z}}_L(\beta,\mu):=\sum_{s \in \mathcal{S}}\widehat w_{L,s}$.
By the assumption, $\abs{\widehat Z_{L,s}(\beta) - Z_{L,s}(\beta)} \le \epsilon Z_{L,s}(\beta)$, which implies $\abs{w_{L,s}-\widehat w_{L,s}}\le \epsilon w_{L,s}$. It follows  $\abs{\widehat{\mathcal{Z}}_L(\beta,\mu) - \mathcal{Z}_L(\beta,\mu)} \leq \epsilon \mathcal{Z}_L(\beta,\mu)$ by the triangle inequality. This proves the first statement.

We note that $\widehat{\mathcal{Z}}_L(\beta,\mu) \geq (1-\epsilon)\mathcal{Z}_L(\beta,\mu)$. For the thermal average estimation, we observe that
\begin{align}
\abs{\widehat o_{L}^{\mathrm{GC}}-o_{L}^{\mathrm{GC}}}
&\leq
\abs{
\frac{\sum_s\widehat w_{L,s}(\widehat o_{L,s}-o_{L,s})}
{\widehat{\mathcal{Z}}_L(\beta,\mu)}
}
+
\abs{
\frac{\sum_s\widehat w_{L,s} o_{L,s}}{\widehat{\mathcal{Z}}_L(\beta,\mu)}
-
\frac{\sum_s w_{L,s} o_{L,s}}{\mathcal{Z}_L(\beta,\mu)}
}.
\end{align}
The first term satisfies
\begin{align}
\abs{
\frac{\sum_s\widehat w_{L,s}(\widehat o_{L,s}-o_{L,s})}
{\widehat{\mathcal{Z}}_L(\beta,\mu)}
}\le \sum_s \frac{\widehat w_{L,s}}{\widehat{\mathcal{Z}}_L(\beta,\mu)}\abs{\widehat o_{L,s}-o_{L,s}}\le \sum_s \frac{\widehat w_{L,s}}{\widehat{\mathcal{Z}}_L(\beta,\mu)} \cdot \epsilon
=
\epsilon.
\end{align}
For the second term, since $\abs{o_{L,s}}\leq\norm{O}_2\le 1$, we obtain
\begin{align}
&\abs{
\frac{\sum_s\widehat w_{L,s} o_{L,s}}{\widehat{\mathcal{Z}}_L(\beta,\mu)}
-
\frac{\sum_s w_{L,s} o_{L,s}}{\mathcal{Z}_L(\beta,\mu)}
}\\
&\le
\frac{1}{\widehat{\mathcal{Z}}_L(\beta,\mu)}
\sum_s\abs{\widehat w_{L,s}-w_{L,s}}\abs{o_{L,s}}
+
\frac{\abs{\widehat{\mathcal{Z}}_L(\beta,\mu)-\mathcal{Z}_L(\beta,\mu)}}{\widehat{\mathcal{Z}}_L(\beta,\mu) \mathcal{Z}_L(\beta,\mu)}
\abs{\sum_s w_{L,s} o_{L,s}}\\
&\leq
\frac{\epsilon \mathcal{Z}_L(\beta,\mu)}{\widehat{\mathcal{Z}}_L(\beta,\mu)}
+
\frac{\epsilon \mathcal{Z}_L(\beta,\mu)}{\widehat{\mathcal{Z}}_L(\beta,\mu)}\\
&\leq
\frac{2\epsilon}{1-\epsilon}.
\end{align}
Therefore, 
$\abs{\widehat o_{L}^{\mathrm{GC}}-o_{L}^{\mathrm{GC}}}
\leq
\left(
\epsilon+\frac{2\epsilon}{1-\epsilon}
\right)=
\frac{\epsilon(3-\epsilon)}{1-\epsilon}$, which proves the second statement for sufficiently small $\epsilon$.
\end{proof}

\begin{algorithm}[H]
\SetAlgoLined
    \KwData{System size $L$, inverse temperature $\beta$, target precision $\epsilon>0$, chemical potential $\mu$, Hamiltonian $H$, number operator $N$, and observable $O$}
    \KwResult{Approximations of the grand-canonical partition function $\widehat{\mathcal{Z}}_{L}(\beta,\mu)\approx \tr(e^{-\beta(H-\mu N)})$ and thermal expectation value $\widehat{o}_{L}^{\mathrm{GC}}\approx \frac{\tr(e^{-\beta(H-\mu N)}O)}{\tr(e^{-\beta(H-\mu N)})}$}

    \tcc{Let $\mathcal{S}$ be the set of sectors (either $\{0,\ldots,L\}$ or $\{0,\ldots,L\}^2$) and let $n_s$ denote the particle number for sector $s \in \mathcal{S}$ (i.e., $k$ or $k+k'$).}

    1. \For{$s \in \mathcal{S}$}{
       Use \cref{alg: cano} to estimate the sector partition function $\widehat{Z}_{L,s}(\beta)$ and the sector thermal expectation value $\widehat{o}_{L,s}$.
    }

    2. \textbf{Partition Function Estimate:} Compute the grand-canonical partition function via
       $\widehat{\mathcal{Z}}_{L}(\beta,\mu) = \sum_{s \in \mathcal{S}} e^{\beta\mu n_s} \widehat{Z}_{L,s}(\beta)$, which is based on \eqref{eq: identities}.

    3. \textbf{Observable Estimate:} Compute the grand-canonical expectation value of $O$ via
       $\widehat{o}_{L}^{\mathrm{GC}} = \frac{1}{\widehat{\mathcal{Z}}_{L}(\beta,\mu)} \sum_{s \in \mathcal{S}} e^{\beta\mu n_s} \widehat{Z}_{L,s}(\beta) \widehat{o}_{L,s}$, which is based on \eqref{eq: identities}.

\caption{Quantum estimation of grand-canonical ensemble properties}
\label{alg: grand cano new}
\end{algorithm}

\begin{remark}
    For comparison, one may consider a standard approach that applies the Trotter formula directly for $H-\mu N$ defined in the full Hilbert space, like QSP approaches without amplitude amplification \cite{chakraborty2025quantum},
as summarized in \cref{alg: grand cano prev}. Due to this, its complexity (e.g. circuit sampling complexity and gate complexity) depends on the chemical potential. Thus, in practical situations where one requires to search a proper value of chemical potential, \cref{alg: grand cano prev} must re-execute the quantum circuit for every trial value of $\mu$. In contrast, \cref{alg: grand cano new} avoids this bottleneck, saving a significant reduction of quantum computational resources.  Additionally,  the Trotter circuits in \cref{alg: grand cano new} do not involve the number operator, thereby yielding a  reduction of circuit complexity for Trotterized Hamiltonian simulation compared to  \cref{alg: grand cano prev}.

\end{remark}

\section*{Discussion}

\begin{center}
    \textbf{Comparison to a classical baseline}
\end{center}

As mentioned earlier, we consider the Lanczos method with  stochastic trace estimation as a classical baseline to compare with our schemes. In particular, we choose a vanilla version, specifically the method with a symmetric formulation  \cite{aichhorn2003low} rather the original approach \cite{PhysRevB.49.5065}, as such a formulation is more technically and numerically superior in the task of matrix function trace estimation \cite{ko2024ground}. While many variants have been developed (e.g. see \cite{chen2023krylov,epperly2024xtrace} and references therein), the computational cost would not significantly differ as it is dictated by the cost of matrix-vector multiplications. 

Let us denote a target Hilbert space dimension by $d$. For the case of grand canonical ensemble, $d=4^L$, and for the canonical ensemble (\cref{alg: cano}), $d=d_s$. The computational cost of the Lanczos approach is determined by the product of three factors:  $\mathcal{O}\left(\frac{\log d}{\epsilon^2}\log\frac{1}{\eta}\right)$ Rademacher vectors  multiplied with the positive-definite matrix $e^{-\frac{\beta}{2}H}$ \cite[Theorem 7.1]{avron2011randomized} (the Hutchinson's estimator), the degree of a Chebyshev polynomial approximation to the exponential function 
$\mathcal{O}\left(\beta\norm{H}_2 +\log\frac{1}{\epsilon}\right)$ \cite{tosta2024randomized}, and $\mathcal{O}(\text{poly}(L)d)$ the computational cost of matrix-vector multiplication. For the last factor, it is because the second quantized model \eqref{ex: models} forms a $\mathcal{O}(L^4)$-sparse matrix. Therefore, without polynomial prefactors, the computational cost of Lanczos approach scales as $\widetilde{\mathcal{O}}(d)$. 

In the following observation, we show that our schemes (\cref{alg: cano,alg: grand cano new}) achieve $\widetilde{\mathcal{O}}(d^\gamma)$ with $\gamma<1$ in constant-temperature regimes. Specifically, we define a figure of merit as follows,
\begin{equation}\label{eq: FoM}
\gamma(L,\beta):=\log_{d} Q(L,\beta).
\end{equation}Here $Q(L,\beta)$ accounts for a quantum circuit sampling complexity estimated with respect to system size $L$ and fixed $\beta$, without polynomial prefactors and constant computational overheads (e.g. circuit complexity, the factor $\frac{1}{\epsilon^2}$ in the circuit sampling complexity, and the factor involving failure probability $\log\frac{1}{\eta}$). The three cases are $Q_1^{(s)}$ for a canonical calculation in fixed particle-number sector $s$ (\cref{alg: cano}), $Q_1$ for a grand-canonical calculation reconstructed from the canonical sectors (\cref{alg: grand cano new}), and $Q_2$ for a direct grand-canonical calculation of $H-\mu N$ (\cref{alg: grand cano prev}),
\begin{equation}\label{eq: gammas}
    \begin{split}
        &Q_1^{(s)}=\frac{d_s e^{\beta\lambda_{L,s}}}{\tr(e^{-\beta H(L,s)})}\implies\gamma_{1}^{(s)}(L,\beta) =\log_{d_s} Q_1^{(s)},\\
    &Q_1 = \max_{s}Q_1^{(s)}\implies\gamma_1(L,\beta)=\log_{4^L}(Q_1),\\
    &Q_2 = \frac{4^L e^{\beta \lambda(H-\mu N)}}{\tr(e^{-\beta (H-\mu N)})}\implies\gamma_2(L,\beta)=\log_{4^L}(Q_2).
    \end{split}
\end{equation} 
The reason that the base differs in $\gamma_1^{(s)}$ and $\gamma_1$ is that they correspond to canonical and grand-canonical ensemble estimation, for which the cost of Lanczos method scales as $\widetilde{\mathcal{O}}(d_s)$ and $\widetilde{\mathcal{O}}(4^L)$. Intuitively, when $\beta=0$, $\gamma_1^{(s)}=\gamma_1=\gamma_2=0$ and in these cases, the circuit sampling complexity of \cref{alg: cano,alg: grand cano new,alg: grand cano prev} becomes $\mathcal{O}(1)$, which is consistent. If $\beta$ tends to $\infty$, the $\gamma$'s do so. That is, we are able to find regimes of constant temperature, where the values of $\gamma$ is strictly less than $1$. However, we note that for \cref{alg: cano,alg: grand cano new}, their complexity is independent of $\mu$ according to \cref{thm: main 1}, while \cref{alg: grand cano prev} is not. This difference is reflected in their temperature-dependent complexity as shown in the following observation.
\begin{observation}\label{obs 2}
For a fixed system size $L$, there exist positive thresholds $\beta_1^{(s)}(L),\beta_1(L)$ and $\beta_2(L)$ for \eqref{eq: gammas} such that
\begin{align*}
    &\gamma_{1}^{(s)}(L,\beta) <0.99,\quad \forall \beta \in [0,\beta_1^{(s)}),\quad \beta_1^{(s)}(L)=\frac{0.99\log d_s}{2\lambda_{L,s}}\\
    &\gamma_{1}(L,\beta) <0.99,\quad \forall \beta \in [0,\beta_1),\quad \beta_1(L)=\min_s\frac{0.99L\log 4}{2\lambda_{L,s}}\\
    &\gamma_2(L,\beta)<0.99,\quad \forall \beta\in[0,\beta_2),\quad
    \beta_2(L)=\frac{0.99L\log4}{2\lambda_L},
\end{align*}
where $\lambda_{L,s}$ and $\lambda_L$ are defined in \cref{alg: cano,alg: grand cano prev}.  Consequently, for any class of Hamiltonians satisfying
\begin{equation}\label{eq: spec norm condition}
    \norm{H(L)}_2 =\mathcal{O}(L),
\end{equation}
where \(H(L)\) denotes the Hamiltonian with system size $L$,  any $d_s$ satisfying $\inf_L \frac{\log d_s}{L}>0$, any $c_s\in(0,\frac{\log d_s}{L}]$ and a fixed $\mu$, there exist system-size-independent thresholds \(\beta_1\ge \frac{0.99(1-\delta)\log 4}{2\sup_L \frac{\norm{H(L)}_2}{L}}\), \(\beta_2\ge \frac{0.99(1-\delta)\log 4}{2\sup_L \frac{\norm{(H-\mu N)(L)}_2}{L}}\) and \(\beta_1^{(s)}\ge \frac{0.99c_s(1-\delta)}{2\sup_L \frac{\norm{H(L,s)}_2}{L}}\) such that  
\begin{align*}
    &\sup_L \gamma_1^{(s)}(L,\beta)<0.99,\quad \forall \beta \in [0,\beta_1^{(s)}),\\
    &\sup_L \gamma_1(L,\beta)<0.99,\quad \forall \beta \in [0,\beta_1),\\
    &\sup_L \gamma_2(L,\beta)<0.99,\quad \forall \beta \in [0,\beta_2).
\end{align*}

\end{observation}

\begin{proof}
The proof is straightforward by Jensen's inequality. We notice that
\begin{align}
    \frac{\tr(e^{-\beta (H(L,s))})}{d_s}
    \ge
    \exp\!\left(
    -\beta \frac{\tr(H(L,s))}{d_s}
    \right),
\end{align}and it follows
\begin{align}\label{eq: gamma1}
    \gamma_1^{(s)}(L,\beta)\le \frac{\beta}{\log d_s}\left(\lambda_{L,s}+\frac{\tr(H(L,s))}{d_s}\right)\implies \gamma_1^{(s)}(L,\beta)\le \frac{2\beta\lambda_{L,s}}{\log d_s},
\end{align}since $\left|
    \frac{\tr(H(L,s))}{d_s}
    \right|
    \le
    \norm{H(L,s)}_2
    \le
    \lambda_{L,s}$. This implies the threshold in the statement for $\gamma_1^{(s)}(L,\beta)$. Similarly using Jensen's inequality, we can prove the first statement. For the second statement, we use the first result and the condition on the Hamiltonian class. By letting the RHS of \eqref{eq: gamma1} less than $0.99$ and taking $\sup_L$, we complete the proof. 
\end{proof}
In Observation \ref{obs 2}, we apply our quantum algorithms (\cref{alg: cano,alg: grand cano new,alg: grand cano prev}) to sectors satisfying $\inf_L \frac{\log d_s}{L}>0$, whose dimensions grow exponentially with $L$. while polynomial-size sectors can be treated classically using the Lanczos algorithm.

\begin{center}
    \textbf{Numerical result}
\end{center}

In this section, we present numerical results to validate Observation \ref{obs 2} with the 2D next-nearest-neighbor Hubbard model
\begin{align}\label{eq: model hubbard}
    H &= -t \sum_{\langle i, j \rangle, \sigma} \left( c_{i\sigma}^{\dagger} c_{j\sigma} + \text{h.c.} \right) 
    - t' \sum_{\langle\langle i, j \rangle\rangle, \sigma} \left( c_{i\sigma}^{\dagger} c_{j\sigma} + \text{h.c.} \right) \nonumber \\
    &\quad + U \sum_{i} n_{i\uparrow} n_{i\downarrow}  
    - \mu \sum_{i,\sigma}n_{i\sigma}.
\end{align}Since  there are $\mathcal{O}(L)$ operators, the spectral norm of this model scales linearly with $L$ for fixed $t,t',\mu$ and $U$, and the assumption in Observation \ref{obs 2} is satisfied. We consider the geometry $L=L_x\times L_y$ with $L_x=2$ and $L_y=2,4,6,8$ with periodic boundary condition. In the numerical results, we set $t=1$ and $t'=-0.2$, corresponds to a parameter setting in the non-particle-hole symmetric case \cite{leblanc2015solutions}. We consider two cases: $U=4$ (weak-coupling regime) and $U=8$ (intermediate-coupling regime).

For our numerical results, we estimate $\gamma_1(L,\beta)$ and $\gamma_2(L,\beta)$ by using a Lanczos method with the Hutchinson estimator. First, we collect the spectra of submatrices of the full Hamiltonian matrix by applying the method to the fixed-particle subspaces in a blockwise manner. Using these computed spectra, we estimate $\gamma_1(L,\beta)$ and $\gamma_2(L,\beta)$ by introducing a spectral shift. Accordingly, given the estimated minimum and maximum eigenvalues, denoted as $\widetilde{\lambda}_{\min}$ and $\widetilde{\lambda}_{\max}$, we apply the following shifts,
\begin{align}
    H(L,(k,k'))&\longrightarrow H(L,(k,k'))-\frac{\widetilde{\lambda}_{\min,L,(k,k')}+\widetilde{\lambda}_{\max,L,(k,k')}}{2}\\
    H&\longrightarrow H-\frac{\widetilde{\lambda}_{\min}+\widetilde{\lambda}_{\max}}{2}. 
\end{align}
Applying these shifts reduces the normalization factors required when compiling the Hamiltonian simulation via Trotter circuits in the GQSP in \cref{fig: GQSP}, thereby reducing the circuit sampling complexity. Thus, we can modify the quantities in  \eqref{eq: gammas} as follows,
\begin{align}\label{eq: estimates of min,max}
    &\gamma_1(L,\beta)=\max_{k,k'}\frac{1}{L\log 4}\left(\log\left[\binom{L}{k}\binom{L}{k'} \right]-\beta \widetilde{\lambda}_{\min,L,(k,k')}-\log\tr(e^{-\beta(H(L,(k,k')))}) \right)\\
    &\gamma_2(L,\beta)=\frac{1}{L\log 4}\left(L-\beta \widetilde{\lambda}_{\min}(H-\mu N) - \log \tr(e^{-\beta(H-\mu N)}) \right).
\end{align}
Notice that the error involved in $\widetilde{\lambda}_{\max}$ does not affect the sample complexity but the circuit complexity, such as gate complexity, as seen in \cref{thm: main 1} and \cref{thm: main 2}. In our simulations, we define approximations of the extreme eigenvalues, $\widetilde{\lambda}_{\min}=(1+\epsilon_{\min})\lambda_{\min}$ and $\widetilde{\lambda}_{\max}=(1+\epsilon_{\max})\lambda_{\max}$. For the results in \cref{fig: g1U4}, \cref{fig: g1U8}, and \cref{fig: g2}, we set $\epsilon_{\min}=0.005$ and $\epsilon_{\max}=0.007$ as $\lambda_{\min}<0$ and $\lambda_{\max}>0$ in the results. We set the dimension of Krylov subspace to be $40$ for large fixed-particle subspaces and use $30$ Rademacher vectors as the number of random samples.  For small subspaces, we use the exact eigensolver. 
Once the spectrum of the Hamiltonian has been obtained, we compute an approximation of the chemical potential $\mu$ for a given density $n$ by using the Brent's Method in the scipy.optimize.

\cref{fig: g2} plots the values of $\gamma_1(L,\beta)$, which converge to constants below approximately $0.9$ at temperatures $T=0.05, 0.1,$ and $0.5$. In contrast, at the lowest temperature ($T=0.01$), the value exceeds $1$. Together with Observation \ref{obs 2}, these observations suggest that \cref{alg: grand cano new} can outperform the classical Lanczos method for estimating grand-canonical thermal properties at temperatures down to somewhere between $T=0.05$ and $T=0.01$. As mentioned earlier, a key advantage of \cref{alg: grand cano new} is that its performance, as characterized by $\gamma_1(L,\beta)$, is independent of the chemical potential and hence of the electron density. For comparison, under the same parameter settings as in \cref{fig: g2}, we evaluate \cref{alg: grand cano prev} by estimating $\gamma_2(L,\beta)$ in \eqref{eq: gammas} for electron densities $n=0.875, 0.9, 0.925, 0.95, 0.975,$ and $1$. In \cref{fig: g1U4} and \cref{fig: g1U8}, for a fixed temperature, we observe that $\gamma_2(L,\beta)$ generally increases as the electron density increases from $0.875$ to $1$. At the lowest temperature ($T=0.01$), $\gamma_2(L,\beta)>1$ in all cases, whereas at the highest temperature ($T=0.5$), it remains below $1$. At the intermediate temperatures ($T=0.05$ and $0.1$), whether $\gamma_2(L,\beta)$ remains below $1$ depends on the electron density, in contrast to the density-independent behavior of \cref{alg: grand cano new} observed in \cref{fig: g2}.

We also examine how the relative error of the ground state energy affects the circuit sampling complexity of \cref{alg: grand cano new}. Using the same spectral data from our previous results, we perform simulations with $\epsilon_{\min}$ values of $0.001$, $0.003$, $0.005$, $0.007$, and $0.01$. Based on \eqref{eq: estimates of min,max}, we expect the value of $\gamma_1(L,\beta)$ to scale linearly with the relative error and the inverse temperature. Indeed, \cref{fig: g2err1} and \cref{fig: g2err2} confirm this expectation.  For example, \cref{fig: g2err2} shows that at the highest temperature ($T=0.5$), a relative error of up to $0.1$\% has an almost negligible effect on $\gamma_1(L,\beta)$ as the value remains less than $0.8$ for all cases.  At intermediate temperatures, it is observed that $\gamma_1(L,\beta)$ does not exceed $0.9$ ($T = 0.1$) and $0.95$ ($T = 0.05$). However, the results at the lowest temperature ($T = 0.01$) imply that to outperform the classical baseline, the relative error of the ground state energy must be sufficiently small. In the case of weak-coupling regime ($U=4$) in \cref{fig: g2err1}, we clearly see worse scaling results than the intermediate-coupling case ($U=8$), which aligns consistently with the observations in \cref{fig: g2}.

\begin{figure}[htbp]
    

    \begin{center}
   \begin{subfigure}[b]{0.45\textwidth}
    \centering
    \includegraphics[width=\textwidth]{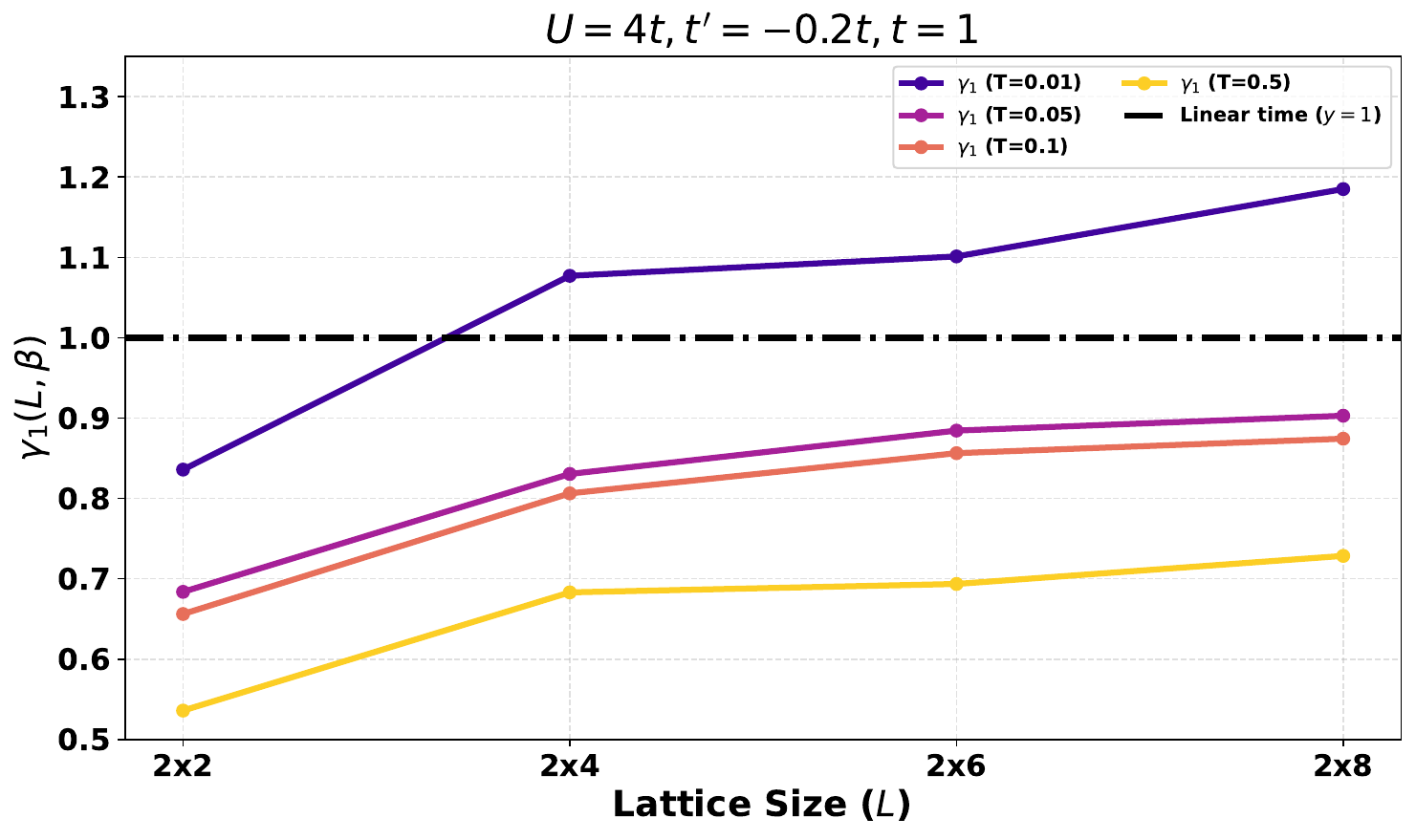}
    \end{subfigure}
    \begin{subfigure}[b]{0.45\textwidth}
    \centering
    \includegraphics[width=\textwidth]{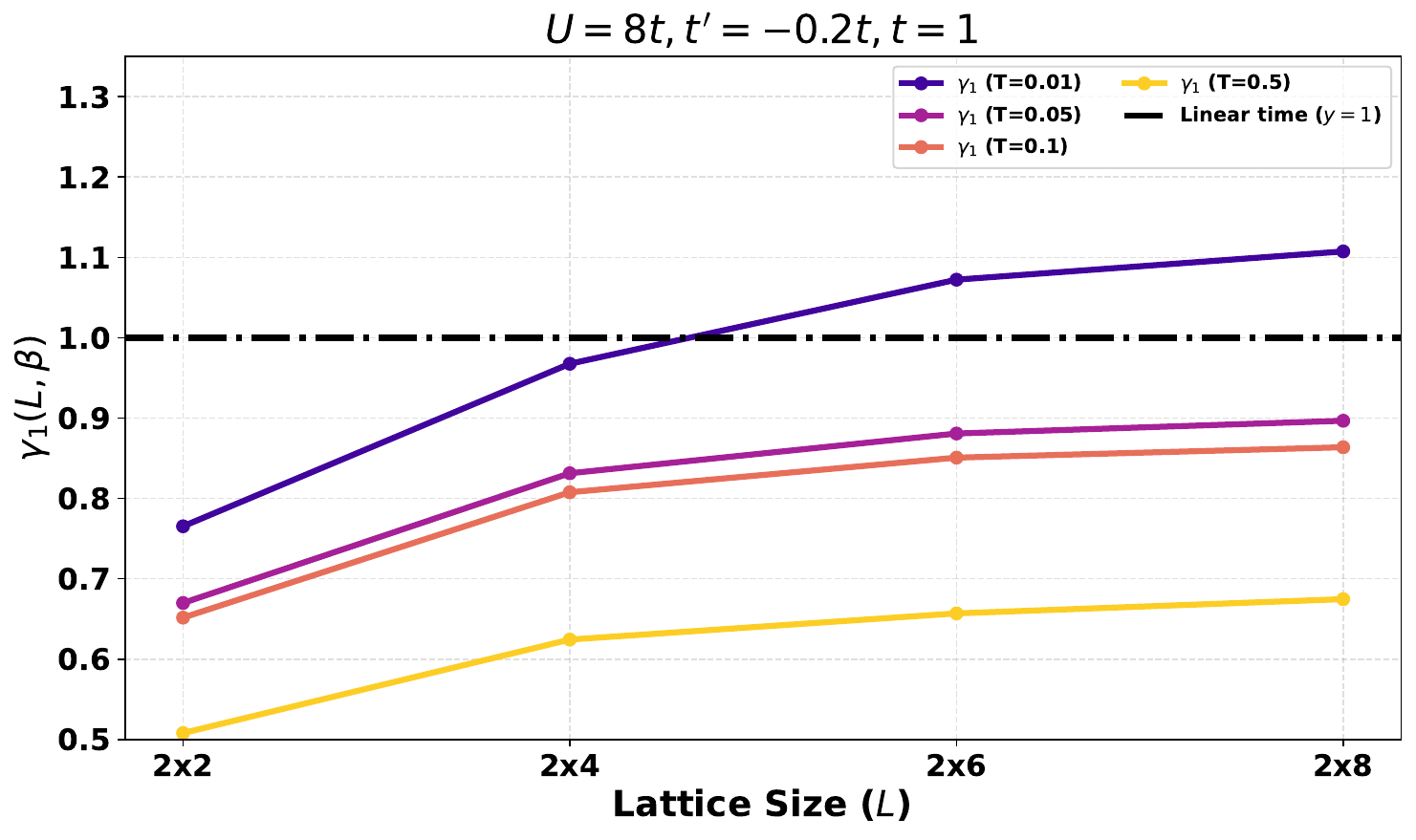}
    \end{subfigure}

\end{center}

    \caption{Simulation results evaluating the figure of merit $\gamma_1(L,\beta)$ for the proposed grand-canonical scheme (\cref{alg: grand cano new}). The figure of merit $\gamma_1$ is plotted as a function of system size ($L$) across different temperatures ($T \in \{0.01, 0.05, 0.1, 0.5\}$). (Left panel) Results for the weak-coupling regime ($U=4$). (Right panel) Results for the intermediate-coupling regime ($U=8$).}
    \label{fig: g2}
\end{figure}

\begin{figure}[htbp]
    

    \begin{center}
   \begin{subfigure}[b]{0.45\textwidth}
    \centering
    \includegraphics[width=\textwidth]{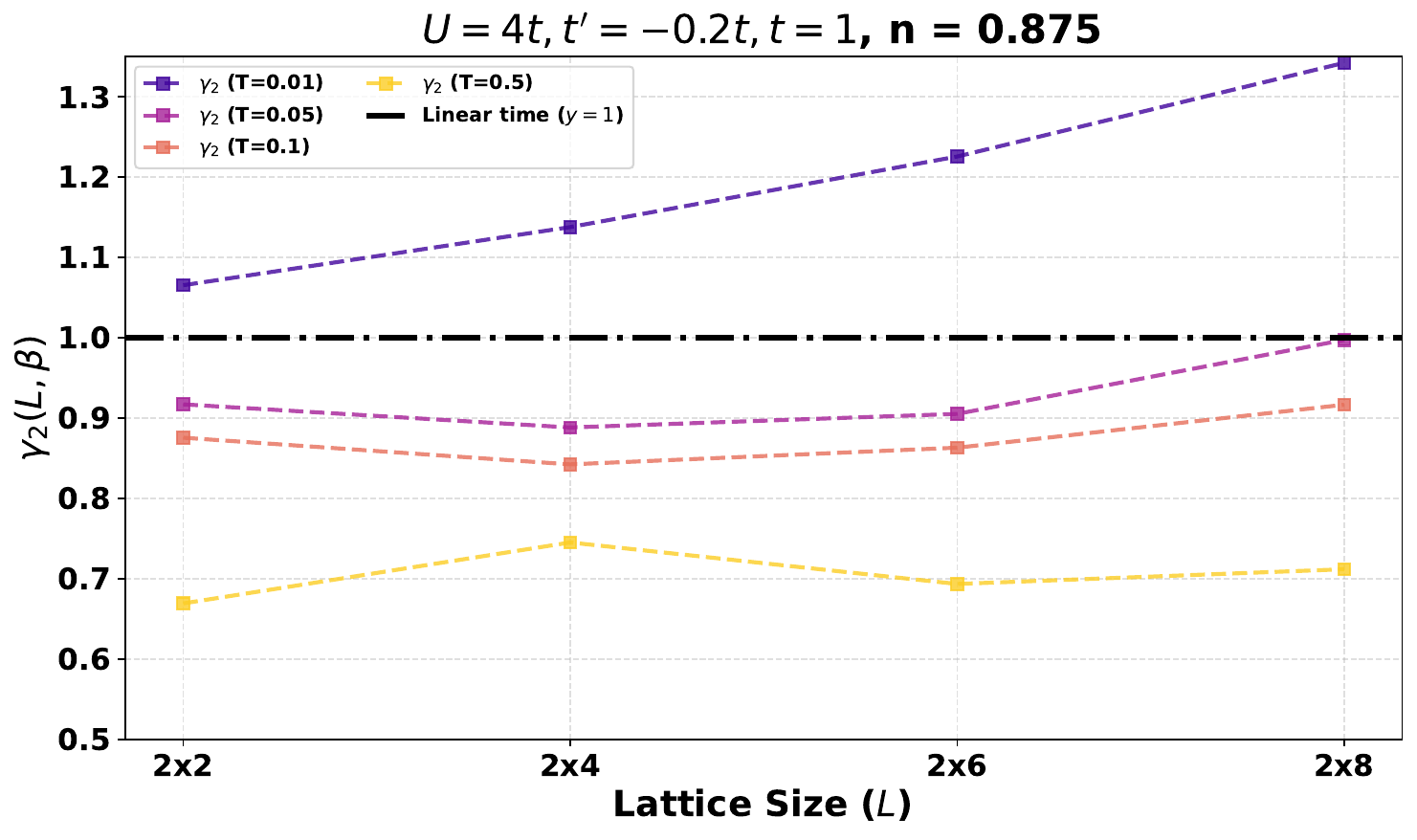}
    \end{subfigure}
    \begin{subfigure}[b]{0.45\textwidth}
    \centering
    \includegraphics[width=\textwidth]{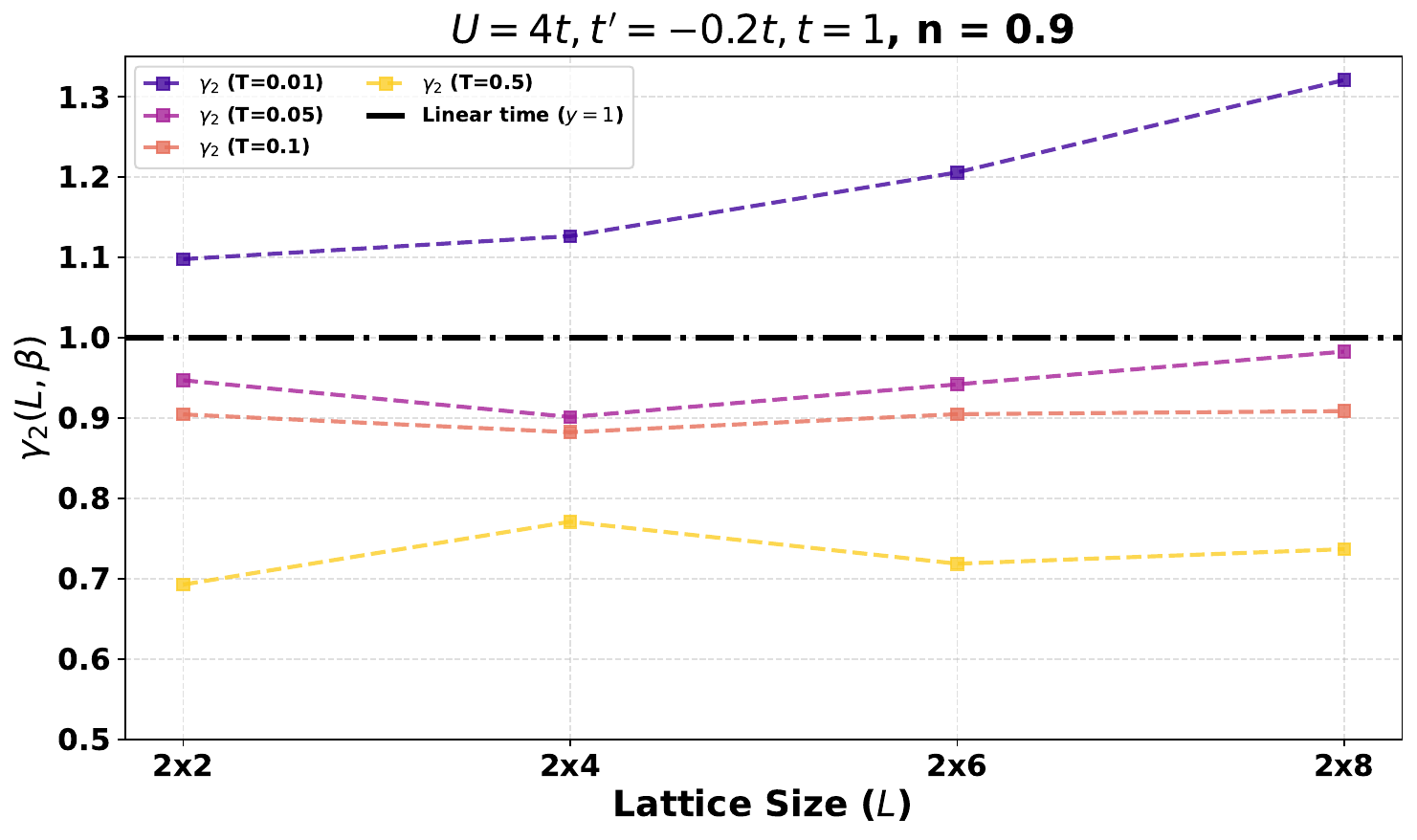}
    \end{subfigure}
    \begin{subfigure}[b]{0.45\textwidth}
    \centering
    \includegraphics[width=\textwidth]{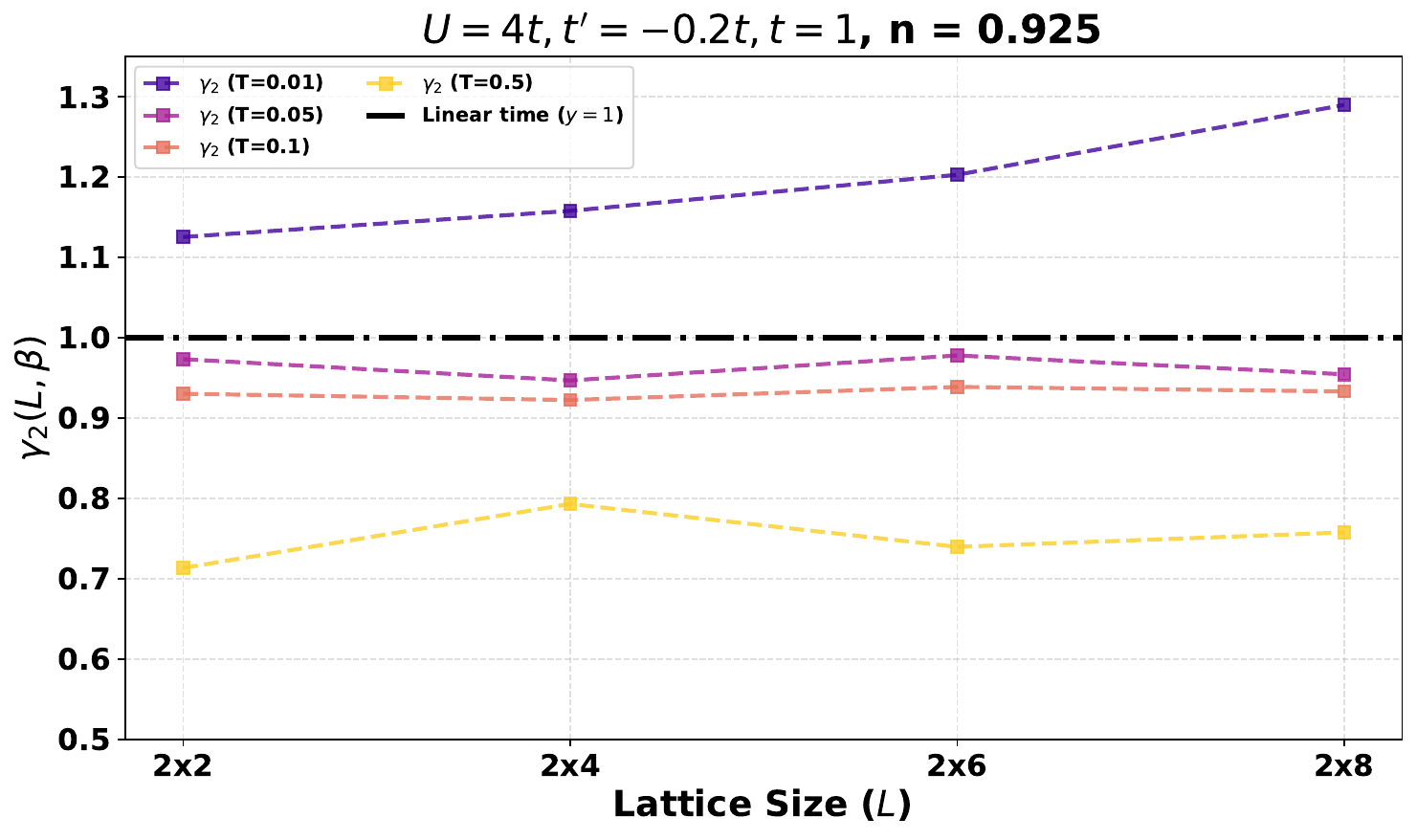}
    \end{subfigure}
    \begin{subfigure}[b]{0.45\textwidth}
    \centering
    \includegraphics[width=\textwidth]{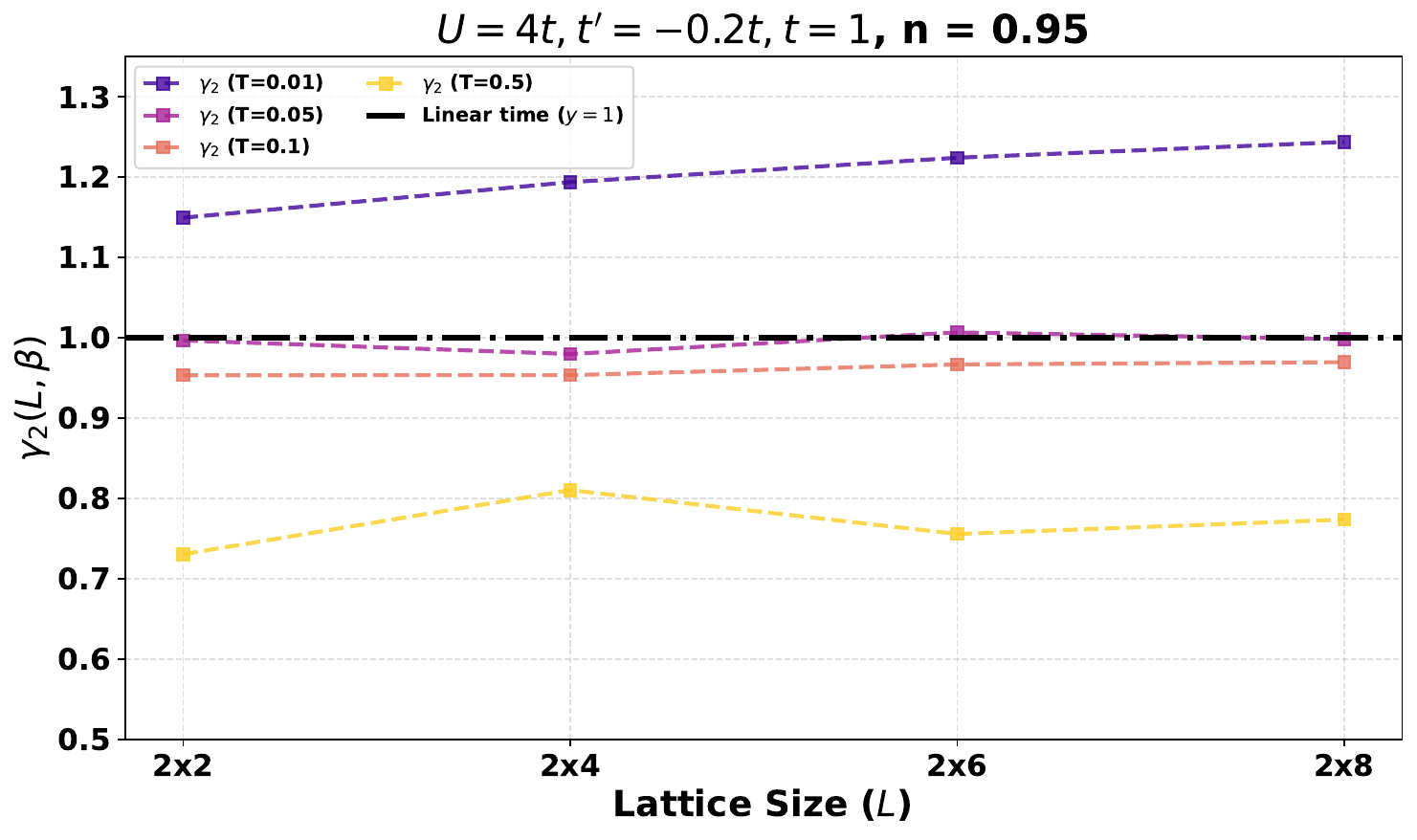}
    \end{subfigure}

    \begin{subfigure}[b]{0.45\textwidth}
    \centering
    \includegraphics[width=\textwidth]{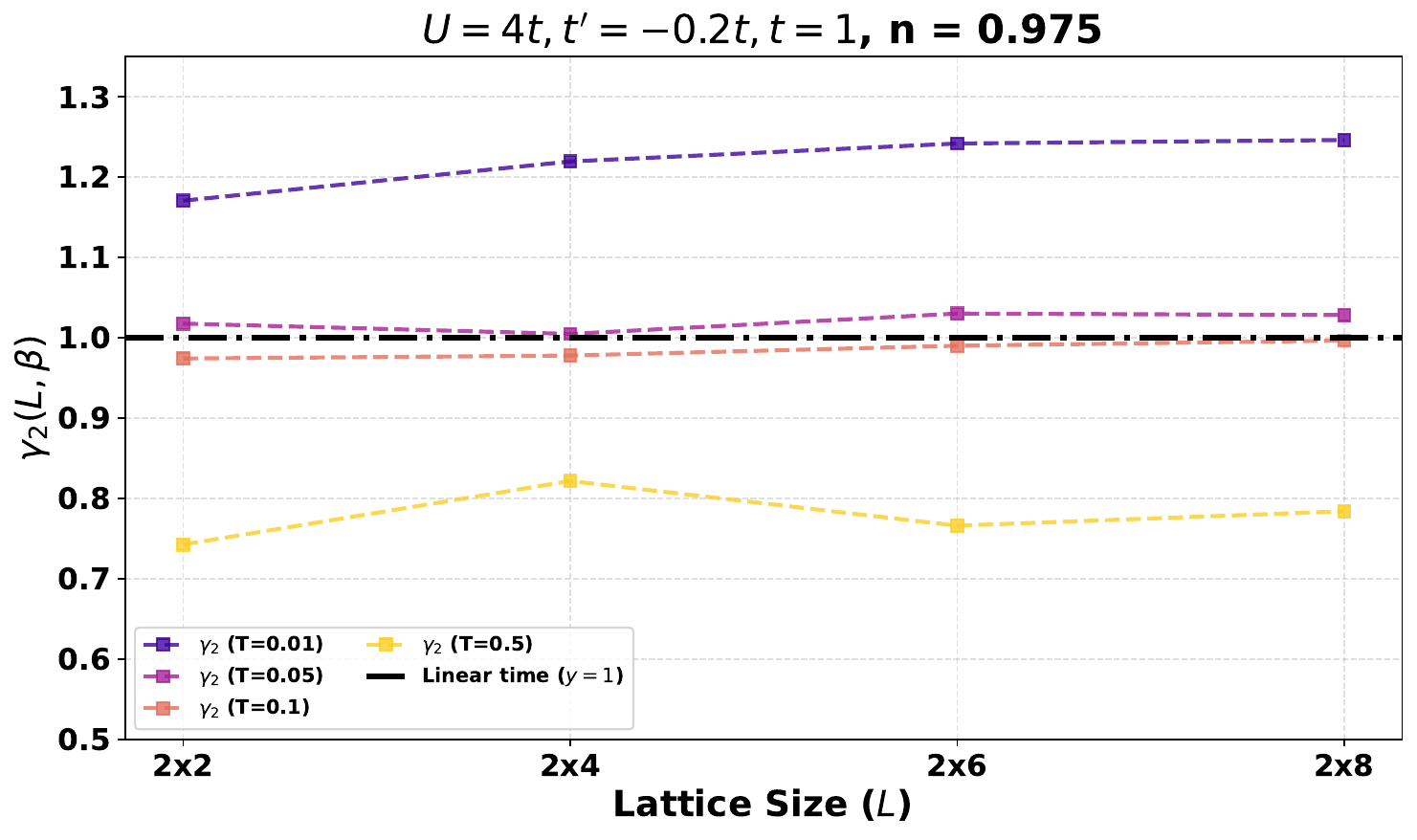}
    \end{subfigure}
    \begin{subfigure}[b]{0.45\textwidth}
    \centering
    \includegraphics[width=\textwidth]{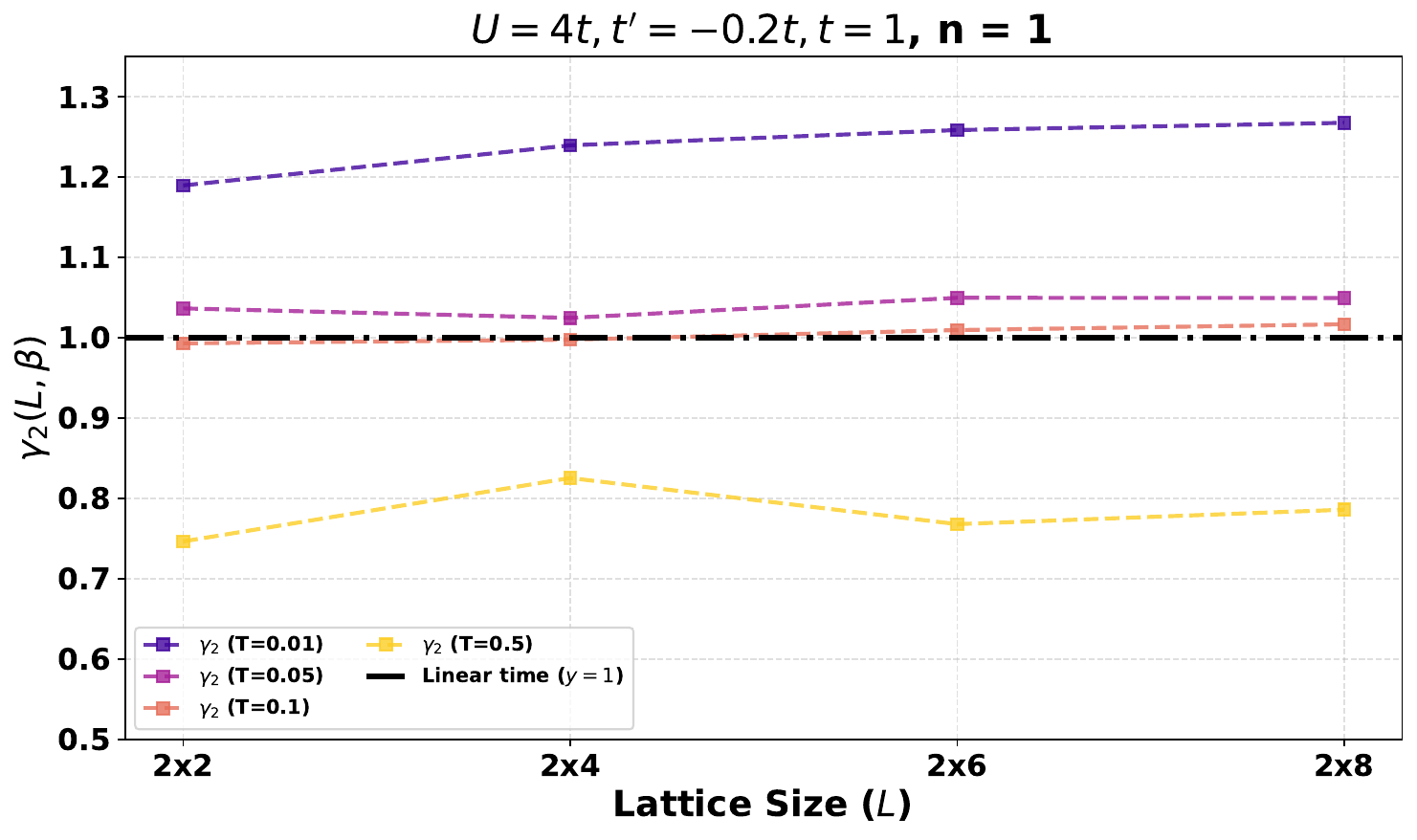}
    \end{subfigure}

\end{center}
   
    \caption{Simulation results evaluating the figure of merit $\gamma_2(L,\beta)$ for the conventional grand-canonical approach (\cref{alg: grand cano prev}) in the weak-coupling regime ($U=4$). The figure of merit $\gamma_2$ is plotted as a function of system size ($L$) across different temperatures ($T \in \{0.01, 0.05, 0.1, 0.5\}$). Results are obtained for electron densities $n=0.875, 0.9, 0.925, 0.95, 0.975$, and $1.0$, respectively.}
    \label{fig: g1U4}
\end{figure}

\begin{figure}[htbp]
    

    \begin{center}
    \begin{subfigure}[b]{0.45\textwidth}
    \centering
    \includegraphics[width=\textwidth]{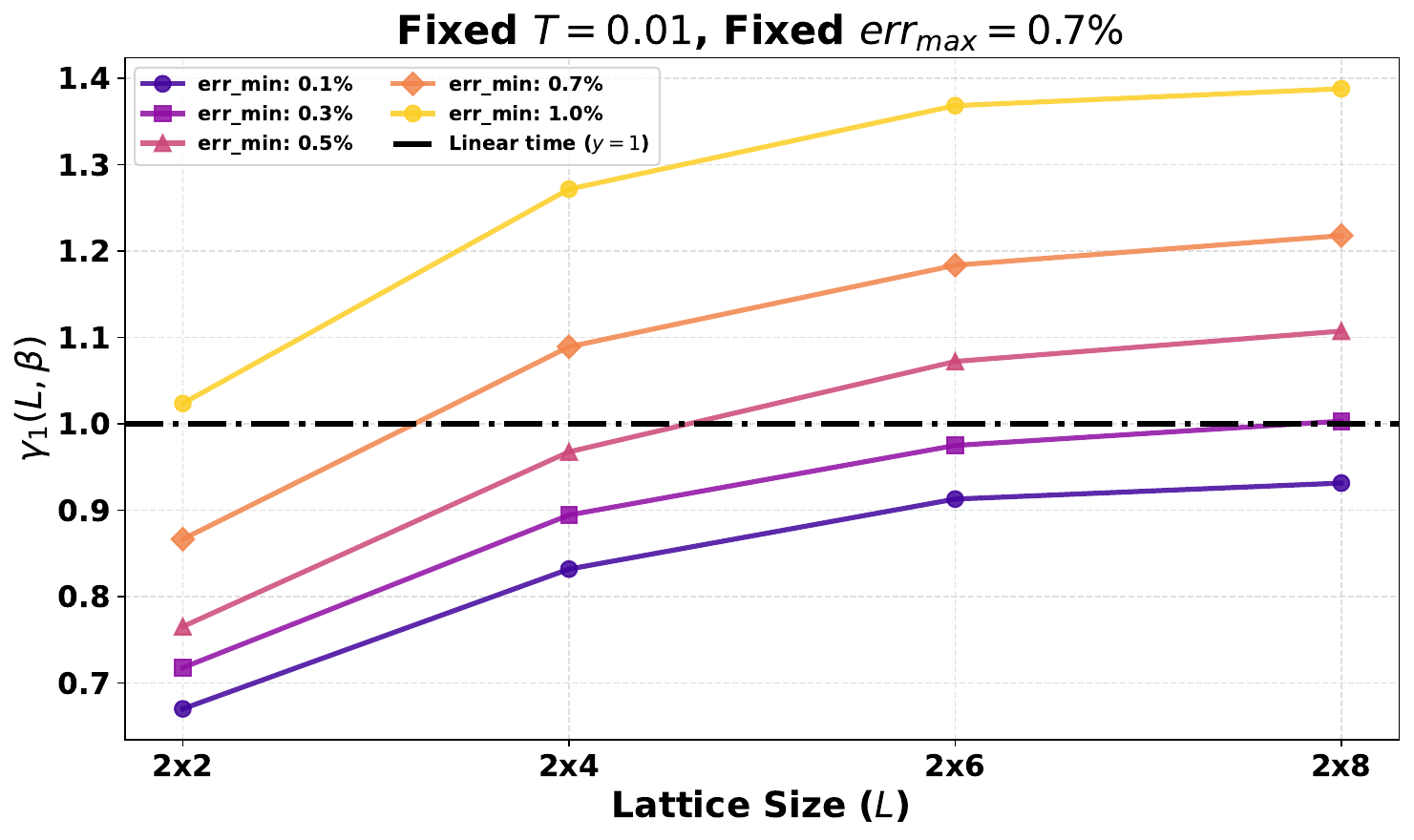}
    \end{subfigure}
    \begin{subfigure}[b]{0.45\textwidth}
    \centering
    \includegraphics[width=\textwidth]{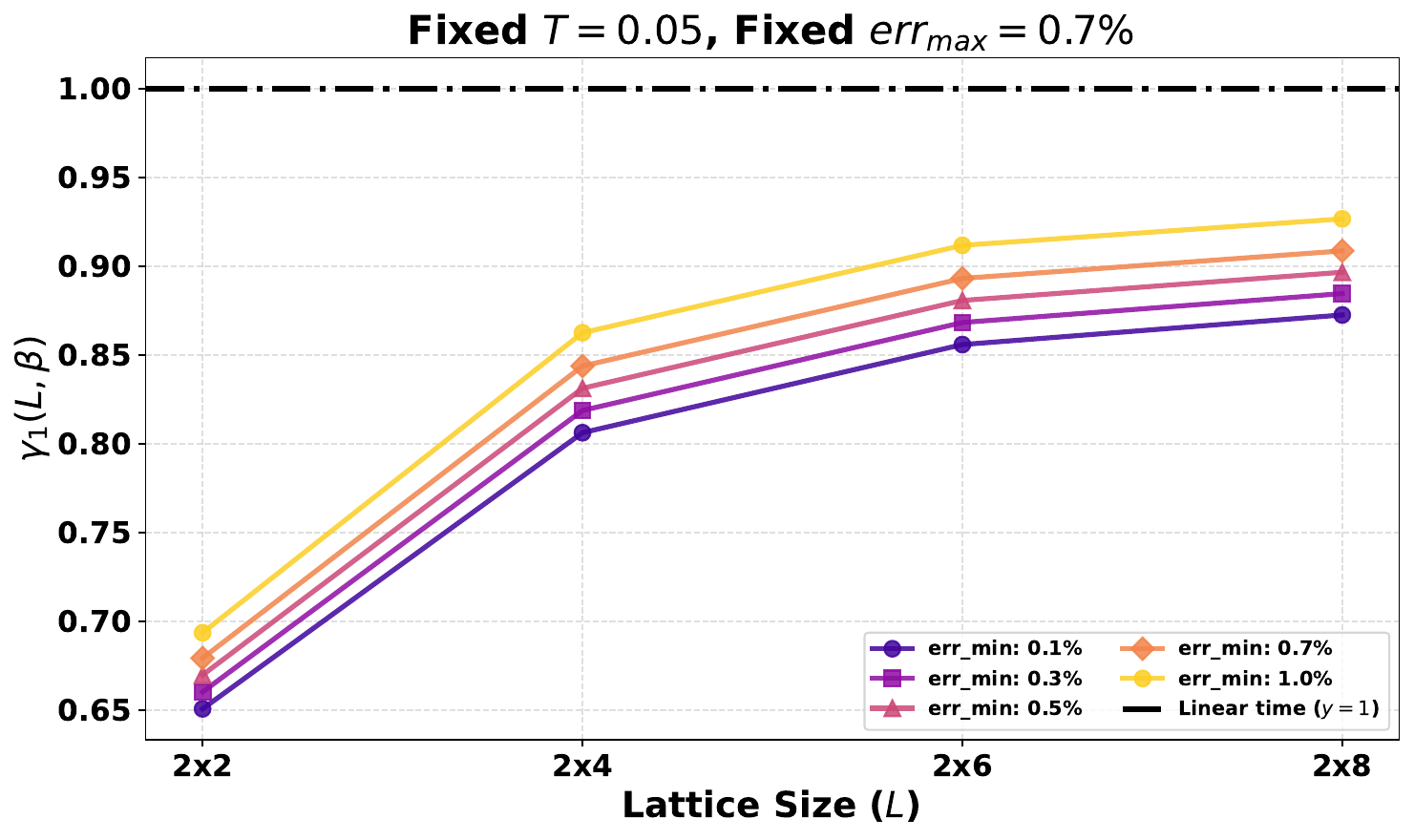}
    \end{subfigure}
    \begin{subfigure}[b]{0.45\textwidth}
    \centering
    \includegraphics[width=\textwidth]{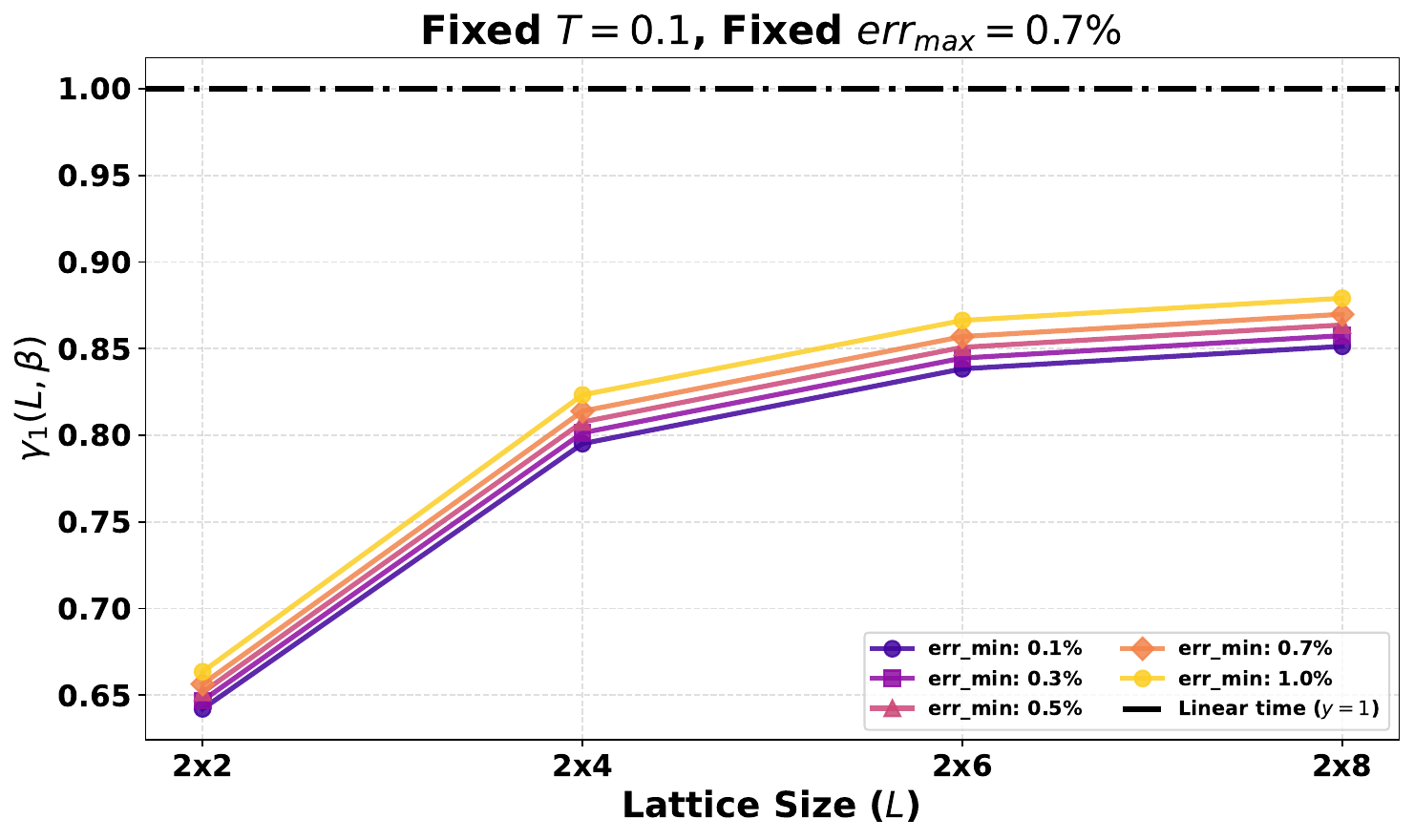}
    \end{subfigure}
    \begin{subfigure}[b]{0.45\textwidth}
    \centering
    \includegraphics[width=\textwidth]{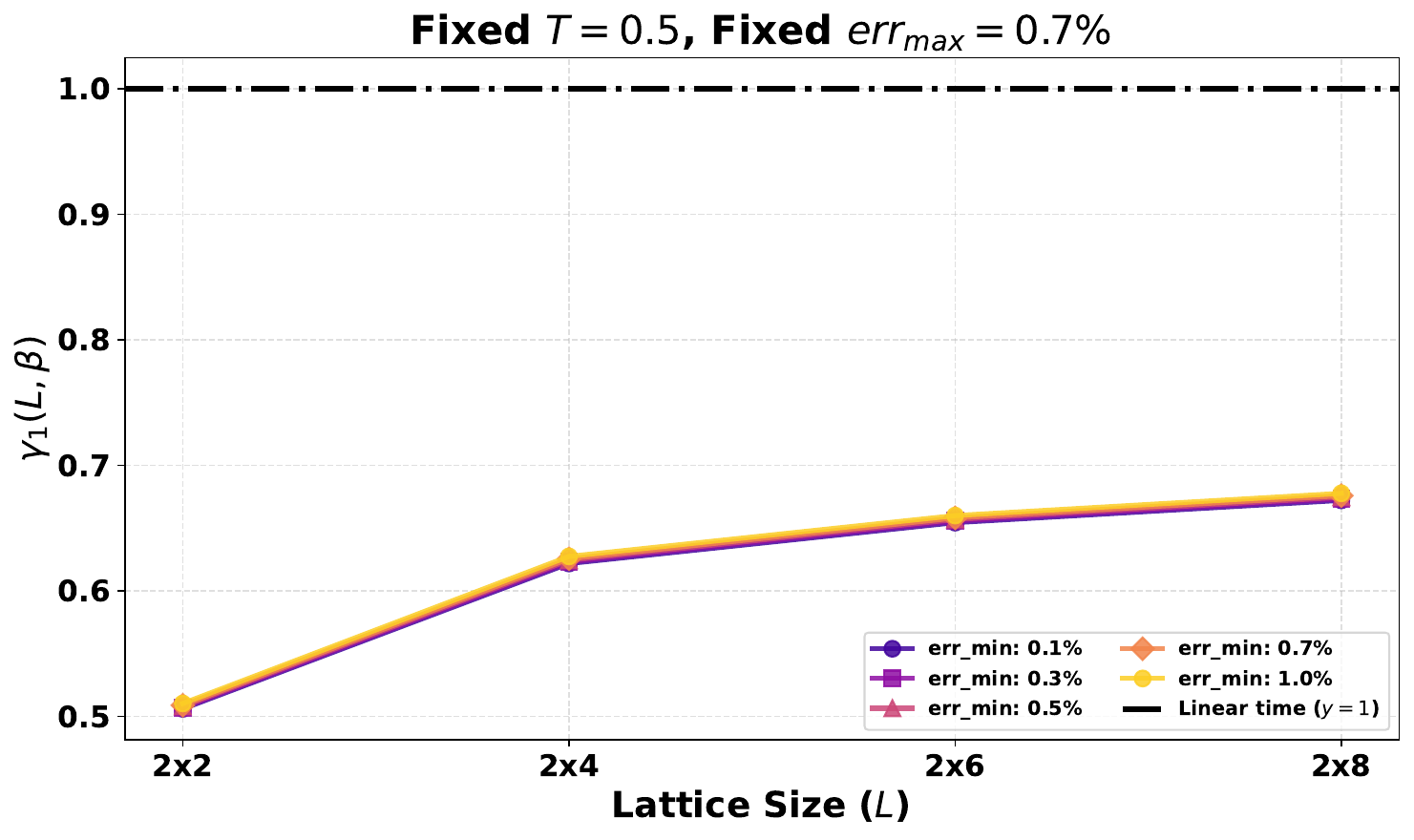}
    \end{subfigure}

\end{center}

    \caption{Simulation results illustrating the impact of the ground state energy relative error ($\epsilon_{\min}$) on the figure of merit $\gamma_1(L,\beta)$ for the proposed grand-canonical scheme in the intermediate-coupling regime ($U=8$). The figure of merit $\gamma_1$ is plotted as a function of system size ($L$) for varying relative errors $\epsilon_{\min} \in \{0.001, 0.003, 0.005, 0.007, 0.01\}$.   Results are obtained for different temperatures $T=0.01,0.05,0.1,0.5$.}
    \label{fig: g2err2}
\end{figure}

\section*{Methods}

\begin{table}[t]
\centering
\renewcommand{\arraystretch}{1.25}
\begin{tabular}{ll}
\hline
\textbf{Notation} & \textbf{Description} \\
\hline

$L$ & System size (total number of spin orbitals, or total number of spatial modes/lattice sites) \\

$s$ 
& Particle number sector ($s=k$ for total particles, $s=(k,k')$ for spin-resolved) \\

$k,\, k'$ 
& Number of total/spin-up particles ($k$) and spin-down particles ($k'$) \\

$d_s$ 
& Dimension of the Hilbert space for sector $s$ (e.g., $\binom{L}{k}$ or $\binom{L}{k}\binom{L}{k'}$) \\

$H$ 
& Full second-quantized many-body Hamiltonian \\

$H(L,s)$ 
& Restriction of the Hamiltonian $H$ to the $s$-particle sector \\

$\widetilde{H}$ 
& Effective Hamiltonian generated by the Trotter formula \\

$N$ 
& Total particle number operator \\

$c_i^\dagger,\, c_i$ 
& Fermionic creation and annihilation operators on mode $i$ \\

$X_i,\, Y_i,\, Z_i$ 
& Pauli matrices acting on the $i$-th qubit \\

$\beta$ 
& Inverse temperature \\

$\mu$ 
& Chemical potential \\

$Z_{L,s}(\beta)$ 
& Canonical partition function for the $s$-particle sector \\

$\mathcal{Z}_L(\beta, \mu)$ 
& Grand-canonical partition function \\

$O$ 
& Target observable to be measured \\

$\widehat{Z},\, \widehat{o}$ 
& Algorithmic estimates of the partition function and observable expectation \\

$\omega_{L,s}$ 
& Mixed Dicke state (maximally mixed state representing the sector $s$) \\

$\epsilon$ 
& Target precision (relative for partition functions, additive for observables) \\

$\eta$ 
& Target failure probability for the overall estimation \\

$\delta$ 
& Spectral normalization margin \\

$\lambda_{L,s}$ 
& Normalization factor satisfying $\norm{H(L,s)/\lambda_{L,s}}_2 \le 1-\delta$ \\

$r$ 
& Number of Trotter steps \\

$N_{\text{shots}}$ 
& Circuit sampling complexity (total number of circuit measurements) \\

$\gamma$ 
& Figure of merit exponent defining the scaling of circuit sampling complexity \\

\hline
\end{tabular}
\caption{Summary of key mathematical notations and parameters used throughout this work.}
\label{tab:notation}
\end{table}

\begin{center}
    \textbf{Conventional grand-canonical scheme}
\end{center}

Beyond \eqref{ex: models}, we consider the usual class of quantum Hamiltonians of the form
\begin{align}\label{eq: generic Ham}
    H=\sum_{\ell=1}^m a_\ell P_{\ell},
\end{align}
where $m=\mathcal{O}(\text{poly}(L))$, $a_\ell=\mathcal{O}(1)$, and $P_\ell$ denotes an $L$-qubit Pauli operator. For this model, a numerical scheme for estimating grand-canonical thermal properties is summarized in \cref{alg: grand cano prev}. The following theorem establishes the overall complexity of the algorithm.
\begin{theorem}\label{thm: main 2}
    Assume a normalization factor $\lambda_L$ such that $\norm{\frac{H-\mu N}{\lambda_L}} \le 1-\delta$ for some $\delta \in (0,1)$. 
    Using the first-order Trotter formula with a GQSP circuit in \cref{fig: GQSP}, estimating the partition function $\tr(e^{-\beta (H-\mu N)})$ to a relative error $\epsilon$ and Gibbs-state property $\tr(\frac{e^{-\beta (H-\mu N)}}{\tr(e^{-\beta (H-\mu N)})} O)$ to an additive error $\epsilon\norm{O}_2$ with failure probability $\eta$ for the Hamiltonian in \eqref{eq: generic Ham} requires
    \begin{itemize}
        \item $2L$ qubits for the initial state preparation 
        \item $L$ system qubits and a single ancilla qubit for the GQSP
        \item  $r = \Omega\left(\Lambda_H^2\max\left[\frac{1}{\lambda_L\delta},\; \frac{\beta}{\log(1+\epsilon)}\right]\right)$
        \item $\mathcal{O}\left(r\text{poly}(L)\left(\ln\left(\frac{1}{\epsilon}\right)+\beta\lambda_L\right)\frac{1}{\delta}\right)$ gate complexity
        \item $\mathcal{O}\left(\frac{2^L e^{\beta\lambda_L}}{\operatorname{tr}(e^{-\beta (H-\mu N)})\epsilon^2}\log\frac{1}{\eta}\right)$ sample complexity
    \end{itemize}  
\end{theorem}
The proof follows immediately from that of \cref{thm: main 1}. 

We note that for the second-quantized model \eqref{ex: models}, the polynomial factor in gate complexity reduces to $L^4$. In addition, under JW transformation using $2L$ qubits, the qubit number $L$ appearing in this theorem is replaced by $2L$, so that the full Hilbert-space dimension $2^L$ becomes $4^L$.

\begin{algorithm}[H]
\SetAlgoLined
    \KwData{System size $L$, inverse temperature $\beta$, precision $\epsilon>0$, chemical potential $\mu$, normalization factor $\lambda_{L}>0$ (where $\norm{(H-\mu N)/\lambda_{L}}\leq 1-\delta$), and observable $O$}
    \KwResult{Approximations of the grand-canonical partition function $\widehat{\mathcal{Z}}_{L}(\beta,\mu)$ and the thermal expectation value $\widehat{o}_{L}^{\mathrm{GC}}$ of $O$}

    1. Construct the GQSP circuit (\cref{fig: GQSP}) using Trotter circuits for
$e^{\pm i(H-\mu N)/\lambda_{L}}$ and initial state $\omega=I/2^L$
($\omega=I/4^L$ for a spin-$\frac{1}{2}$ fermionic system with $L$ spatial modes).

    2. \textbf{Partition Function Estimate:} Execute the circuit for
$N_{\text{shots}}$ shots. Let $N_0$ be the number of ancilla $\ket{0}$ outcomes.
Estimate
$\widehat{\mathcal{Z}}_{L}(\beta,\mu)
=2^L e^{\beta\lambda_{L}}\frac{N_0}{N_{\text{shots}}}$
($\widehat{\mathcal{Z}}_{L}(\beta,\mu)
=4^L e^{\beta\lambda_{L}}\frac{N_0}{N_{\text{shots}}}$
for the spinful case).

    3. \textbf{Observable Estimate:} Execute the circuit for $N_{\text{shots}}$ shots. Conditioned on the $M_0$ occurrences where the ancilla is $\ket{0}$, measure $O$ on the system register to obtain outcomes $\{o_j\}_{j=1}^{M_0}$. Estimate $\widehat{o}_{L}^{\mathrm{GC}} = \frac{1}{M_0} \sum_{j=1}^{M_0}o_j$.

\caption{Quantum estimation of grand-canonical ensemble properties (direct method)}
\label{alg: grand cano prev}
\end{algorithm}

\begin{center}
    \textbf{Lemmas for proving the main theorem}
\end{center}

To estimate the number of Trotter steps in \cref{thm: main 1}, we use the result \cite[Lemma 1]{watson2025exponentially}. Let $H=\sum_{\gamma=1}^{\Gamma} H_j$ be a target Hamiltonian composed of $\Gamma$ time-independent terms. For a discrete time step $\Delta\tau=\tau/r$, the first-order Trotter formula is $S_1(\Delta\tau)
    :=
    \prod_{\gamma=1}^{\Gamma}e^{-iH_\gamma\Delta\tau}
    =
    e^{-i\widetilde H(\Delta\tau)\Delta\tau}$,
where $\widetilde H(\Delta\tau)$ denotes the corresponding effective Hamiltonian. For convenience, we write $\widetilde H=\widetilde H(\Delta\tau)$.

To exploit the result \cite[Lemma 1]{watson2025exponentially} for our purposes, we restate it as follows.

\begin{lemma}\label{lem: p_order_generator}If $r\ge 4\tau \Lambda_H$,
then $\norm{\widetilde H-H}_2
    \le
    \Lambda_H^2\Delta \tau
    $, where $\Lambda_H:=\sum_{\gamma=1}^{\Gamma}\norm{H_\gamma}_2$.
\end{lemma}

\begin{proof}

Based on the result \cite[Lemma 1]{watson2025exponentially}, the first-order Trotter formula satisfies that $a_{\max}=1$, $\Upsilon=1$, and $\norm{E_{j+1}}_2\le \frac{\alpha_{\mathrm{comm}}^{(j)}}{j^2}$, where $\alpha_{\mathrm{comm}}^{(j)}=\sum_{\gamma_1,\ldots,\gamma_j=1}^{\Gamma}
\norm{
[H_{\gamma_1},[H_{\gamma_2},\ldots,H_{\gamma_j}]]}_2
$, the $j$-th folded nested commutator. Using $\norm{[A,B]}_2\le 2\norm{A}_2\norm{B}_2$, we have 
\begin{align}
    \alpha_{\mathrm{comm}}^{(j)}\le \sum_{\gamma_1,\ldots,\gamma_j=1}^{\Gamma}
2^{j-1}\prod_{i=1}^j\norm{H_{\gamma_i}}_2=2^{j-1}\Lambda_H^j,
\end{align}and thus
\begin{align}
    \norm{E_{j+1}}_2\le \frac{\alpha_{\mathrm{comm}}^{(j)}}{j^2}=\frac{1}{2(j+1)^2}(2\Lambda_H)^{j+1}.
\end{align}
Then, $\norm{\widetilde H-H}_2
    \le
    \Lambda_H
    \sum_{j=1}^{\infty}
    \frac{(2\Lambda_H\Delta\tau)^j}{(j+1)^2}$.
If $4\Lambda_H\Delta\tau<1$, using $(j+1)^{-2}\le 1/4$ for $j\ge1$ gives
\begin{align}
    \norm{\widetilde H-H}_2
    \le
    \frac{\Lambda_H^2\Delta\tau}
    {2(1-2\Lambda_H\Delta\tau)}\le \Lambda_H^2\Delta\tau.
\end{align}

\end{proof}

\begin{lemma}\label{lem: effective Ham}
The first-order Trotter formula composed of circuits in \cref{Trotter circuits} preserves the number symmetry and so does the corresponding effective Hamiltonian to the Hamiltonian simulation $e^{iH\tau}$.
\end{lemma}

\begin{proof}
The case of $2$-body terms in \cref{Trotter circuits} is the only non-trivial case. In this case, Trotter circuit blocks defined by  $H_{pquv}^{(2)}$ are divided into two parts by considering the real and imaginary parts of the coefficient. The case of imaginary part is shown in \cite[Appendix]{romero2019strategies}. Using the same idea, we verify the case of real part under the JW transformation,
\begin{align}
    &\Re\{h_{ij}\} (c_i^\dagger c_j + c_j^\dagger c_i) = \frac{\Re\{h_{ij}\}}{2} \left( \bigotimes_{k=i+1}^{j-1} Z_k \right) (X_i X_j + Y_i Y_j)\\
    &\Re\{g_{pquv}\} (c_p^\dagger c_q^\dagger c_u c_v + c_v^\dagger c_u^\dagger c_q c_p) = \frac{\Re\{g_{pquv}\}}{8} \left( \bigotimes_{a=v+1}^{u-1} Z_a \bigotimes_{b=q+1}^{p-1} Z_b \right) \big( P_{vuqp} \big),
\end{align}where the sum of the Pauli strings $P_{vuqp}$ is
\begin{align}
    P_{vuqp} = &-X_v X_u X_q X_p + X_v X_u Y_q Y_p - X_v Y_u X_q Y_p - X_v Y_u Y_q X_p \nonumber \\
    &- Y_v X_u X_q Y_p - Y_v X_u Y_q X_p + Y_v Y_u X_q X_p - Y_v Y_u Y_q Y_p.
\end{align}The commutativity between any of the operators is equivalent to the commutativity of the strings containing $X$ and $Y$ operators only. As in \cite{romero2019strategies}, we consider two arbitrary strings of $X$ and $Y$ operators of length $2n$, $P_A = \bigotimes_{i=1}^{2n} \sigma_i^a$ and $P_B = \bigotimes_{i=1}^{2n} \sigma_i^b$, acting on the same set of qubits. The commutator is then given by,
\begin{equation}
    [P_A, P_B] = \bigotimes_{i=1}^{2n} (\sigma_i^a \sigma_i^b) - \bigotimes_{i=1}^{2n} (\sigma_i^b \sigma_i^a) 
\end{equation}
where the product $\sigma_i^a \sigma_i^b$ satisfies
\begin{equation}
    \sigma_i^a \sigma_i^b = \begin{cases} 
        1 & \text{if } a = b \\ 
        iZ_i & \text{if } a = x, b = y \\ 
        -iZ_i & \text{if } a = y, b = x 
    \end{cases}. 
\end{equation}
Then, we obtain
\begin{align}
    [P_A, P_B] &= \left[ i^{(n_y^A - c_y - c_x)} (-i)^{(n_y^B - c_y - c_x)} - (-i)^{(n_y^A - c_y - c_x)} i^{(n_y^B - c_y - c_x)} \right] P\\
    &=P \cdot i^{(n_y^A + n_y^B - 2c_y - 2c_x)} \left( 1 - (-1)^{n_y^A - n_y^B} \right),
\end{align}
where $P$ is the string of Pauli matrices obtained from the multiplication, and $n_x^A$ and $n_y^A$ are the numbers of $X$ and $Y$ operators in string $A$, respectively. $n_x^B$ and $n_y^B$ are defined analogously. $c_x$ is the number of times $\sigma_i^a = \sigma_i^b = x$; $c_y$ is defined accordingly. Noticing that $n_y^A$ and $n_y^B$ are both even, we can deduce that $[P_A, P_B]=0$.

Now we prove that an effective Hamiltonian constructed from a Trotter formula composed of circuits in \cref{Trotter circuits} preserves the number symmetry. As in \cref{lem: p_order_generator}, we consider the form of Hamiltonian $H=\sum_{\gamma} H_{\gamma}$. Using the number-symmetry-preserving representations in \cref{Trotter circuits}, we suppose that $[H_{\gamma},N]=0$ for all $\gamma$. We claim that for the effective Hamiltonian $\widetilde{H}$ corresponding to the given Trotter formula, $[\widetilde{H},N]=0$. By the BCH formula, $\widetilde{H}$ can be expressed as a convergent infinite sum of nested commutators of the local terms $H_\gamma$ for sufficiently small $\Delta\tau$. Thus, it suffices to show that every nested commutator of any length $\ell$ formed by the local terms $H_\gamma$, which we denote as $C_\ell$,  $[C_\ell,N]=0$. We prove this by induction.

For the trivial case $\ell=1$, $C_1$ is simply a single local term $H_\gamma$, and by definition,
\begin{equation*}
    [C_1, N] = [H_\gamma, N] = 0.
\end{equation*}
Suppose that any nested commutator of length up to $\ell$ commutes with $N$. Any nested commutator of length $\ell+1$, denoted $C_{\ell+1}$, can be written as the commutator of two smaller nested commutators, say $C_p$ and $C_q$, such that $C_{\ell+1} = [C_p, C_q]$ where $p, q \leq \ell$ and $p+q=\ell+1$. By the hypothesis, $[C_p, N] = 0$ and $[C_q, N] = 0$. Using the Jacobi identity, we obtain
\begin{equation*}
    [C_p, [C_q, N]] + [C_q, [N, C_p]] + [N, C_{\ell+1}] = 0\Longrightarrow [N, C_{\ell+1}] = 0.
\end{equation*}By induction, this completes the proof.

\end{proof}

\begin{lemma}[Bernstein inequality]\label{lem: bernstein ineq}
    Let $X_1, X_2, \dots, X_n$ be independent random variables. Suppose that $\abs{X_i-\mathbb{E}[X_i]}\le M$ almost surely for all $i$. Then, for the sum $S = \sum_{i=1}^n X_i$ and any $\epsilon > 0$,
    \begin{align}
        \mathbb{P}(\abs{S-\mathbb{E}[S]}\ge \epsilon)\le 2\exp\left(-\frac{\epsilon^2}{2 Var(S) + \frac{2}{3}M\epsilon}\right).
    \end{align}
\end{lemma}

We apply a result in  \cite{van2017quantum}  as the following lemma.
\begin{lemma}\label{lem: fourier approx exp}
    For $\delta,\epsilon\in(0,1)$ with $\delta=\mathcal{O}(1)$ and $\beta\ge 0$, there exists $\bm c\in\mathbb{C}^{2M+1}$ such that
    \begin{align}
        \max_{\abs{x}\le \frac{(1-\delta)\pi}{2}}\abs{e^{-\beta x}-\sum_{m=-M}^Mc_me^{imx}} \le \epsilon,
    \end{align}where
    \begin{align}
        M= \mathcal{O}\left(\beta +\log\frac{1}{\epsilon}\right),\quad \norm{\bm c}_1\le e^{\frac
        {\beta\pi}{2}}.
    \end{align}Subsequently, it follows
    \begin{align}
        \max_{\abs{x}\le \pi}\abs{\sum_{m=-M}^Mc_me^{imx}}\le  e^{\frac{\beta\pi}{2}}.
    \end{align}
     Moreover, the Fourier coefficient $\bm c$ can be efficiently calculated on a classical computer in time $\mathrm{poly}(\beta , \log\frac{1}{\epsilon})$.
\end{lemma}
\begin{proof}
By using the Taylor series of $e^{-\beta x}$ and the Stirling's approximation, we have 
\begin{align}
    \max_{\abs{x}\le 1}\abs{\sum_{j=0}^Ka_jx^j-e^{-\beta x}}\le \frac{\epsilon}{4},\quad K=\mathcal{O}\left(\beta + \log\frac{1}{\epsilon}\right), \quad \norm{\bm a}_1\le e^{\beta}.
\end{align}
Next, using \cite[Lemma 37]{van2017quantum}, we obtain a Fourier-series approximation such that with coefficient $\bm c\in \mathbb{C}^{2M+1}$,
\begin{align}
    \max_{\abs{x}\le 1-\delta}\abs{\sum_{m=-M}^Mc_me^{\frac{i\pi m}{2}x}-e^{-\beta x}}\le \epsilon,\quad M=2\left\lceil\log\left(\frac{4\norm{\bm a}_1}{\epsilon}\right)\frac{1}{\delta} \right\rceil,
\end{align}and
\begin{align}
    \norm{\bm c}_1\le \norm{\bm a}_1\le e^{\beta}.
\end{align}
By replacing $\beta$ with $\frac{\beta\pi}{2}$ and $x$ with $\frac{2x}{\pi}$ in the above, we prove the statement. Furthermore, since the classical time spent finding the coefficient $\bm c$ is $\mathrm{poly}(K, M , \log\frac{1}{\epsilon})$ by \cite[Lemma 37]{van2017quantum}, we verify the time complexity as stated.

\end{proof}
Using \cref{lem: fourier approx exp}, we evaluate the circuit sampling complexity of our approach. To this end, we require the following lemma.  
\begin{lemma}\label{lem: state GQSP}
    Let $\ket{\psi}=\sum_{\ell=1}^Sa_\ell\ket{\lambda_{i_\ell}}$ where $\ket{\lambda_{i_\ell}}$'s are eigenstates of $H$ corresponding to the eigenvalues $\lambda_{i_\ell}$'s. Assume a normalization factor $\lambda_S$ such that $\max_{\ell\in[S]}\abs{\frac{\lambda_{i_\ell}}{\lambda_S}}\le 1-\delta$. 
Then, for $\beta\ge0$ and a relative error $\epsilon_0=\mathcal{O}(1)$, there exists a Laurent polynomial $p_{\frac{\beta}{2}}$ satisfying
\begin{align}
\norm{p_{\frac{\beta}{2}}}_1
=
\mathcal{O}\left(e^{\frac{\beta\lambda_S}{2}}\right),
\end{align}
such that the ancilla success probability of the GQSP circuit in \cref{fig: GQSP}, after rescaling by
$\norm{p_{\frac{\beta}{2}}}_1$, approximates
$\bra{\psi}e^{-\beta H}\ket{\psi}$
within relative error $\epsilon_0$.   
\end{lemma}

\begin{proof}
By \cref{lem: fourier approx exp},
    we take a $p_{\frac{\beta}{2}}(e^{ix})$ such that $p_{\frac{\beta}{2}}(e^{ix})$ is real in $x$ by averaging with its conjugation and $\max_{\abs{x}\le \frac{\pi(1-\delta)}{2}}\abs{p_{\frac{\beta}{2}}\left(e^{ix}\right)-e^{-\frac{\beta\lambda_S}{\pi}x}} \le \epsilon$.
   We define $P(e^{ix}) = \frac{p_{\frac{\beta}{2}}\left(e^{ix}\right)}{\norm{p_{\frac{\beta}{2}}}_1}$, where $\norm{p_{\frac{\beta}{2}}}_1$ denotes the $\ell_1$-norm of Fourier coefficients of $p_{\frac{\beta}{2}}(x)$.  Note that 
    \begin{align}
        \norm{p_{\frac{\beta}{2}}}_1\le e^{\frac{\beta\lambda_S}{2}},
    \end{align}and 
    \begin{align}
      &\abs{p_{\frac{\beta}{2}}(e^{ix})^2 - e^{-\frac{2\beta\lambda_S}{\pi} x}} \\
      &\le 2e^{-\frac{\beta\lambda_S}{\pi}x}\abs{p_{\frac{\beta}{2}}\left(e^{ix}\right)-e^{-\frac{\beta\lambda_S}{\pi}x}} + \abs{p_{\frac{\beta}{2}}\left(e^{ix}\right)-e^{-\frac{\beta\lambda_S}{\pi}x}}^2\\
      &\le 2e^{\frac{\beta\lambda_S}{2}(1-\delta)}\epsilon +\epsilon^2,\quad \abs{x}\le \frac{\pi(1-\delta)}{2}\\
      &\le 3e^{\frac{\beta\lambda_S}{2}(1-\delta)}\epsilon,\quad \epsilon\in(0,1),
    \end{align}
    by squaring the identity, $p_{\frac{\beta}{2}}\left(x\right)=p_{\frac{\beta}{2}}\left(x\right)-e^{-\frac{\beta\lambda_S}{\pi}x}+e^{-\frac{\beta\lambda_S}{\pi}x}$, and using the triangle inequality.
    
    Since $\max_{\abs{x}\le \pi}\abs{P(x)}\le 1$,  \cite[Corollary 5, Theorem 6]{motlagh2024generalized} yields a circuit construction such that for some $\bm\theta, \bm\phi \in \mathbb{R}^{2D+1}$ and $\lambda\in \mathbb{R}$, 
\begin{align}
    \begin{pmatrix}
        P(U) &  *\\
        * & *
    \end{pmatrix}=\prod_{m=1}^D\left(R(\theta_{D+m},\phi_{D+m},\lambda)C U^\dag \right)\prod_{m=1}^D\left(R(\theta_m,\phi_m,\lambda)C U \right)R(\theta_0,\phi_0,\lambda).
\end{align}Here we set
\begin{align}
    U=e^{i\frac{\pi}{2}\frac{H}{\lambda_S}}.
\end{align}Note that the probability of ancilla state to be $0$ is $\bra{\psi}P\left(U\right)^2\ket{\psi}$. By classically multiplying this probability by $\norm{p_{\frac{\beta}{2}}}_1^2$, we can estimate $\bra{\psi}p_{\frac{\beta}{2}}\left(U\right)^2\ket{\psi}$. Using the error bound derived earlier, we then obtain
\begin{align}
    &\abs{\bra{\psi}p_{\frac{\beta}{2}}\left(U\right)^2\ket{\psi}
    -\bra{\psi}e^{-\beta H}\ket{\psi}} \nonumber \\
    &=\abs{
    \sum_{\ell=1}^S\abs{a_\ell}^2
    p_{\frac{\beta}{2}}\left(
    e^{i\frac{\pi \lambda_{i_\ell}}{2\lambda_S}}
    \right)^2
    -
    \sum_{\ell=1}^S\abs{a_\ell}^2 e^{-\beta \lambda_{i_\ell}}
    } \nonumber \\
    &\le
    \sum_{\ell=1}^S\abs{a_\ell}^2
    \abs{
    p_{\frac{\beta}{2}}\left(
    e^{i\frac{\pi \lambda_{i_\ell}}{2\lambda_S}}
    \right)^2
    -
    e^{-\beta \lambda_{i_\ell}}
    } \nonumber \\
    &\le \sum_{\ell=1}^S\abs{a_\ell}^2
    \left(
    3e^{\frac{\beta\lambda_S}{2}(1-\delta)}\epsilon
    \right) \nonumber \\
    &=3e^{\frac{\beta\lambda_S}{2}(1-\delta)}\epsilon
    \le 3e^{\frac{\beta\lambda_S}{2}}\epsilon.
\end{align}
By taking $\epsilon=\min\{\frac{1}{2},\frac{\bra{\psi}e^{-\beta H}\ket{\psi}}{3e^{\frac{\beta\lambda_S}{2}}}\epsilon_0\}$,
we estimate the desired quantity to relative error $\epsilon_0$. 
\end{proof}

Using this lemma, we derive the required query complexity and circuit sampling complexity as follows.
\begin{corollary}\label{cor: main}
    Assume a normalization factor $\lambda_{L,s}$ such that $\norm{\frac{H(L,s)}{\lambda_{L,s}}}\le 1-\delta$ for some $\delta \in (0,1)$. Further assume access to the ideal input models, $e^{i\frac{\pi H}{2\lambda_{L,s}}}$ and $e^{-i\frac{\pi H}{2\lambda_{L,s}}}$. Estimating the canonical partition function, $Z_{L,s}(\beta)=\tr(e^{-\beta H(L,s)})$, to relative error $\epsilon$ with failure probability at most $\eta$  using the GQSP circuit in \cref{fig: GQSP} requires
    \begin{itemize}
        \item  $\mathcal{O}\left(\frac{1}{\delta}\left(\ln\left(\frac{1}{\epsilon}\right)+\beta\lambda_{L,s} \right)\right)$ query complexity to $e^{i\frac{\pi H}{2\lambda_{L,s}}}$ and $e^{-i\frac{\pi H}{2\lambda_{L,s}}}$ per circuit sample
        \item  $\mathcal{O}\left(\frac{d_s e^{\beta\lambda_{L,s}}}{Z_{L,s}(\beta)\epsilon^2}\log\frac{1}{\eta}\right)$ circuit sampling complexity
    \end{itemize}
    where $d_s$ is the dimension of the symmetry sector ($d_s=\binom{L}{k}$ or $\binom{L}{k}\binom{L}{k'}$).
\end{corollary}

\begin{proof}
Let $\omega_{L,s} = \frac{1}{d_s}\Pi_s$ from \cref{def: Dicke} be the initial maximally mixed state within this sector. Let $Z_{\text{poly}, s} = \tr\left(\Pi_s p_{\frac{\beta}{2}}\left(e^{i\frac{\pi H }{2\lambda_{L,s}}}\right)^2\right)$ be the polynomial approximation of the canonical partition function.

To estimate the partition function to a relative error $\epsilon$, we require the polynomial approximation to satisfy $\abs{Z_{\text{poly}, s} - Z_{L,s}(\beta)} \le \frac{\epsilon}{3} Z_{L,s}(\beta)$. Following the proof technique for \cref{lem: state GQSP}, this approximation error is bounded by
\begin{align}
    \abs{Z_{\text{poly}, s} - Z_{L,s}(\beta)} \le d_s \cdot 3e^{\frac{\beta\lambda_{L,s}}{2}} \epsilon_{\text{poly}},
\end{align}
where $\epsilon_{\text{poly}}$ is the error bound from \cref{lem: fourier approx exp} without the normalization factor $\lambda_{L,s}$. Thus, we choose $\epsilon_{\text{poly}}$ such that
\begin{align}
    \epsilon_{\text{poly}} = \frac{\epsilon Z_{L,s}(\beta)}{9 d_s e^{\frac{\beta\lambda_{L,s}}{2}}}.
\end{align}
By \cref{lem: fourier approx exp}, achieving this $\epsilon_{\text{poly}}$ requires a Laurent polynomial of degree scaling as
\begin{align}
    \mathcal{O}\left(\frac{1}{\delta}\left(\beta\lambda_{L,s} + \ln\frac{1}{\epsilon_{\text{poly}}}\right)\right) 
    = \mathcal{O}\left(\frac{1}{\delta}\left(\ln\frac{1}{\epsilon} + \beta\lambda_{L,s} + \ln\frac{d_s}{Z_{L,s}(\beta)}\right)\right).
\end{align}Similar to the proof of Observation \ref{obs 2}, Jensen's inequality yields
 the stated query complexity. 

When executing the GQSP circuit on $\omega_{L,s}$, the probability of measuring the ancilla state as $\ket{0}$ is
\begin{align}
    P_s^{\text{poly}}(0) = \tr\left(\omega_{L,s} P\left(e^{i\frac{\pi H}{2\lambda_{L,s}}}\right)^2\right) = \frac{Z_{\text{poly}, s}}{d_s \norm{p_{\frac{\beta}{2}}}_1^2}.
\end{align}
We estimate this probability using the sample mean estimator $\widehat{P}_s^{\text{poly}}(0) = \frac{N_{\ket{0}}}{N_{\mathrm{shots}}}$, where $\mathbb{E}[\widehat{P}_s^{\text{poly}}(0)]=P_s^{\text{poly}}(0)$. 

By the properties of a binomial distribution, the variance is bounded by $Var(\hat{P}_s^{\text{poly}}(0))\le P_s^{\text{poly}}(0)$. To ensure that our estimate $\widetilde{Z}_{\text{poly}, s}$ satisfies $\abs{\widetilde{Z}_{\text{poly}, s} - Z_{\text{poly}, s}} \le \frac{\epsilon}{3} Z_{L, s}(\beta)$, we must estimate $P_s^{\text{poly}}(0)$ to an additive error of $\epsilon_P = \frac{\epsilon Z_{L,s}(\beta)}{3 d_s \norm{p_{\frac{\beta}{2}}}_1^2}$. 
By \cref{lem: bernstein ineq} and assuming $\epsilon \le 1$, achieving this precision with failure probability at most $\eta$ requires a number of samples bounded by
\begin{align}
    N_{\mathrm{shots}} = \mathcal{O}\left(\frac{P_s^{\text{poly}}(0)}{\epsilon_P^2}\log \frac{1}{\eta}\right) = \mathcal{O}\left( \frac{Z_{\text{poly}, s} \cdot d_s \norm{p_{\frac{\beta}{2}}}_1^2}{\epsilon^2 Z_{L,s}(\beta)^2} \log \frac{1}{\eta} \right).
\end{align}
Since $Z_{\text{poly}, s} \le Z_{L,s}(\beta)\left(1+\frac{\epsilon}{3}\right) = \mathcal{O}(Z_{L,s}(\beta))$ and  $\norm{p_{\frac{\beta}{2}}}_1^2 = \mathcal{O}(e^{\beta\lambda_{L,s}})$, we achieve
\begin{align}
    N_{\mathrm{shots}} = \mathcal{O}\left(\frac{d_s e^{\beta\lambda_{L,s}}}{Z_{L,s}(\beta)\epsilon^2}\log\frac{1}{\eta}\right).
\end{align}
Finally, using the triangle inequality, we obtain
\begin{align}
    \abs{\widetilde{Z}_{\text{poly}, s} - Z_{L,s}(\beta)} &\le \abs{\widetilde{Z}_{\text{poly}, s} - Z_{\text{poly}, s}} + \abs{Z_{\text{poly}, s} - Z_{L,s}(\beta)} \nonumber \\
    &\le \frac{\epsilon}{3} Z_{L,s}(\beta) + \frac{\epsilon}{3} Z_{L,s}(\beta) \le \epsilon Z_{L,s}(\beta).
\end{align}
This completes the proof.
\end{proof}

Finally, the following two lemmas will be used for the task of estimating observables in Problem \ref{def: problem}.
\begin{lemma}\label{lem: thermal_observable_error}
For any observable $O$, the following inequality holds
\begin{align}
    \abs{\tr(O\widetilde{\rho}) - \tr(O\rho)} \le 2\beta \norm{O}_2 \norm{\widetilde{H}-H}_2,
\end{align}where $\rho = \frac{e^{-\beta H}}{\tr(e^{-\beta H})}$ and $\widetilde{\rho} = \frac{e^{-\beta \widetilde{H}}}{\tr(e^{-\beta \widetilde{H}})}$.
\end{lemma}

\begin{proof}
We define a one-parameter family of Hamiltonians $H_\tau = H + \tau E$ with corresponding thermal states,
\begin{equation}
    \rho_\tau = \frac{e^{-\beta H_\tau}}{\tr(e^{-\beta H_\tau})}, \quad \text{for } \tau \in [0, 1].
\end{equation}
Let us denote $E = \widetilde{H} - H$. Differentiating $\rho_\tau$ with respect to $\tau$ using the Duhamel formula yields
\begin{equation}
    \frac{d\rho_\tau}{d\tau} = -\beta \int_0^1 \rho_\tau^{1-u} \big(E - \tr(\rho_\tau E)I\big) \rho_\tau^u \, du.
\end{equation}
By using the Schatten H\"older inequality, we obtain that for any $A$,
\begin{equation}
    \norm{\rho_\tau^{1-u} A \rho_\tau^u}_1\le \norm{\rho_\tau^{1-u}}_{\frac{1}{1-u}} \norm{A\rho_\tau^u}_{\frac{1}{u}}  \le \norm{\rho_\tau^{1-u}}_{\frac{1}{1-u}} \norm{A}_2 \norm{\rho_\tau^u}_{\frac{1}{u}}.
\end{equation}
Since $\rho_\tau$ is a density matrix,  $\norm{\rho_\tau^p}_{1/p} = \left(\tr(\rho_\tau)\right)^p = 1$. Then, the inequality simplifies to
\begin{equation}
    \norm{\rho_\tau^{1-u} A \rho_\tau^u}_1 \le \norm{A}_2.
\end{equation}
Setting $A = E - \operatorname{Tr}(\rho_\tau E)I$, we bound its spectral norm as
\begin{equation}
    \norm{E - \tr(\rho_\tau E)I}_2 \le 2\norm{E}_2,
\end{equation}by noticing $\abs{\tr(\rho_\tau E)}\le \norm{E}_2$ and using the triangle inequality. Substituting this back into the integral gives
\begin{equation}
    \norm{\frac{d\rho_\tau}{d\tau}}_1  \le 2\beta \norm{E}_2, 
\end{equation}and we obtain 
\begin{equation}
    \norm{\widetilde{\rho} - \rho}_1 = \norm{\int_0^1 \frac{d\rho_\tau}{d\tau} d\tau}_1 \le 2\beta \norm{\widetilde{H} - H}_2.
\end{equation}
Finally, we arrive at
\begin{align}
    \abs{\tr(O\widetilde{\rho}) - \tr(O\rho)} 
    \le \norm{O}_2 \norm{\widetilde{\rho} - \rho}_1 \le 2\beta \norm{O}_2 \norm{\widetilde{H} - H}_2,
\end{align}which completes the proof.
\end{proof}

As a standard approach, we estimate an observable of the form \eqref{eq: observable} by measuring each Pauli operator. This results in the following sampling complexity.
\begin{lemma}\label{lem: observable}
    Given access to state $\rho$, for an observable in \eqref{eq: observable}, an estimate $\widehat{o}$ satisfying $\abs{ \widehat{o} - \tr(O\rho)} \le \epsilon$ 
    with probability at least $1-\eta$ can be obtained using $\mathcal{O}\left( \frac{\left(\sum_a \abs{c_a}\right)^2}{\epsilon^2} \log\frac{1}{\eta} \right)$ circuit sampling complexity.
\end{lemma}

\begin{proof}
    Suppose we 
    measure each $P_a$ independently using $N_a$ samples of $\rho$, and let $\widehat{P}_a$ denote the corresponding sample mean. We define 
$\widehat o=\sum_a c_a\widehat P_a
$.  Since each measurement outcome lies in $\{-1, 1\}$, applying Hoeffding's inequality to the independent measurement outcomes  yields
    \begin{align*}
        \mathbb{P}\left( \abs{ \widehat{o} - \tr(O\rho)} \ge \epsilon \right)
        \le 2\exp\left( -\frac{\epsilon^2}{2\sum_a \frac{\abs{c_a}^2}{N_a}} \right).
    \end{align*}We prove the statement by
    choosing $N_a \propto \abs{c_a}$ yields $\sum_a \frac{\abs{c_a}^2}{N_a} = \frac{\left(\sum_a \abs{c_a}\right)^2}{N}$, where $N = \sum_a N_a$, the total sampling complexity.
\end{proof}

\begin{center}
    \textbf{The proofs of \cref{thm: main 1} and \cref{thm: main 2}}
\end{center}\label{Appendix: thm proofs}

\begin{proof}[Proof of \cref{thm: main 1}]
In \cref{alg: cano}, we replace the ideal Hamiltonian simulation in \cref{fig: GQSP} with the first-order Trotter circuit composed of Pauli exponentials from \cref{Trotter circuits}. Since we set a $s$-sector mixed Dicke state in \cref{def: Dicke} as the initial state, we can consider the corresponding effective Hamiltonian $\widetilde{H}(L,s)$ to the Trotter formula. By imposing the error condition
\begin{align}\label{eq: H error}
    \norm{\widetilde{H}(L,s) - H(L,s)}_2  \le \min\left[\frac{\delta\lambda_{L,s}}{2}, \frac{\log(1+\epsilon)}{\beta}\right],
\end{align}
and applying \cref{lem: p_order_generator}, we obtain $\Omega\left(\Lambda_{H(L,s)}^2\max\left[\frac{1}{\lambda_{L,s}\delta},\; \frac{\beta }{\log(1+\epsilon)}\right]\right)$ Trotter steps as stated in \cref{thm: main 1}. Here $\Lambda_{H(L,s)}$ is derived from the definition of $\Lambda_H$ in \cref{lem: p_order_generator} as we consider $H(L,s)$ not $H$. That is, $H(L,s)=\Pi_{s}H\Pi_s = \sum_{\gamma} \Pi_s H_\gamma \Pi_s$, and we define 
\begin{align}\label{eq: lambda_H_s}
    \Lambda_{H(L,s)}:= \sum_\gamma \norm{\Pi_s H_\gamma \Pi_s}.
\end{align}
With the above Trotter step number, the gate complexity is determined by multiplying the total query complexity of Hamiltonian simulations in \cref{cor: main} and $\mathcal{O}(L^4)$ gate complexity for each Trotter step for the model \eqref{ex: models}. Furthermore, we verify that this Trotter step number guarantees the desired precision for the partition function and observable estimation. 

Let $\widetilde{H}(L,s)$ denote the restriction of $\widetilde{H}$ to the particle sector $s$. Let $\widehat{Z}_{\text{poly}, s} = \tr\left(\Pi_s p_{\frac{\beta}{2}}\left(e^{i\frac{\pi \widetilde{H}}{2\lambda_{L,s}}}\right)^2\right)$ be the polynomial approximation of the partition function that incorporates the Trotter error. We observe that
\begin{align}\label{eq: Z analysis}
    \abs{\widehat{Z}_{\text{poly}, s} -Z_{L,s}(\beta)} &\le \abs{\widehat{Z}_{\text{poly}, s} -\tr(e^{-\beta \widetilde{H}(L,s)})} + \abs{\tr(e^{-\beta \widetilde{H}(L,s)})-\tr(e^{-\beta H(L,s)})} \nonumber \\
    &\le \epsilon \tr(e^{-\beta \widetilde{H}(L,s)}) + \abs{\sum_{j=1}^{d_s}\left(e^{-\beta(\lambda_j+\epsilon_j)} - e^{-\beta\lambda_j}\right)} \nonumber \\
    &\le \epsilon \tr(e^{-\beta H(L,s)}) + (1+\epsilon)\abs{\sum_{j=1}^{d_s}\left(e^{-\beta(\lambda_j+\epsilon_j)} - e^{-\beta\lambda_j}\right)} \nonumber \\
    &\le \epsilon \tr(e^{-\beta H(L,s)}) + (1+\epsilon)\sum_{j=1}^{d_s}e^{-\beta\lambda_j}\abs{e^{-\beta\epsilon_j} - 1} \nonumber \\
    &\le \left(\epsilon+\epsilon(1+\epsilon)\right)Z_{L,s}(\beta) = (2\epsilon + \epsilon^2)Z_{L,s}(\beta).
\end{align}The third inequality follows from the triangle inequality, and the final inequality holds by Weyl's inequality and the numeric identity $\abs{e^{ax}-1} \le \epsilon$ for $\abs{x} \le \frac{\log(1+\epsilon)}{a}$, guaranteed by our condition on the eigenvalue perturbations $\epsilon_j$ from \eqref{eq: H error}. 

Recalling that $\widehat{Z}_{L,s}(\beta)$ in \cref{alg: cano} is the statistical estimate of $\widehat{Z}_{\text{poly}, s}$ obtained from measurement data, we apply the technique for circuit sampling complexity from \cref{cor: main}. Specifically, estimating this quantity requires $\mathcal{O}\left(\frac{d_s e^{\beta\lambda_{L,s}}}{\widehat{Z}_{\text{poly},s}\epsilon^2}\log\frac{1}{\eta}\right)$ samples. Since $\widehat{Z}_{\text{poly},s}$ approximates $Z_{L,s}(\beta)$ to within a relative error of $\mathcal{O}(\epsilon)$, we obtain the desired sampling complexity $\mathcal{O}\left(\frac{d_s e^{\beta\lambda_{L,s}}}{Z_{L,s}(\beta)\epsilon^2}\log\frac{1}{\eta}\right)$.

Now we prove the statement regarding the estimation of canonical thermal averages. By \cref{cor: main}, the success probability of measuring the ancilla in $\ket{0}$ to prepare the approximate Gibbs state $\widetilde{\rho}$ scales as $\Theta\left(\frac{Z_{L,s}(\beta)}{d_s e^{\beta\lambda_{L,s}}}\right)$. In other words, we obtain an output of $\widetilde{\rho}$ with $\Theta\left(\frac{d_s e^{\beta\lambda_{L,s}}}{Z_{L,s}(\beta)}\right)$ circuit samplings. Consequently, estimating the expectation value of an observable $O$ via sample means in \cref{alg: cano} requires $\mathcal{O}\left(\left(\sum_a \abs{c}_a\right)^2\frac{d_s e^{\beta\lambda_{L,s}}}{Z_{L,s}(\beta)\epsilon^2}\log\frac{1}{\eta}\right)$ samples to reach an additive precision $\epsilon$, according to \cref{lem: observable}. The error bound on the observable expectation immediately follows from the  condition \eqref{eq: H error} evaluated via \cref{lem: thermal_observable_error}.
\end{proof}

\begin{proof}[Proof of \cref{thm: main 2}]
The proof follows straightforwardly from the arguments above by replacing the restricted initial state with the maximally-mixed state over the full Hilbert space, $\rho = \frac{I}{2^L}$, and utilizing an arbitrary first-order Trotter circuit. 
\end{proof}

\begin{center}
    \textbf{Further numerical tests}
\end{center}

Following standard conventions \cite{leblanc2015solutions}, in the Hamiltonian \eqref{eq: model hubbard}, we set $t=1$ as the unit of energy and take $t'=-0.2$ as in our numerical results. We consider two interaction regimes: $U=4$ (weak-coupling) and $U=8$ (intermediate-coupling). The geometry is an $L_x \times L_y$ torus with $L_x=2$ and $L_y \in \{2, 4, 6, 8\}$ under periodic boundary conditions.

The sampling theorem in \cite[Supplemental Material, Theorem 4.11]{qprk-k7wn} applies when
\begin{equation}\label{eq: temp from qprk}
    \beta J \le \frac{1}{50R d_{\mathrm{ov}}^2},
\end{equation}
where \(R\) is the maximum Majorana support size, \(d_{\mathrm{ov}}\) is the maximum number of Hamiltonian terms whose Majorana supports overlap with that of a given term, and \(J\) is the maximum magnitude of a Majorana-string coefficient. We emphasize that \eqref{eq: temp from qprk} is a sufficient condition for the particular classical sampling algorithm of \cite{qprk-k7wn}, rather than a classical-hardness threshold. 

Expressing the spin-orbital operators in terms of Majoranas,
\[
c_{i\sigma}
=
\frac{\gamma_{i\sigma,1}+i\gamma_{i\sigma,2}}{2},
\]
the hopping terms give 2-Majorana strings, whereas the on-site interaction contains a 4-Majorana string. Therefore, $R=4$. We next determine \(d_{\mathrm{ov}}\) directly from overlaps of Majorana supports. Below, $d_{\mathrm{quartic}}$, $d_{\mathrm{quadratic}}$, and $d_{\mathrm{hopping}}$ denote the maximum overlap degrees of an on-site quartic Majorana term, an on-site quadratic Majorana term, and a hopping Majorana term, respectively, so that
$d_{\mathrm{ov}}=\max\{d_{\mathrm{quartic}},,d_{\mathrm{quadratic}},,d_{\mathrm{hopping}}\}$.

Let $z=z_1+z_2$ denote the number of distinct nearest- and next-nearest-neighbor sites connected to a given spatial site. An on-site quartic Majorana string overlaps with itself, the two on-site quadratic strings, and the \(4z\) hopping Majorana strings incident on that site. Hence,
\begin{equation}
d_{\mathrm{quartic}}\le 3+4z.
\end{equation}
An on-site quadratic string overlaps with itself, the on-site quartic string, and \(2z\) hopping strings involving the same spin, yielding
\begin{equation}
d_{\mathrm{quadratic}}\le 2+2z.
\end{equation}
Finally, a hopping Majorana string overlaps with the corresponding on-site quadratic and quartic strings at its two endpoints and with the hopping strings sharing either of its two Majorana operators, giving
\begin{equation}
d_{\mathrm{hopping}}\le 2z+3.
\end{equation}
Thus, the on-site quartic string determines the maximum overlap degree. For our model \eqref{eq: model hubbard} with \(2\times L_y\) geometries (\(L_y=4,6,8\)), periodic boundary conditions with duplicate bonds removed give \(z_1=3\) and \(z_2=2\), and hence
\begin{equation}
d_{\mathrm{ov}}\le 3+4(3+2)=23.
\end{equation}

After combining the quadratic terms originating from the interaction and chemical potential, the maximum Majorana-string coefficient is
\begin{equation}
J=
\max\left\{
\frac{|t|}{2},
\frac{|t'|}{2},
\frac{U}{4},
\left|\frac{\mu}{2}-\frac{U}{4}\right|
\right\}.
\end{equation}
For \(t=1\), \(t'=-0.2\), and the chemical potentials reported in \cref{tab:mus_combined}, which satisfy \(0<\mu<U\), this gives
\begin{equation}
J=
\begin{cases}
1, & U=4,\\
2, & U=8.
\end{cases}
\end{equation}
Using \(R=4\) and the uniform value \(d_{\mathrm{ov}}=23\), Eq.~\eqref{eq: temp from qprk} therefore gives
\begin{equation}
T=\beta^{-1}\ge 50Rd_{\mathrm{ov}}^2J
=105800J,
\end{equation}
corresponding to
\begin{equation}
T\ge105800 \quad (U=4),
\qquad
T\ge211600 \quad (U=8).
\end{equation}
Again, these values characterize a rigorous sufficient temperature regime for the algorithm in \cite{qprk-k7wn} and do not imply classical hardness below these temperatures. Table~\ref{tab:mus_combined} summarizes these temperature thresholds with chemical potentials computed in each interaction regime.

\begin{table}[h]
\centering
\caption{Values of the chemical potential $\mu$ for varying particle densities $n$ and temperatures $T$, alongside the theoretical temperature lower bounds $T = \frac{1}{\beta}$ for the sampling algorithm in \cite{qprk-k7wn}.}
\label{tab:mus_combined}
\vspace{0.2cm}
\begin{tabular}{c c | c c c | c c c}
\toprule
\multirow{4}{*}{$n$} & \multirow{4}{*}{$T$} & \multicolumn{6}{c}{\textbf{Chemical potential $\mu$}} \\ \cline{3-8} 
 & & \multicolumn{6}{c}{\textbf{Geometry}} \\ \cline{3-8} 
 & & \multicolumn{3}{c|}{$U=4$ \quad ($J=1$)} & \multicolumn{3}{c}{$U=8$ \quad ($J=2$)} \\ \cline{3-8} 
 & & $2 \times 4$ & $2 \times 6$ & $2 \times 8$ & $2 \times 4$ & $2 \times 6$ & $2 \times 8$ \\
\hline
\multirow{4}{*}{$0.875$} 
 & $0.01$ & $0.7348$ & $0.8664$ & $0.7764$ & $1.4751$ & $1.4213$ & $1.4368$ \\
 & $0.05$ & $0.7272$ & $0.8700$ & $0.7951$ & $1.4939$ & $1.4175$ & $1.4499$ \\
 & $0.1$  & $0.7385$ & $0.8929$ & $0.8221$ & $1.5186$ & $1.4324$ & $1.4642$ \\
 & $0.5$  & $0.9859$ & $0.9781$ & $0.9785$ & $1.5770$ & $1.5573$ & $1.5618$ \\
\hline
\multirow{4}{*}{$0.9$} 
 & $0.01$ & $0.7482$ & $0.8715$ & $0.9160$ & $1.4810$ & $1.4345$ & $1.5921$ \\
 & $0.05$ & $0.7658$ & $0.8987$ & $0.8868$ & $1.5213$ & $1.4671$ & $1.5785$ \\
 & $0.1$  & $0.8018$ & $0.9443$ & $0.9114$ & $1.5695$ & $1.5137$ & $1.5767$ \\
 & $0.5$  & $1.1459$ & $1.1338$ & $1.1327$ & $1.7714$ & $1.7614$ & $1.7656$ \\
\hline
\multirow{4}{*}{$0.925$} 
 & $0.01$ & $0.7572$ & $0.8759$ & $0.9277$ & $1.4870$ & $1.4585$ & $1.5970$ \\
 & $0.05$ & $0.8018$ & $0.9272$ & $0.9415$ & $1.5499$ & $1.5184$ & $1.6179$ \\
 & $0.1$  & $0.8654$ & $0.9982$ & $0.9924$ & $1.6235$ & $1.5956$ & $1.6619$ \\
 & $0.5$  & $1.3204$ & $1.3001$ & $1.2980$ & $1.9924$ & $1.9927$ & $1.9963$ \\
\hline
\multirow{4}{*}{$0.95$} 
 & $0.01$ & $0.7652$ & $0.8807$ & $0.9379$ & $1.4938$ & $1.4743$ & $1.6011$ \\
 & $0.05$ & $0.8384$ & $0.9592$ & $0.9901$ & $1.5823$ & $1.5689$ & $1.6500$ \\
 & $0.1$  & $0.9344$ & $1.0607$ & $1.0735$ & $1.6860$ & $1.6815$ & $1.7402$ \\
 & $0.5$  & $1.5127$ & $1.4787$ & $1.4761$ & $2.2666$ & $2.2771$ & $2.2801$ \\
\hline
\multirow{4}{*}{$0.975$} 
 & $0.01$ & $0.7749$ & $0.8874$ & $0.9495$ & $1.5029$ & $1.4866$ & $1.6060$ \\
 & $0.05$ & $0.8859$ & $1.0038$ & $1.0466$ & $1.6274$ & $1.6251$ & $1.6896$ \\
 & $0.1$  & $1.0270$ & $1.1493$ & $1.1754$ & $1.7745$ & $1.7869$ & $1.8375$ \\
 & $0.5$  & $1.7233$ & $1.6694$ & $1.6673$ & $2.6772$ & $2.6974$ & $2.6995$ \\
\hline
\multirow{4}{*}{$1$} 
 & $0.01$ & $1.5385$ & $1.1671$ & $1.2101$ & $5.9802$ & $2.9838$ & $2.1328$ \\
 & $0.05$ & $1.8725$ & $1.8406$ & $1.7852$ & $3.0724$ & $3.1609$ & $3.0281$ \\
 & $0.1$  & $1.9082$ & $1.8433$ & $1.8020$ & $3.9015$ & $3.9231$ & $3.9129$ \\
 & $0.5$  & $1.9439$ & $1.8672$ & $1.8671$ & $3.9666$ & $3.9846$ & $3.9780$ \\
\hline
\multicolumn{2}{c|}{\textbf{Lower bound ($T \ge$)}} & \multicolumn{3}{c|}{$\mathbf{105,800}$} & \multicolumn{3}{c}{$\mathbf{211,600}$} \\
\hline
\end{tabular}
\end{table}

\begin{figure}[htbp]
    

    \begin{center}
   \begin{subfigure}[b]{0.45\textwidth}
    \centering
    \includegraphics[width=\textwidth]{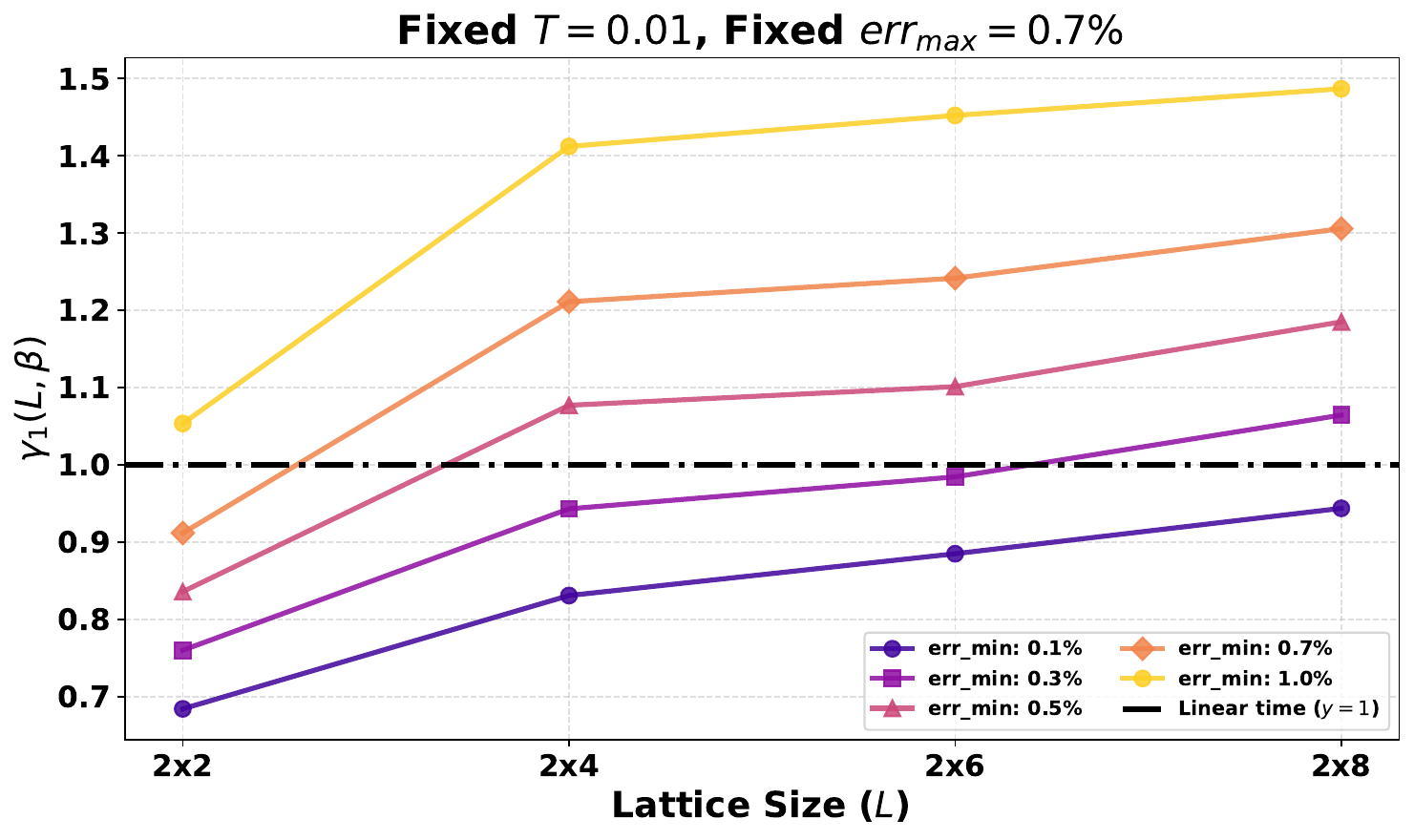}
    \end{subfigure}
    \begin{subfigure}[b]{0.45\textwidth}
    \centering
    \includegraphics[width=\textwidth]{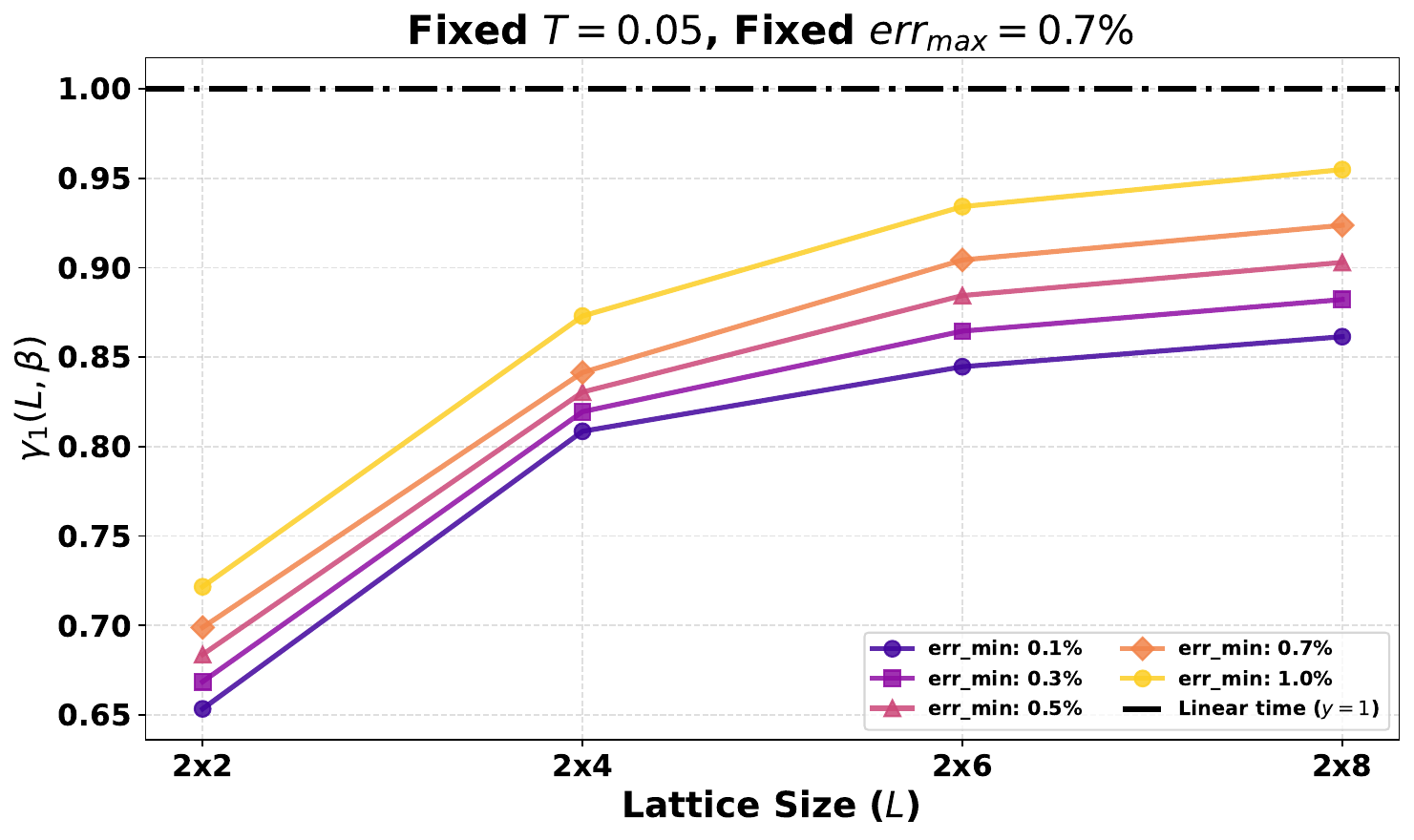}
    \end{subfigure}
    \begin{subfigure}[b]{0.45\textwidth}
    \centering
    \includegraphics[width=\textwidth]{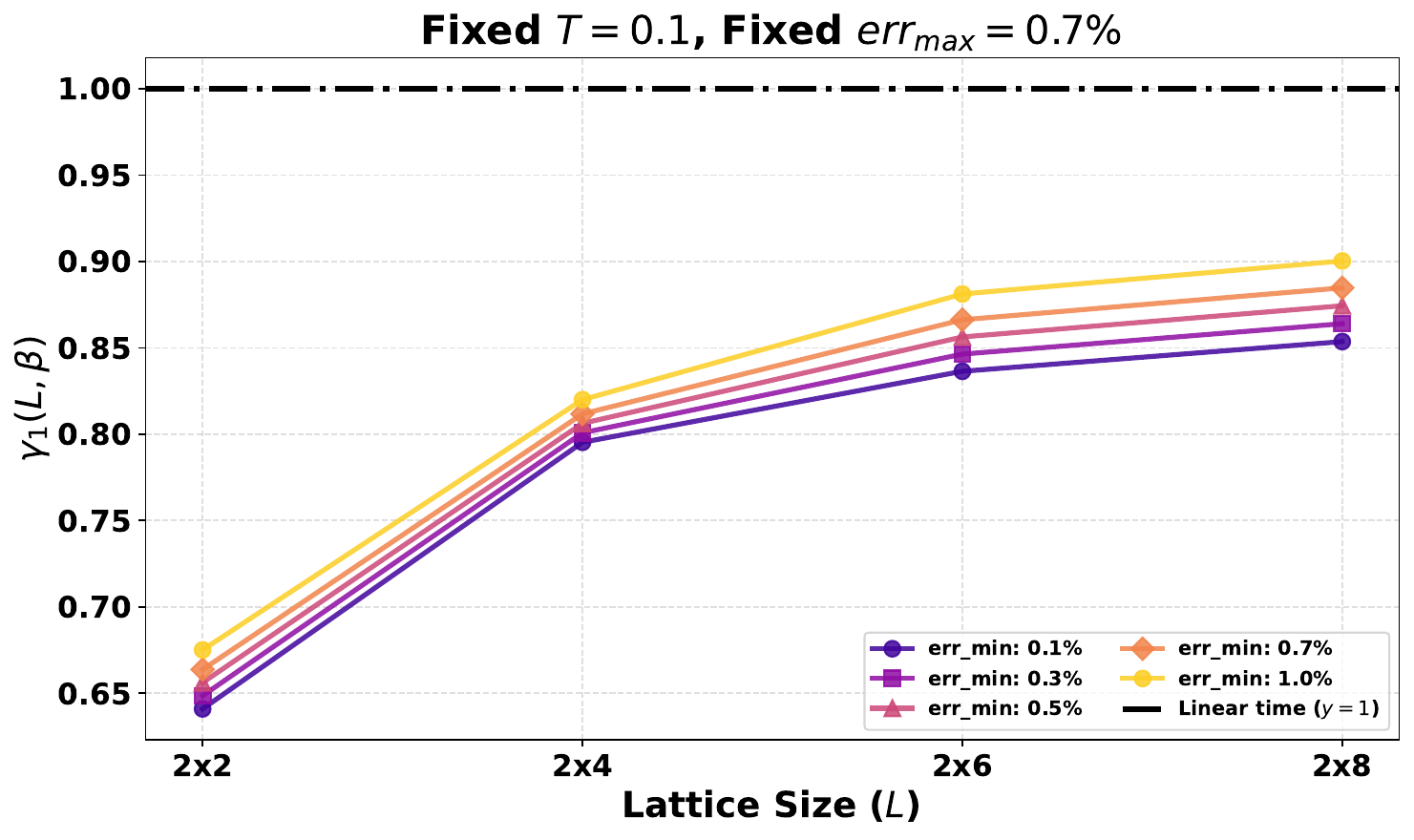}
    \end{subfigure}
    \begin{subfigure}[b]{0.45\textwidth}
    \centering
    \includegraphics[width=\textwidth]{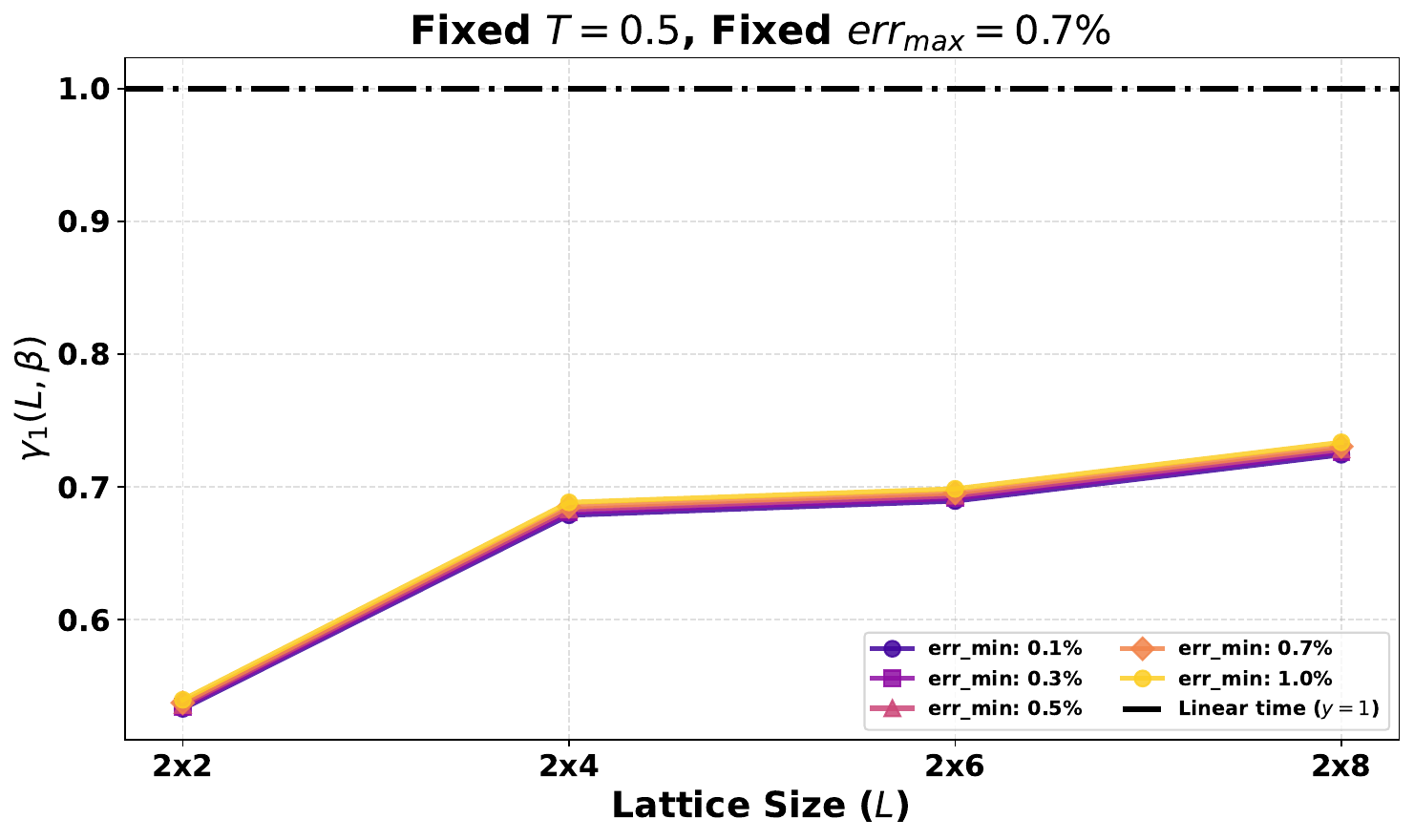}
    \end{subfigure}

\end{center}

    \caption{Simulation results illustrating the impact of the ground state energy relative error ($\epsilon_{\min}$) on the figure of merit $\gamma_1(L,\beta)$ for the proposed grand-canonical scheme in the weak-coupling regime ($U=4$). The figure of merit $\gamma_1$ is plotted as a function of system size ($L$) for varying relative errors $\epsilon_{\min} \in \{0.001, 0.003, 0.005, 0.007, 0.01\}$.   Results are obtained for different temperatures $T=0.01,0.05,0.1,0.5$.}
    \label{fig: g2err1}
\end{figure}

\begin{figure}[htbp]
    

    \begin{center}
   \begin{subfigure}[b]{0.45\textwidth}
    \centering
    \includegraphics[width=\textwidth]{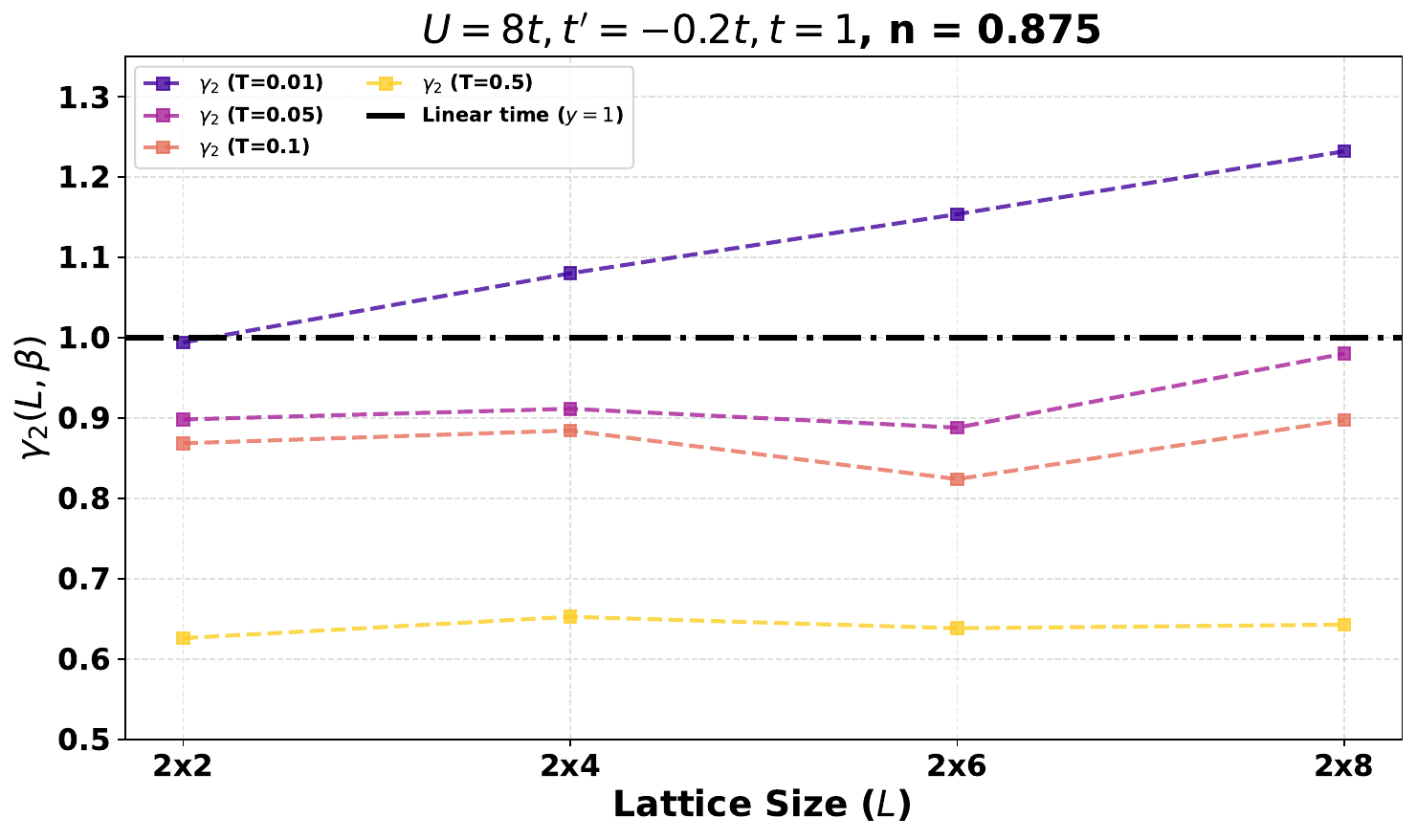}
    \end{subfigure}
    \begin{subfigure}[b]{0.45\textwidth}
    \centering
    \includegraphics[width=\textwidth]{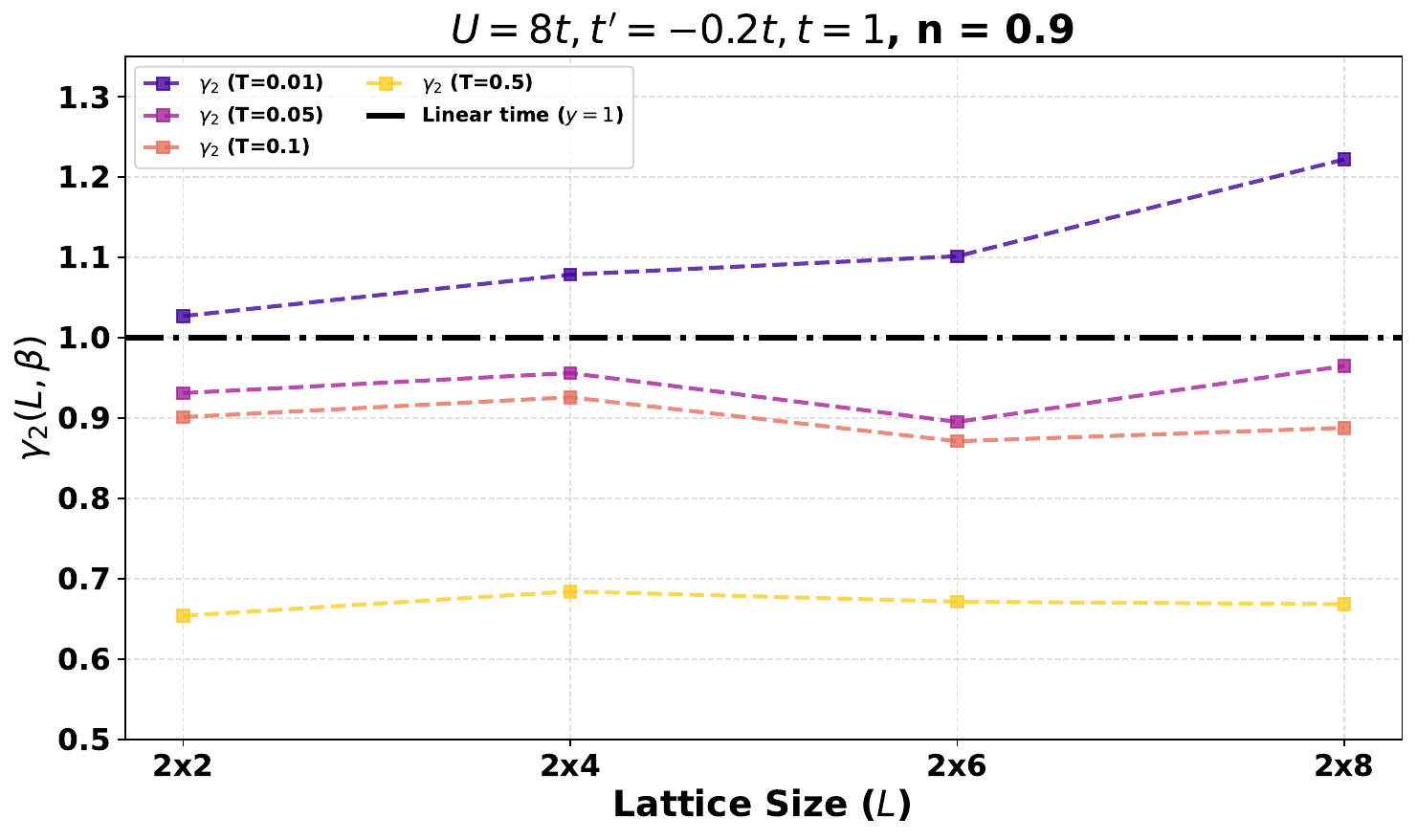}
    \end{subfigure}
    \begin{subfigure}[b]{0.45\textwidth}
    \centering
    \includegraphics[width=\textwidth]{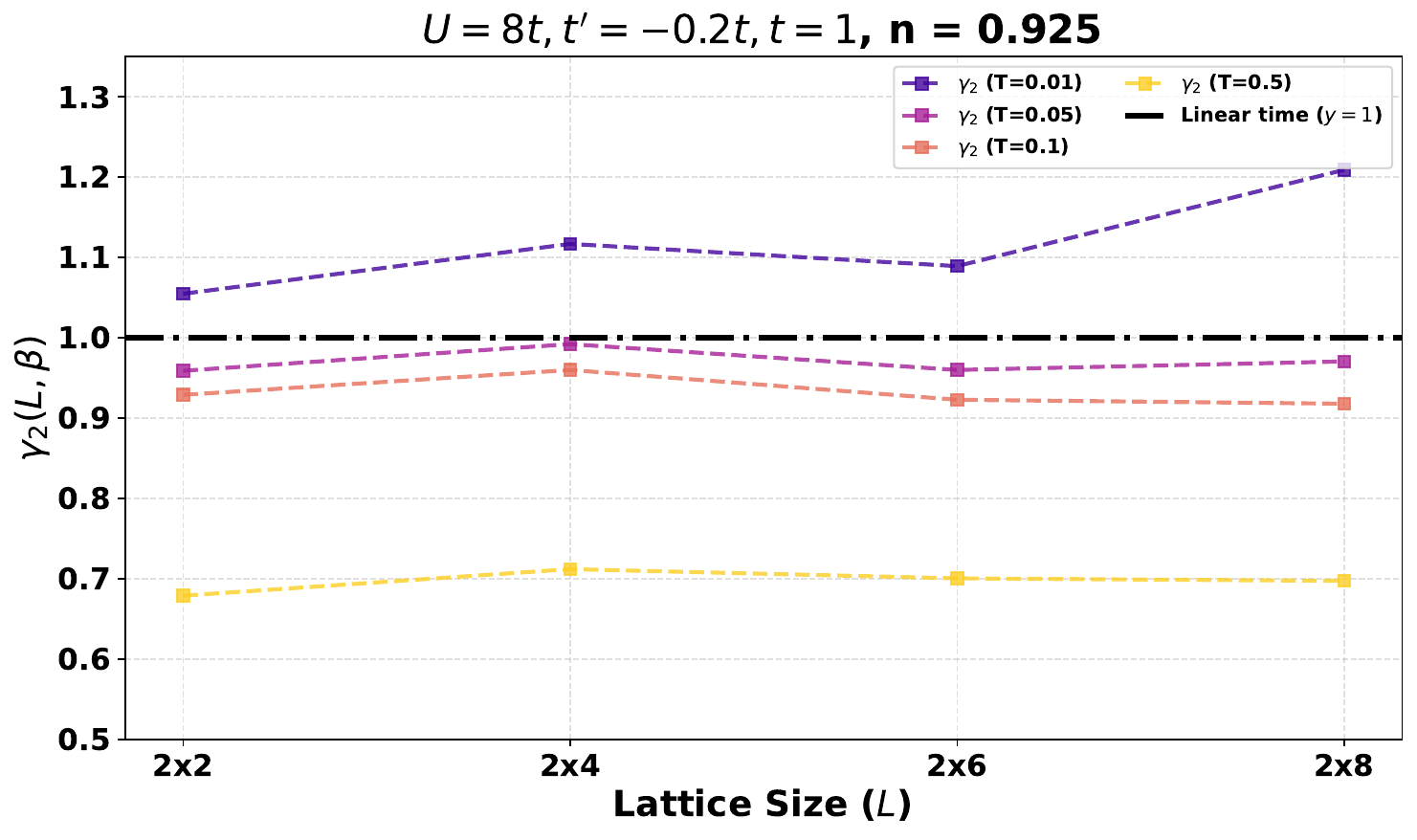}
    \end{subfigure}
    \begin{subfigure}[b]{0.45\textwidth}
    \centering
    \includegraphics[width=\textwidth]{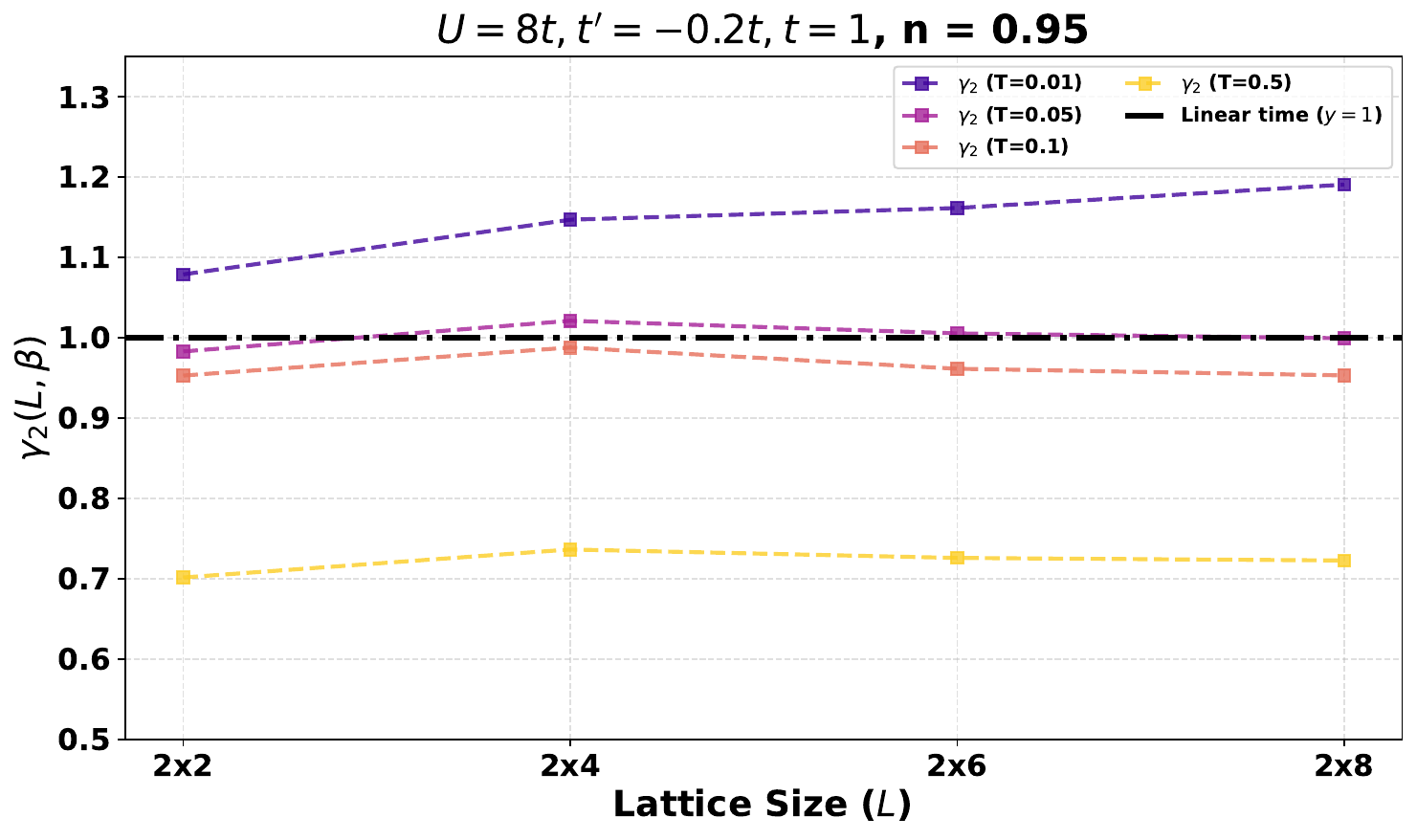}
    \end{subfigure}

    \begin{subfigure}[b]{0.45\textwidth}
    \centering
    \includegraphics[width=\textwidth]{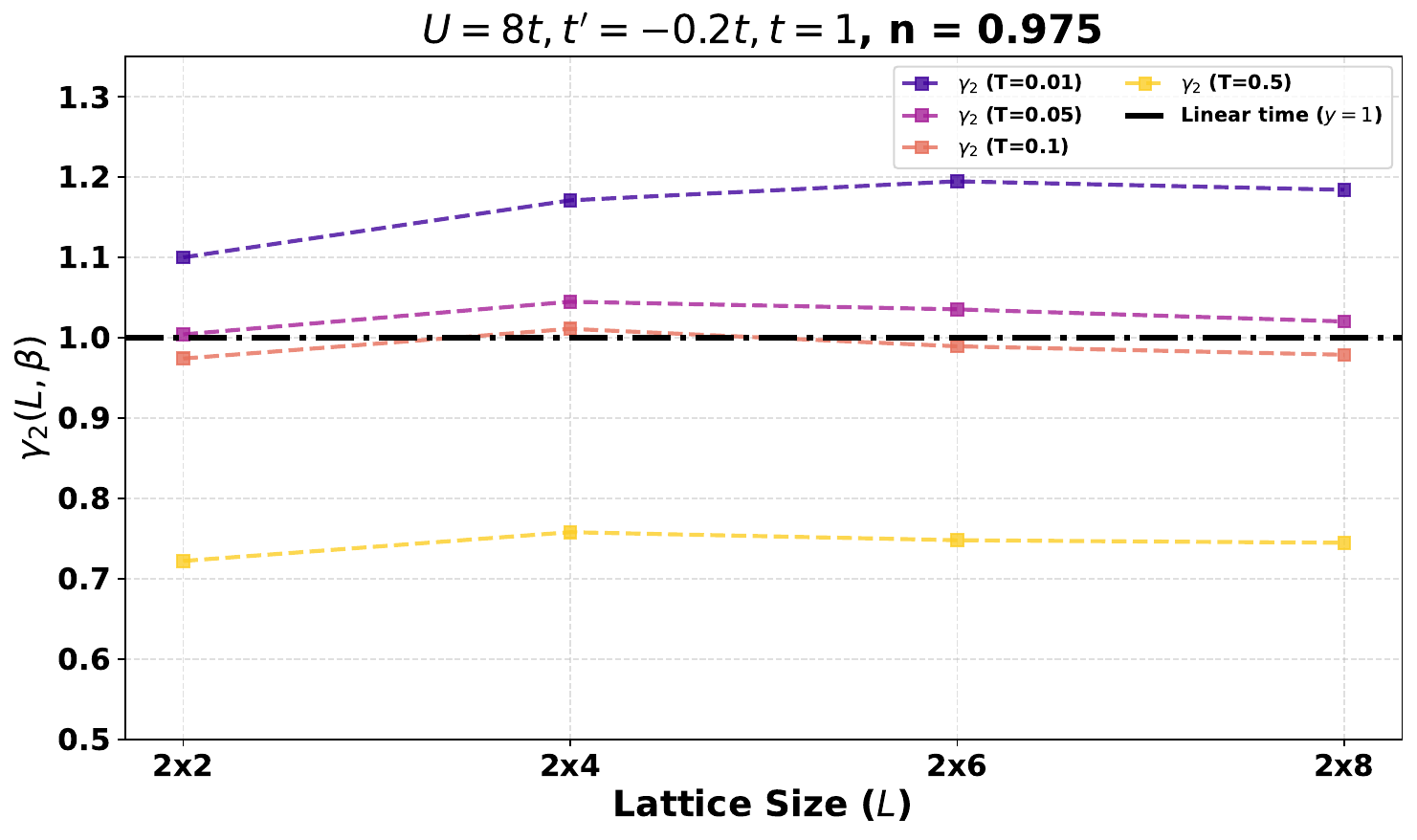}
    \end{subfigure}
    \begin{subfigure}[b]{0.45\textwidth}
    \centering
    \includegraphics[width=\textwidth]{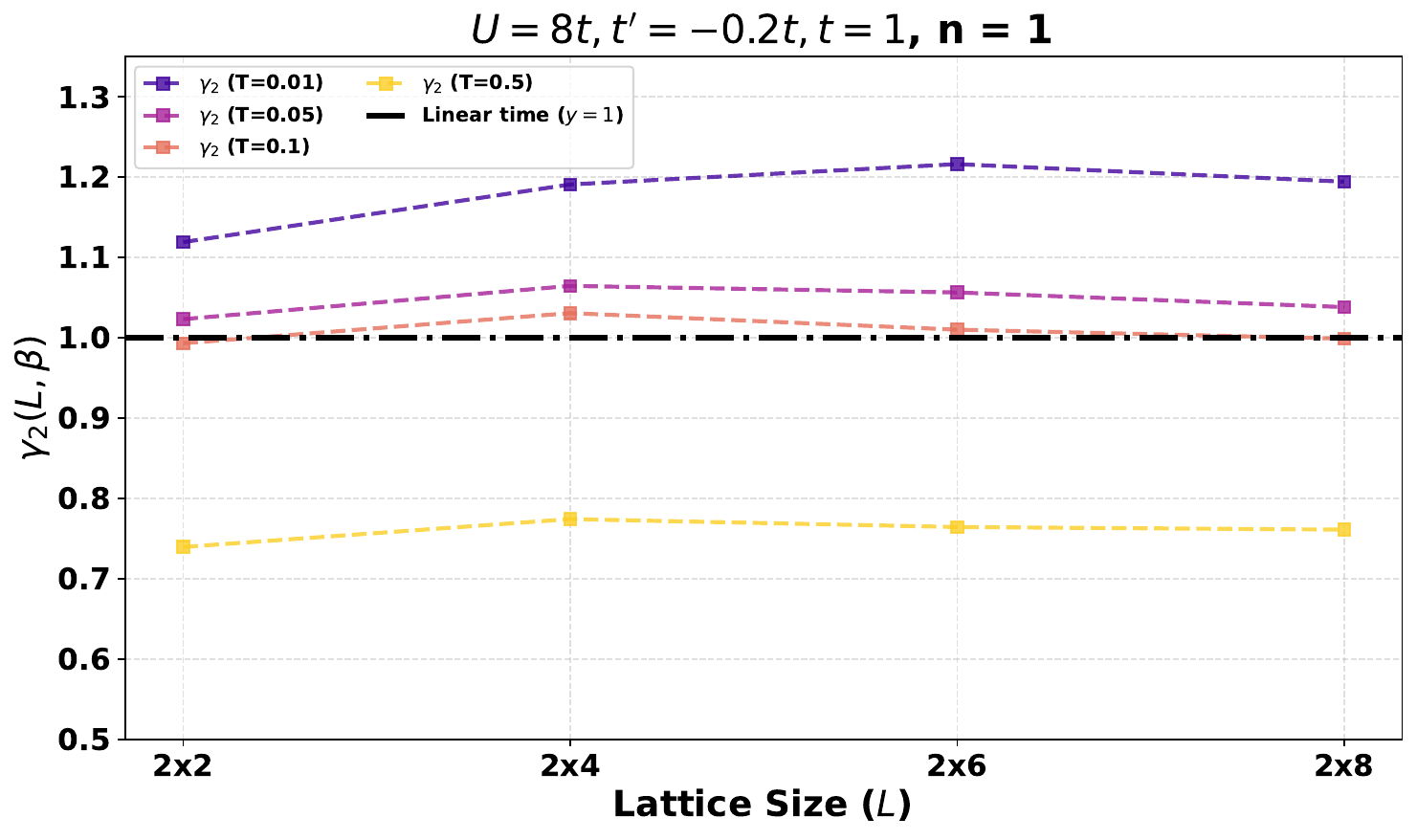}
    \end{subfigure}

\end{center}

    \caption{Simulation results evaluating the figure of merit $\gamma_2(L,\beta)$ for the conventional grand-canonical approach (\cref{alg: grand cano prev}) in the intermediate-coupling regime ($U=8$). The figure of merit $\gamma_2$ is plotted as a function of system size ($L$) across different temperatures ($T \in \{0.01, 0.05, 0.1, 0.5\}$). Results are obtained for electron densities $n=0.875, 0.9, 0.925, 0.95, 0.975$, and $1.0$, respectively.}
    \label{fig: g1U8}
\end{figure}

\section*{Code Availability}
Code used for the current study is available at the following GitHub repository: 

https://github.com/TKmath/Particle-Number-Preserving-Quantum-Signal-Processing-for-Thermal-Properties

\section*{Data Availability}
All data that support the findings of this study are included within the Github repository in the Code Availability section.

\section*{Acknowledgements}

This research was supported by the Korea Institute of Science and Technology Information(KISTI) (P25012). We used resources of the Center for Advanced Computation at Korea Institute for Advanced Study and the
National Energy Research Scientific Computing Center
(NERSC), a U.S. Department of Energy Office of Science User Facility operated under Contract No. DE-AC02-05CH11231. SC was supported by a KIAS Individual Grant (No. CG090601) at Korea Institute for
Advanced Study and by Quantum Simulator Development Project for Materials Innovation through the National Research Foundation of Korea (NRF) funded by
the Korean government (Ministry of Science and ICT
(MSIT)) (No. NRF-2023M3K5A1094813). TK was sup-
ported by a KIAS Individual Grant (No. CG096001) and
by Institute of Information \& Communications Technology Planning \& Evaluation (IITP) grant funded by the
Korea government (MSIT) (No. 2022-0-01026).

\section*{Declarations}

\begin{center}
    \textbf{Author Contributions}
\end{center}TK and SC conceptualized the idea and wrote the main manuscript. TK, SL, and SC performed the theoretical and numerical analyses, and contributed to finalizing the manuscript.

\begin{center}
    \textbf{AI assistance}
\end{center}
The authors acknowledge the use of large language models (LLMs) to assist with language and notation refinement and code development for numerical simulations. The authors chose to employ the conventional Lanczos algorithm combined with stochastic trace estimation as the core methodology for the numerical simulations. For memory-intensive calculations (e.g., the $2\times 8$ lattice geometry), AI tools assisted in coding parallel computing and matrix-free matrix-vector multiplication by exploiting the tensor-product structure of the Hamiltonian. All AI-assisted code was reviewed and modified by the authors, who take full responsibility for the final implementations and numerical results.

\begin{center}
    \textbf{Competing Interests}
\end{center}
The authors declare no competing financial or non-financial interests.

\bibliographystyle{unsrt}
\bibliography{ref}

@article{ko2025fermionic,
  title={Fermionic-Adapted Shadow Tomography for dynamical correlation functions},
  author={Ko, Taehee and Han, Mancheon and Park, Hyowon and Choi, Sangkook},
  journal={arXiv preprint arXiv:2508.03192},
  year={2025}
}

@inproceedings{king2025triply,
  title={Triply efficient shadow tomography},
  author={King, Robbie and Gosset, David and Kothari, Robin and Babbush, Ryan},
  booktitle={Proceedings of the 2025 Annual ACM-SIAM Symposium on Discrete Algorithms (SODA)},
  pages={914--946},
  year={2025},
  organization={SIAM}
}

@article{watson2025exponentially,
  title={Exponentially reduced circuit depths using trotter error mitigation},
  author={Watson, James D and Watkins, Jacob},
  journal={PRX quantum},
  volume={6},
  number={3},
  pages={030325},
  year={2025},
  publisher={APS}
}

@article{mazziotti2011large,
  title={Large-scale semidefinite programming for many-electron quantum mechanics},
  author={Mazziotti, David A},
  journal={Physical review letters},
  volume={106},
  number={8},
  pages={083001},
  year={2011},
  publisher={APS}
}

@article{baumgratz2012lower,
  title={Lower bounds for ground states of condensed matter systems},
  author={Baumgratz, Tillmann and Plenio, Martin B},
  journal={New Journal of Physics},
  volume={14},
  number={2},
  pages={023027},
  year={2012},
  publisher={IOP Publishing}
}

@article{goh2026direct,
  title={Direct estimation of the density of states for fermionic systems},
  author={Goh, Matthew L and Koczor, B{\'a}lint},
  journal={Physical Review Research},
  volume={8},
  number={1},
  pages={013250},
  year={2026},
  publisher={APS}
}

@inproceedings{ko2024ground,
  title={Ground-state energy and related properties estimation in quantum chemistry with linear dependence on the number of atoms},
  author={Ko, Taehee and Li, Xiantao and Wang, Chunhao},
  booktitle={2024 IEEE International Conference on Quantum Computing and Engineering (QCE)},
  volume={1},
  pages={669--680},
  year={2024},
  organization={IEEE}
}

@article{PhysRevB.49.5065,
  title = {Lanczos method for the calculation of finite-temperature quantities in correlated systems},
  author = {Jakli\ifmmode \check{c}\else \v{c}\fi{}, J. and Prelov\ifmmode \check{s}\else \v{s}\fi{}ek, P.},
  journal = {Phys. Rev. B},
  volume = {49},
  issue = {7},
  pages = {5065(R)--5068(R)},
  numpages = {0},
  year = {1994},
  month = {Feb},
  publisher = {American Physical Society},
  doi = {10.1103/PhysRevB.49.5065},
  url = {https://link.aps.org/doi/10.1103/PhysRevB.49.5065}
}

@article{qprk-k7wn,
  title = {High-Temperature Fermionic Gibbs States Are Mixtures of Gaussian States},
  author = {Ramkumar, Akshar and Cai, Yiyi and Tong, Yu and Jiang, Jiaqing},
  journal = {Phys. Rev. Lett.},
  volume = {137},
  issue = {2},
  pages = {020601},
  numpages = {7},
  year = {2026},
  month = {Jul},
  publisher = {American Physical Society},
  doi = {10.1103/qprk-k7wn},
  url = {https://link.aps.org/doi/10.1103/qprk-k7wn}
}

@article{epperly2024xtrace,
  title={Xtrace: Making the most of every sample in stochastic trace estimation},
  author={Epperly, Ethan N and Tropp, Joel A and Webber, Robert J},
  journal={SIAM Journal on Matrix Analysis and Applications},
  volume={45},
  number={1},
  pages={1--23},
  year={2024},
  publisher={SIAM}
}

@article{chakraborty2025quantum,
  title={Quantum singular value transformation without block encodings: Near-optimal complexity with minimal ancilla},
  author={Chakraborty, Shantanav and Hazra, Soumyabrata and Li, Tongyang and Shao, Changpeng and Wang, Xinzhao and Zhang, Yuxin},
  journal={arXiv preprint arXiv:2504.02385},
  year={2025}
}

@article{kasper2017implementing,
  title={Implementing quantum electrodynamics with ultracold atomic systems},
  author={Kasper, Valentin and Hebenstreit, Florian and Jendrzejewski, Fred and Oberthaler, Markus K and Berges, J{\"u}rgen},
  journal={New journal of physics},
  volume={19},
  number={2},
  pages={023030},
  year={2017},
  publisher={IOP Publishing}
}

@article{cao2017unified,
  title={A unified ab initio approach to the correlated quantum dynamics of ultracold fermionic and bosonic mixtures},
  author={Cao, Lushuai and Bolsinger, V and Mistakidis, Simeon I and Koutentakis, GM and Kr{\"o}nke, Sven and Schurer, JM and Schmelcher, Peter},
  journal={The Journal of chemical physics},
  volume={147},
  number={4},
  year={2017},
  publisher={AIP Publishing}
}

@article{chen2023quantum,
  title={Quantum thermal state preparation},
  author={Chen, Chi-Fang and Kastoryano, Michael J and Brand{\~a}o, Fernando GSL and Gily{\'e}n, Andr{\'a}s},
  journal={arXiv preprint arXiv:2303.18224},
  year={2023}
}

@article{tong2025fast,
  title={Fast mixing of weakly interacting fermionic systems at any temperature},
  author={Tong, Yu and Zhan, Yongtao},
  journal={PRX Quantum},
  volume={6},
  number={3},
  pages={030301},
  year={2025},
  publisher={APS}
}

@article{nocera2016symmetry,
  title={Symmetry-conserving purification of quantum states within the density matrix renormalization group},
  author={Nocera, Alberto and Alvarez, Gonzalo},
  journal={Physical Review B},
  volume={93},
  number={4},
  pages={045137},
  year={2016},
  publisher={APS}
}

@article{barthel2016matrix,
  title={Matrix product purifications for canonical ensembles and quantum number distributions},
  author={Barthel, Thomas},
  journal={Physical Review B},
  volume={94},
  number={11},
  pages={115157},
  year={2016},
  publisher={APS}
}

@article{prelovvsek2013ground,
  title={Ground state and finite temperature Lanczos methods},
  author={Prelov{\v{s}}ek, P and Bon{\v{c}}a, J},
  journal={Strongly correlated systems: Numerical methods},
  pages={1--30},
  year={2013},
  publisher={Springer}
}

@article{petras2020using,
  title={Using density matrix quantum Monte Carlo for calculating exact-on-average energies for ab initio Hamiltonians in a finite basis set},
  author={Petras, Hayley R and Ramadugu, Sai Kumar and Malone, Fionn D and Shepherd, James J},
  journal={Journal of chemical theory and computation},
  volume={16},
  number={2},
  pages={1029--1038},
  year={2020},
  publisher={ACS Publications}
}

@article{xie2021twisted,
  title={Twisted bilayer graphene. VI. An exact diagonalization study at nonzero integer filling},
  author={Xie, Fang and Cowsik, Aditya and Song, Zhi-Da and Lian, Biao and Bernevig, B Andrei and Regnault, Nicolas},
  journal={Physical Review B},
  volume={103},
  number={20},
  pages={205416},
  year={2021},
  publisher={APS}
}

@book{nielsen2010quantum,
  title={Quantum computation and quantum information},
  author={Nielsen, Michael A and Chuang, Isaac L},
  year={2010},
  publisher={Cambridge university press}
}

@article{dong2022ground,
  title={Ground-state preparation and energy estimation on early fault-tolerant quantum computers via quantum eigenvalue transformation of unitary matrices},
  author={Dong, Yulong and Lin, Lin and Tong, Yu},
  journal={PRX quantum},
  volume={3},
  number={4},
  pages={040305},
  year={2022},
  publisher={APS}
}

@article{shen2023stable,
  title={Stable recursive auxiliary field quantum Monte Carlo algorithm in the canonical ensemble: Applications to thermometry and the Hubbard model},
  author={Shen, Tong and Barghathi, Hatem and Yu, Jiangyong and Del Maestro, Adrian and Rubenstein, Brenda M},
  journal={Physical Review E},
  volume={107},
  number={5},
  pages={055302},
  year={2023},
  publisher={APS}
}

@article{sokolov2021microcanonical,
  title={Microcanonical and finite-temperature ab initio molecular dynamics simulations on quantum computers},
  author={Sokolov, Igor O and Barkoutsos, Panagiotis Kl and Moeller, Lukas and Suchsland, Philippe and Mazzola, Guglielmo and Tavernelli, Ivano},
  journal={Physical Review Research},
  volume={3},
  number={1},
  pages={013125},
  year={2021},
  publisher={APS}
}

@article{harsha2020wave,
  title={Wave function methods for canonical ensemble thermal averages in correlated many-fermion systems},
  author={Harsha, Gaurav and Henderson, Thomas M and Scuseria, Gustavo E},
  journal={The Journal of Chemical Physics},
  volume={153},
  number={12},
  year={2020},
  publisher={AIP Publishing}
}

@article{li2025dissipative,
  title={Dissipative ground state preparation in ab initio electronic structure theory},
  author={Li, Hao-En and Zhan, Yongtao and Lin, Lin},
  journal={npj Quantum Information},
  volume={11},
  number={1},
  pages={183},
  year={2025},
  publisher={Nature Publishing Group UK London}
}

@article{coopmans2023predicting,
  title={Predicting Gibbs-state expectation values with pure thermal shadows},
  author={Coopmans, Luuk and Kikuchi, Yuta and Benedetti, Marcello},
  journal={PRX Quantum},
  volume={4},
  number={1},
  pages={010305},
  year={2023},
  publisher={APS}
}

@article{shen2020finite,
  title={Finite temperature auxiliary field quantum Monte Carlo in the canonical ensemble},
  author={Shen, Tong and Liu, Yuan and Yu, Yang and Rubenstein, Brenda M},
  journal={The Journal of Chemical Physics},
  volume={153},
  number={20},
  year={2020},
  publisher={AIP Publishing}
}

@inproceedings{gilyen2019quantum,
  title={Quantum singular value transformation and beyond: exponential improvements for quantum matrix arithmetics},
  author={Gily{\'e}n, Andr{\'a}s and Su, Yuan and Low, Guang Hao and Wiebe, Nathan},
  booktitle={Proceedings of the 51st annual ACM SIGACT symposium on theory of computing},
  pages={193--204},
  year={2019}
}

@article{liu2018ab,
  title={Ab initio finite temperature auxiliary field quantum Monte Carlo},
  author={Liu, Yuan and Cho, Minsik and Rubenstein, Brenda},
  journal={Journal of Chemical Theory and Computation},
  volume={14},
  number={9},
  pages={4722--4732},
  year={2018},
  publisher={ACS Publications}
}

@article{wang2024qubit,
  title={Qubit-efficient randomized quantum algorithms for linear algebra},
  author={Wang, Samson and McArdle, Sam and Berta, Mario},
  journal={PRX quantum},
  volume={5},
  number={2},
  pages={020324},
  year={2024},
  publisher={APS}
}

@article{liu2025block,
  title={Block encoding with low gate count for second-quantized hamiltonians},
  author={Liu, Diyi and Zhu, Shuchen and Low, Guang Hao and Lin, Lin and Yang, Chao},
  journal={arXiv preprint arXiv:2510.08644},
  year={2025}
}

@article{chen2025convergence,
  title={Convergence of the Cumulant Expansion and Polynomial-Time Algorithm for Weakly Interacting Fermions},
  author={Chen, Hongrui and Rouz{\'e}, Cambyse and Chen, Jielun and Jiang, Jiaqing and Scalet, Samuel O and Zhan, Yongtao and Chan, Garnet Kin and Ying, Lexing and Tong, Yu},
  journal={arXiv preprint arXiv:2512.12010},
  year={2025}
}

@article{de2024partition,
  title={Partition function estimation with a quantum coin toss},
  author={de Lima Silva, Thais and Borges, Lucas and Aolita, Leandro},
  journal={arXiv e-prints},
  pages={arXiv--2411},
  year={2024}
}

@article{aichhorn2003low,
  title={Low-temperature Lanczos method for strongly correlated systems},
  author={Aichhorn, Markus and Daghofer, Maria and Evertz, Hans Gerd and Von der Linden, Wolfgang},
  journal={Physical Review B},
  volume={67},
  number={16},
  pages={161103},
  year={2003},
  publisher={APS}
}

@article{leblanc2015solutions,
  title={Solutions of the two-dimensional Hubbard model: Benchmarks and results from a wide range of numerical algorithms},
  author={LeBlanc, James PF and Antipov, Andrey E and Becca, Federico and Bulik, Ireneusz W and Chan, Garnet Kin-Lic and Chung, Chia-Min and Deng, Youjin and Ferrero, Michel and Henderson, Thomas M and Jim{\'e}nez-Hoyos, Carlos A and others},
  journal={Physical Review X},
  volume={5},
  number={4},
  pages={041041},
  year={2015},
  publisher={APS}
}

@article{chen2023krylov,
  title={Krylov-aware stochastic trace estimation},
  author={Chen, Tyler and Hallman, Eric},
  journal={SIAM Journal on Matrix Analysis and Applications},
  volume={44},
  number={3},
  pages={1218--1244},
  year={2023},
  publisher={SIAM}
}

@article{avron2011randomized,
  title={Randomized algorithms for estimating the trace of an implicit symmetric positive semi-definite matrix},
  author={Avron, Haim and Toledo, Sivan},
  journal={Journal of the ACM (JACM)},
  volume={58},
  number={2},
  pages={1--34},
  year={2011},
  publisher={ACM New York, NY, USA}
}

@article{romero2019strategies,
  title={Strategies for quantum computing molecular energies using the unitary coupled cluster ansatz},
  author={Romero, Jonathan and Babbush, Ryan and McClean, Jarrod R and Hempel, Cornelius and Love, Peter J and Aspuru-Guzik, Al{\'a}n},
  journal={Quantum Science and Technology},
  volume={4},
  number={1},
  pages={014008},
  year={2019},
  publisher={IOP Publishing}
}

@article{rouze2026optimal,
  title={Optimal quantum algorithm for Gibbs state preparation},
  author={Rouz{\'e}, Cambyse and Stilck Fran{\c{c}}a, Daniel and Alhambra, {\'A}lvaro M},
  journal={Physical Review Letters},
  volume={136},
  number={6},
  pages={060601},
  year={2026},
  publisher={APS}
}

@article{vsmid2025polynomial,
  title={Polynomial-time quantum Gibbs sampling for the weak and strong coupling regime of the Fermi-Hubbard model at any temperature},
  author={{\v{S}}m{\'\i}d, {\v{S}}t{\v{e}}p{\'a}n and Meister, Richard and Berta, Mario and Bondesan, Roberto},
  journal={Nature Communications},
  volume={16},
  number={1},
  pages={10736},
  year={2025},
  publisher={Nature Publishing Group UK London}
}

@article{chowdhury2021computing,
  title={Computing partition functions in the one-clean-qubit model},
  author={Chowdhury, Anirban N and Somma, Rolando D and Suba{\c{s}}{\i}, Yi{\u{g}}it},
  journal={Physical Review A},
  volume={103},
  number={3},
  pages={032422},
  year={2021},
  publisher={APS}
}

@article{bravyi2022quantum,
  title={Quantum Hamiltonian complexity in thermal equilibrium},
  author={Bravyi, Sergey and Chowdhury, Anirban and Gosset, David and Wocjan, Pawel},
  journal={Nature Physics},
  volume={18},
  number={11},
  pages={1367--1370},
  year={2022},
  publisher={Nature Publishing Group UK London}
}

@inproceedings{bartschi2019deterministic,
  title={Deterministic preparation of Dicke states},
  author={B{\"a}rtschi, Andreas and Eidenbenz, Stephan},
  booktitle={International Symposium on Fundamentals of Computation Theory},
  pages={126--139},
  year={2019},
  organization={Springer}
}

@article{motlagh2024generalized,
  title={Generalized quantum signal processing},
  author={Motlagh, Danial and Wiebe, Nathan},
  journal={PRX Quantum},
  volume={5},
  number={2},
  pages={020368},
  year={2024},
  publisher={APS}
}

@article{tosta2024randomized,
  title={Randomized semi-quantum matrix processing},
  author={Tosta, Allan and de Lima Silva, Thais and Camilo, Giancarlo and Aolita, Leandro},
  journal={npj Quantum Information},
  volume={10},
  number={1},
  pages={93},
  year={2024},
  publisher={Nature Publishing Group UK London}
}

@inproceedings{van2017quantum,
  title={Quantum SDP-solvers: Better upper and lower bounds},
  author={Van Apeldoorn, Joran and Gily{\'e}n, Andr{\'a}s and Gribling, Sander and de Wolf, Ronald},
  booktitle={2017 IEEE 58th Annual Symposium on Foundations of Computer Science (FOCS)},
  pages={403--414},
  year={2017},
  organization={IEEE}
}

\end{document}